\documentclass[autheryear, times, 9pt]{elsarticle}

\usepackage{geometry}
\usepackage{self}
\usepackage{amsthm}
\usepackage{booktabs}
\usepackage{times}
\usepackage{url}
\usepackage{float}
\usepackage{algorithm,algorithmic}
\usepackage{bm}
\usepackage{mathrsfs}

\usepackage{amsmath}
\usepackage{longtable}
\usepackage{amsfonts}
\usepackage{amssymb}

\usepackage{graphicx}
\usepackage{epstopdf}

\usepackage{multirow}
\usepackage{subcaption}
\usepackage{verbatim}
\usepackage{lineno}
\usepackage{enumitem}
\usepackage{color}
\usepackage{subcaption}
\usepackage{epsfig}
\usepackage{apxproof}
\usepackage{diagbox}

\usepackage{lineno}
\modulolinenumbers[5]
\usepackage[dvipsnames,svgnames,x11names]{xcolor}
\usepackage[commandnameprefix=ifneeded, markup=underlined]{changes}

\usepackage{hyperref}
\usepackage{xcolor}
\hypersetup{
    colorlinks,
    linkcolor={red!50!black},
    citecolor={blue!50!black},
    urlcolor={blue!80!black}
}

\newtheorem{lemma}{Lemma}
\newtheorem{proposition}{Proposition}
\newtheorem{definition}{Definition}
\newtheorem{observation}{Observation}

\usepackage{natbib}

\hypersetup{
    colorlinks,
    linkcolor={red!50!black},
    citecolor={blue!50!black},
    urlcolor={blue!80!black}
}

\makeatletter

\newcommand{\Rmnum}[1]{\expandafter\@slowromancap\romannumeral #1@}
\newenvironment{breakablealgorithm}
  {
   \begin{center}
     \refstepcounter{algorithm}
     \hrule height.8pt depth0pt \kern2pt
     \renewcommand{\caption}[2][\relax]{
       {\raggedright\textbf{\ALG@name~\thealgorithm} ##2\par}%
       \ifx\relax##1\relax 
         \addcontentsline{loa}{algorithm}{\protect\numberline{\thealgorithm}##2}%
       \else 
         \addcontentsline{loa}{algorithm}{\protect\numberline{\thealgorithm}##1}%
       \fi
       \kern2pt\hrule\kern2pt
     }
  }{
     \kern2pt\hrule\relax
   \end{center}
  }
\makeatother

\biboptions{authoryear}

\graphicspath{{img/}}
\journal{Elsevier}

\begin{document}

\begin{frontmatter}

\title{Maximin Headway Control of Automated Vehicles for System Optimal Dynamic Traffic Assignment in General Networks}
\author[mymainaddress]{Jinxiao Du}
\ead{jinxiao.du@connect.polyu.hk}
\author[mymainaddress,mymainaddress2]{Wei Ma\corref{mycorrespondingauthor}}
\cortext[mycorrespondingauthor]{Corresponding author}
\ead{wei.w.ma@polyu.edu.hk}	
\address[mymainaddress]{Department of Civil and Environmental Engineering, The Hong Kong Polytechnic University, Hong Kong SAR, China}
\address[mymainaddress2]{The Hong Kong Polytechnic University Shenzhen Research Institute, Shenzhen, Guangdong, China}

\begin{abstract}
This study develops the headway control framework in a fully automated road network, as we believe headway of Automated Vehicles (AVs) is another influencing factor to traffic dynamics in addition to conventional vehicle behaviors ({\em e.g.} route and departure time choices).
Specifically, we aim to search for the optimal time headway between AVs on each link that achieves the network-wide system optimal dynamic traffic assignment (SO-DTA). 
To this end, the headway-dependent fundamental diagram (HFD) and headway-dependent double queue model (HDQ) are developed to model the effect of dynamic headway on roads, and a dynamic network model is built. It is rigorously proved that the minimum headway could always achieve SO-DTA, yet the optimal headway is non-unique. Motivated by these two findings, this study defines a novel concept of maximin headway, which is the largest headway that still achieves SO-DTA in the network. Mathematical properties regarding maximin headway are analyzed and an efficient solution algorithm is developed. 
Numerical experiments on both a small and large network verify the effectiveness of the maximin headway control framework as well as the properties of maximin headway.
This study sheds light on deriving the desired solution among the non-unique solutions in SO-DTA and provides implications regarding the safety margin of AVs under SO-DTA.
\end{abstract}

  
\end{frontmatter}


\section{Introduction}

With the rapid development of advanced sensors, wireless communication, and artificial intelligence, it is envisioned that automated vehicles (AVs) will play a significant role in future intelligent transportation systems. Ultimately, the fully automated environment would be realized to significantly mitigate congestion and ensure safety in general urban networks \citep{wadud2016help}. Therefore, it is essential for public agencies ({\em e.g.}, Transport Department) and private sectors ({\em e.g.}, Transportation Network Companies) to understand the network-wide effects of AVs in the fully automated environment and to develop effective management and control schemes for AVs. However, recent studies mainly focus on the microscopic behaviors of AVs, while studies on the spatio-temporal effects of AVs in networked traffic dynamics are still lacking \citep{YU2021103101}.

Among various traffic network models, system optimal dynamic traffic assignment (SO-DTA) quantifies the best possible performance of a traffic network given the dynamic network demand and supply. 
Both the optimal objective function ({\em i.e.}, the total travel time) and the optimal solution ({\em i.e.}, path flow and link flow) could benchmark different traffic operation and management strategies and provide policy implications for decision-makers.
Additionally, it is practical to control all the AVs to achieve SO-DTA under fully automated environments \citep{nguyen2021system}.

In this paper, we consider the headway control of AVs, in addition to the travelers' behaviors ({\em e.g.}, route and departure time choices) for SO-DTA. That is, AVs can decide the headway between each other according to the mobility and safety requirements of each road segment \citep{hatipoglu1996longitudinal}. In particular, headway refers to the time headway throughout the paper.
Existing studies have shown that different headway settings of AVs could affect the fundamental diagrams of each road \citep{ZHOU2020102614, GONG2016314}, hence affecting the solutions of SO-DTA. However, most current studies are at a microscopic level, while how the AV headway control could affect the network-wide SO-DTA is still an open question. 

Many studies have explored the applicability and practicality of AV headway control in future traffic scenarios \citep{xiao2023adaptive,elmorshedy2023freeway}. The concept and control principles of headway control are first proposed as an effective longitudinal control strategy of vehicles \citep{chiu1977vehicle,shladover1991automated}. The development of automated technology has attracted interest in studying the application of headway control in fully automated vehicles. \citet{li2017dynamical} and \citet{xiao2011practical} propose a practical constant time headway (CTH) control policy for AVs. Further, \citet{middleton2010string} considers the control strategy under a heterogeneous headway. Moreover, a variable time headway (VTH) control method was developed and proved as a more flexible, effective, and practical method compared with the CTH control policy \citep{chen2020connected}. A brief illustration of the headway control around an intersection is shown in Figure \ref{fig:headway}. All the AVs are connected through Vehicle-to-Vehicle (V2V) communications, and their headway can be set differently on different roads. The cooperative area is at the intersection \citep{de2015autonomous}, where AVs can communicate with the Roadside Unit (RSU) to receive the headway setting for each road and adjust their headway accordingly. The cooperative area should be designed to have strong communication capacities between AVs and RSUs and enough space for the adjustment of headway.


\begin{figure}[h]
    \centering
    \includegraphics[width=0.6\linewidth]{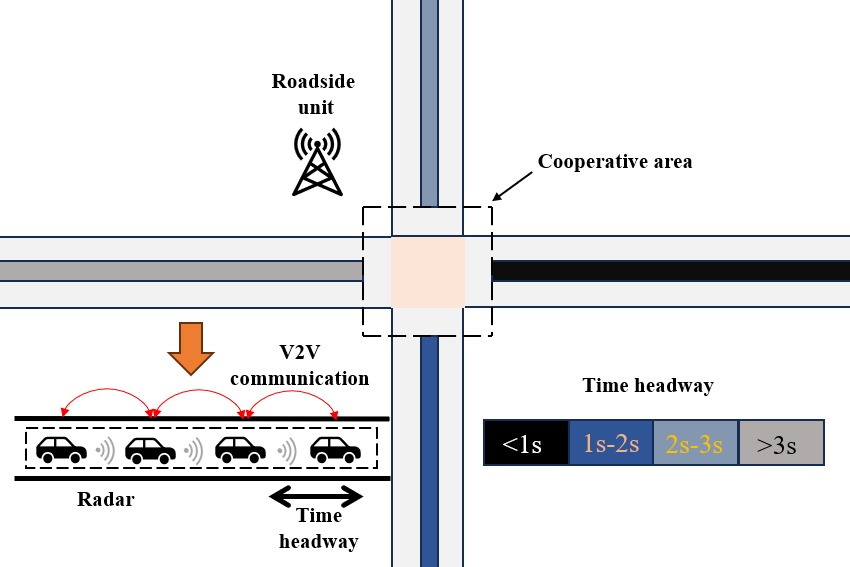}
    \caption{Illustration of headway control around an intersection.}
    \label{fig:headway}
\end{figure}

Intuitively, a smaller AV headway yields higher road capacities and throughput, leading to shorter travel time, and hence SO-DTA can be achieved with the smallest AV headway. To verify this intuition, later this paper {\em analytically} proves that Observation~\ref{ob:1} holds in general networks.

\begin{observation}
\label{ob:1}
Given a range of the AV headway, SO-DTA is always achieved when the AV headway is set as the smallest value in the range.
\end{observation}

Note that the range of AV headway depends on the minimum safety requirements of AVs, specifications of AVs, and the maximum capacity of the specific road \citep{li2017vehicle}.
Furthermore, similar to the existing SO-DTA models \citep{SHEN20141}, we have Observation~\ref{ob:nonunique} hold, indicating that SO-DTA may be achieved under different headway settings.

\begin{observation}
\label{ob:nonunique}
There can be multiple AV headway settings yielding the same SO-DTA solution in general networks.
\end{observation}







Observation~\ref{ob:nonunique} inspires us to search for the second criterion in the SO-DTA problem. For AVs, safety is always an indispensable factor.  The impact of headway on traffic safety is widely studied in the literature \citep{fenton1979headway,ayres2001preferred}. For example, 
\citet{vogel2003comparison} demonstrates that headway is a suitable indicator of safety and suggests using headway to measure safety for enforcement purposes given that the small headway has the potential to create dangerous driving situations. Moreover, \citet{biswas2022drivers} states that the small headway corresponds to a high risk of crash due to the inadequate response time, and in contrast, the large headway allows drivers to adjust their vehicles to decrease the risk of crash. Therefore, we present Figure \ref{fig:probability} to briefly illustrate a general relationship between the probability of accident and headway. The headway on each road should be between the minimum and maximum headway, and within the range, a larger headway could yield a lower probability of an accident, thus leading to safer traffic conditions.

\begin{figure}[h]
    \centering
    \includegraphics[width=0.5\linewidth]{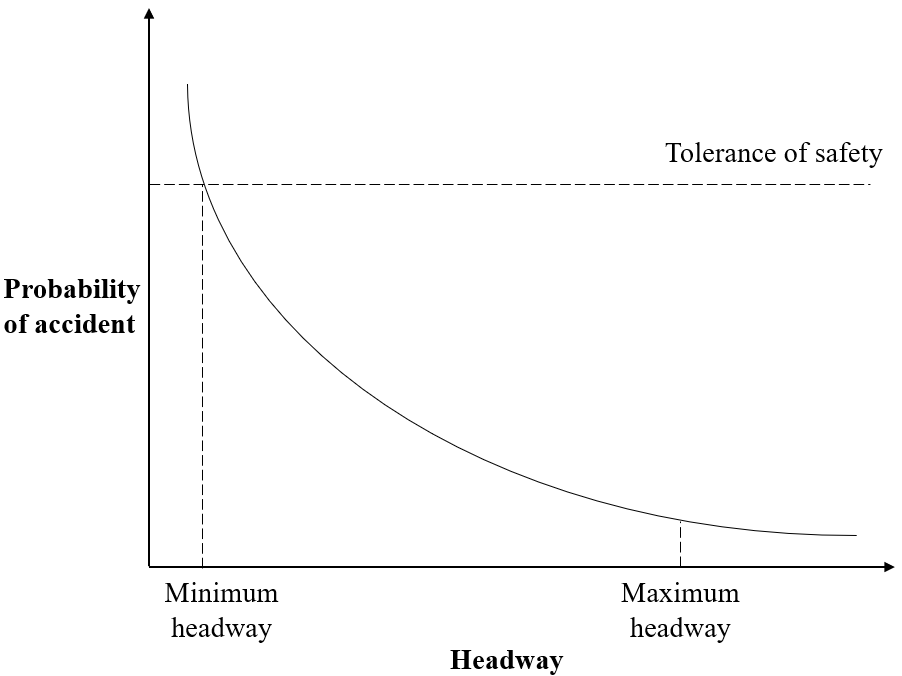}
    \caption{Illustrative relationship between the probability of accident and headway.}
    \label{fig:probability}
\end{figure}

In view of this, this paper proposes a novel concept of the maximin headway of AVs for SO-DTA, as defined in Definition~\ref{def:maximin}.

\begin{definition}[Maximin Headway - Layman's version]
\label{def:maximin}
The largest headway setting for AVs that still achieves the SO-DTA is referred to as the maximin headway of AVs. 
\end{definition}
One can see that the maximin headway is the ``largest'' headway setting that achieves SO-DTA, meaning that such a headway setting is the ``safest'' under the SO-DTA. An illustration of the maximin headway is also presented in Figure~\ref{fig:maximin}. Existing literature mainly focuses on the solutions to SO-DTA \citep{zhang2020path, qian2012system}, while this paper discusses the ``best'' solution among the non-unique solutions of SO-DTA for AVs, which we believe is interesting and unique.  In particular, differences between minimum headway and maximin headway provide policy implications for road segment design and AV deployment.

\begin{figure}[h]
    \centering
    \includegraphics[width=0.7\linewidth]{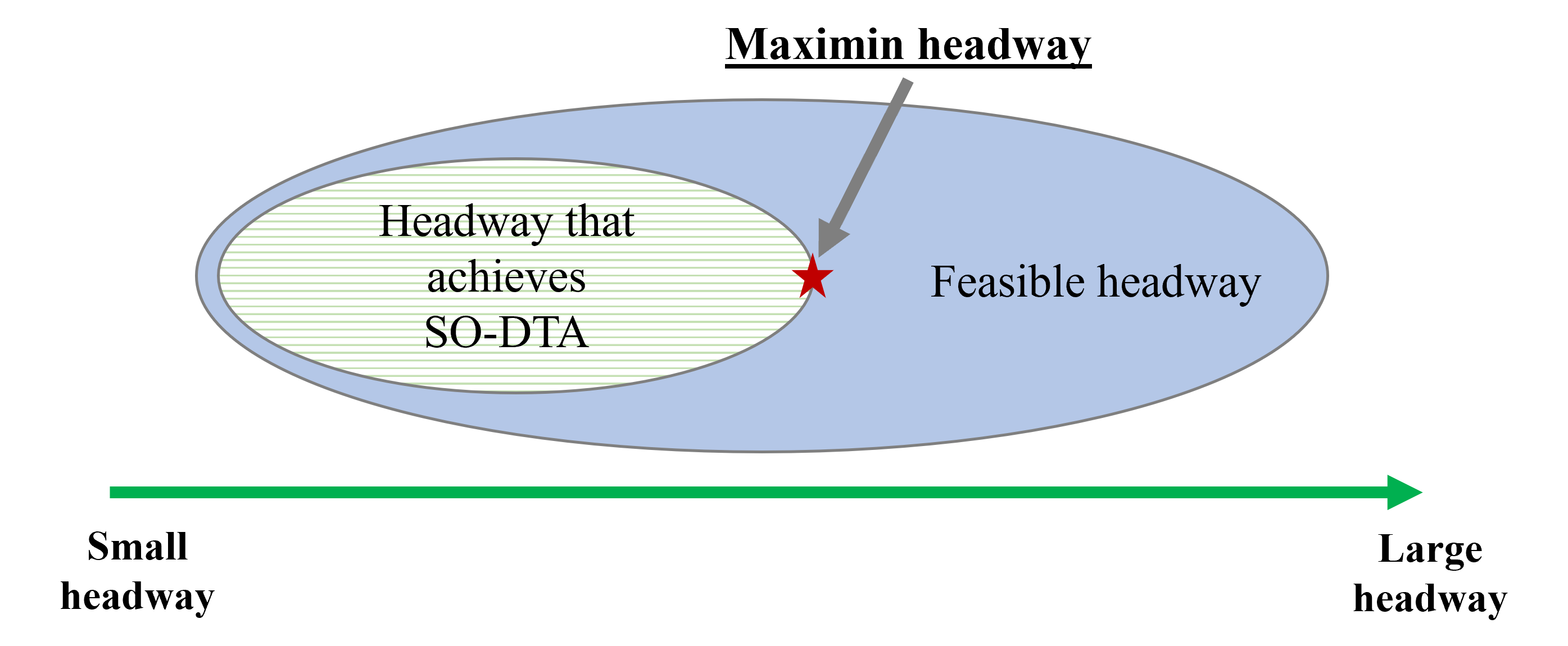}
    \caption{Illustration of the maximin headway for SO-DTA.}
    \label{fig:maximin}
\end{figure}

To summarize, this paper proposes a maximin headway control framework that achieves SO-DTA. Specifically,
we first explore how the headway of AVs would affect the road fundamental diagram and develop a headway-dependent double queue model (HDQ), which is proved to be a generalized form of the double queue model \citep{MA201498}.
Then we formulate the network-wide traffic dynamics with SO-DTA for AVs, followed by the rigorous proof that Observation~\ref{ob:1} and Observation \ref{ob:nonunique} hold. Then we formally define the maximin headway control for SO-DTA, propose a provable approach to solve it analytically, and show the uniqueness of the maximin headway. Lastly, numerical experiments are conducted on the small and large networks to verify our proofs.


The contributions of this paper are summarized as follows:
\begin{itemize}
    \item The headway-dependent double queue (HDQ) model is developed to depict the effects of AV headway on network-wide traffic dynamics.
    \item It is analytically proved that the SO-DTA can be achieved under minimum headway.
    \item It first time proposes and solves the maximin headway control of AVs for SO-DTA in general dynamic networks. 
\end{itemize}

The remainder of this paper is organized as follows: related literature is reviewed in Section~\ref{sec:lit}. Section \ref{sec:model} presents headway-dependent flow dynamics using HFD and HDQ, Section \ref{sec_system_optimal_headway_formulation} develops the formulation of optimal headway control, Section \ref{sec_Optimal_maximin_Headway_Formulation} proposes the formulation of maximin headway control and the corresponding solution algorithm, and Section \ref{sec:Numerical_Results} shows the numerical results in the small and large networks. Lastly, conclusions are drawn in section \ref{sec:Conclusion}.

\section{Literature review}
\label{sec:lit}
In this section, we summarize the literature regarding the traffic network models for AVs. In particular, we compare with existing studies on dynamic link models, network models, and SO-DTA, to further highlight the contributions of this paper.

\subsection{Dynamic link models}

Dynamic link models are an essential component that models the flow dynamics on road segments, and 
it is necessary to develop an appropriate link model to depict the effects of AV headway control. The Lighthill, Whitham, and Richards (LWR) model forms the basis of the dynamic network loading model (DNL) by capturing key congested patterns on road segments \citep{lighthill1955kinematic,richards1956shock}. The LWR model is a system of partial differential equations (PDE), which consists of three parts: static hydrodynamics relation for flow, density, and speed; fundamental diagram; and dynamic continuity equation. Many famous DNL models are designed to solve the PDE system of LWR in different methods.

The cell transmission model (CTM) provides an approximate solution to the PDE system of the LWR model through the discretization in time and space \citep{DAGANZO199579}. CTM is formulated based on the density in each cell and the flow rates over the boundaries of cells. While the CTM is capable of accurately capturing spillback phenomena and determining detailed traffic states within a link, it results in high computational costs in our problem due to the potentially larger dimension of a DTA problem when using CTM compared to other link models \citep{nie2005comparative}. Compared with CTM which approximately solves LWR by discretizing in space and time, the link transmission model (LTM) directly solves the LWR system using the Newell-Daganzo method \citep{LTM}. LTM only tracks the cumulative number of vehicles at the entrance and exit of a link with a smaller computation cost, and it could also capture the spillbacks by incorporating the receiving constraints for congested states and the sending constraints for free-flow states. Recently, \citet{osorio2011dynamic} first proposes the concept of double queue, where the dynamics of a link are captured by the downstream queue and upstream queue. The downstream queue models the delay in free-flow time, the queuing process, and the outflow at the downstream end of a link. On the other hand, the upstream queue captures the delay caused by backward shockwave and the congestion propagation process, as well as the inflow at the upstream end of a link. The continuous form of DQ is developed by \citet{MA201498} and it is proved that DQ could be directly derived by LTM but with the simplified mathematical representation compared with LTM. Besides, the spillback in DQ could be tracked by the traffic states at either the entrance or the exit of a link \citep{MA201498}. Furthermore, \cite{NGODUY202156} proposes a multi-class two regime transmission model (MTTM) to capture the evolution of the queue lengths, where a link is divided into two areas: free-flow area and congested area. But MTTM is a nonlinear model for traffic states and it would add complexity and difficulty in solving the traffic states and analyzing the properties in DTA problems. 

The abovementioned models are all capable of capturing the spillback phenomenon representing the time when the congestion shockwave reaches the entrance of a link. Besides, they all require a node model to satisfy the invariance principle \citep{lebacque2005first}. To conclude, all these models are also designed to solve the PDE system of the LWR model, which actually depicts the dynamics of traffic states. Technically, the choice of DNL models would not significantly influence the modeling of the dynamic networks, while we prefer a simpler model that has more analytical capabilities. Compared with other DNL models, DQ  has the advantages of low computation cost, simplified mathematical formulation, and linear representations for traffic state variables. Therefore, we choose the DQ model as the basis to incorporate the effects of headway in flow dynamics in this paper.

\subsection{DTA with AVs}

Dynamic traffic assignment (DTA) models are designed to explore the time-dependent traffic state pattern given a time-varying travel demand on general traffic networks \citep{szeto2006dynamic}. There are two representative types of DTA: dynamic user equilibrium (DUE) and system optimal dynamic traffic assignment (SO-DTA). SO-DTA quantifies the best possible performance of a traffic network based on the dynamic extension of Wardrop’s second principle \citep{wardrop1952road}. For its advantage of exploring the spatial-temporal traffic dynamics to achieve the best performance of traffic networks, the research of the application of SO-DTA problem focuses on many different areas, such as emission reduction \citep{lu2016eco,ma2017emission,long2018link,tan2021emission}, congestion mitigation \citep{levin2017congestion,samaranayake2018discrete,liu2020integrated}, parking services \citep {levin2019linear,qian2014optimal}, traffic signal control \citep{lo2001cell,han2016robust}, emergency evacuation \citep{liu2006cell,chiu2007modeling} and network design \citep{waller2001stochastic,waller2006linear}.

Different from human-driven vehicles (HVs), AVs show great potential in improving control efficiency and may play an important role in achieving SO-DTA \citep{wang2018dynamic}. 
The current research on traffic assignment with AVs is summarized in Table \ref{literature}. For the static traffic assignment problem of AVs, \cite{CHEN2016143} proposes a mathematical model to derive an optimal AV lane control framework for mixed traffic, where AV lines are only permitted for AVs and it is necessary to specify the deployment plan. \citet{Bagloee} focuses on a static traffic assignment problem on mixed traffic where AVs are seeking the system optimal (SO) and HVs are seeking for user equilibrium (UE), and considers the influence of many realistic features in the solution, such as road capacity and elastic demand. \cite{CHEN201744} presents an optimal AV zone design framework to improve the performance of the whole traffic network, where the government may make a plan of AVs zones only allowing AVs to enter to promote the adoption of AVs in the future. \cite{ZHANG201875} proposes a route control scheme to balance the system efficiency and the control intensity by controlling a relatively small portion of all vehicles. \cite{WANG2019139} addresses the multi-class traffic assignment problem by assuming AVs follow SO and HVs follow cross-nested logit (CNL) stochastic user equilibrium (SUE) to capture HV drivers' uncertainty due to their limited knowledge of traffic conditions in the whole network. \cite{WANG2020227} develops a mixed behavioral equilibrium model in mixed traffic considering different mode choices for HVs and AVs, 
and \cite{chen2020path} determines the minimum control ratio of AVs to achieve SO by path controlling. \citet{LEVIN2016103} proposes a multiclass cell transmission model for HVs and AVs and proves the consistency with the hydrodynamic theory. A collision avoidance car-following model incorporating drivers’ reaction time is developed to explore the effect of reaction time on link capacity and backward wave speed. Their works provide valuable insights into the wide applications of AVs in avoiding collision by a smaller reaction time compared with human drivers.

\begin{table}[H]
    \centering
    \resizebox{1.02\textwidth}{!}{
    \begin{tabular}{lllll}
        \toprule
        Reference      & Type  &  HV & AV & Description \\
        \midrule
        \cite{LEVIN2016103} &  Dynamic  & DUE & DUE 
        & A multi-class CTM model for mixed traffic \\
        \cite{CHEN2016143}  & Static & UE & UE   
        & An optimal AV lanes control framework for mixed traffic\\
        \cite{Bagloee} & Static & UE & SO 
        & The mixed traffic network with elastic demand \\
        \cite{CHEN201744} & Static & UE & SO 
        & Optimal design of AV zones for mixed traffic \\
        \cite{ZHANG201875} & Static & UE & SO 
        & A route control scheme to balance efficiency and control \\
        \cite{WANG2019139} & Static & CNL & UE 
        & A multi-class traffic assignment model for mixed traffic \\
        \cite{WANG2020227} & Static & SUE  & SO
        & A mixed equilibrium model with mode choice \\
        \cite{chen2020path} & Static & UE  & SO
        & Minimum control ratio of AVs for SO in mixed traffic\\
        \cite{NGODUY202156} & Dynamic & SO-DTA & SO-DTA 
        & Impact of AVs in traffic network by multi-class SO-DTA \\
        \cite{9408374} & Dynamic & \textbackslash & SO-DTA
        & A unified optimization strategy framework for shared AVs \\
        {\bf This paper} & Dynamic & \textbackslash & SO-DTA 
        & Maximin AV headway control for SO-DTA on general networks \\
        \bottomrule
    \end{tabular}
    }
    \caption{Literature summary of traffic assignment with AVs.}
    \label{literature}
\end{table}

In summary, the work related to the traffic assignment problem of AVs is mainly static and there is few work exploring how AV headway control would affect the traffic states in whole general networks as a SO-DTA problem. 
There are few studies discussing the roles of headway in the control of  AVs. 
Therefore, this paper is designed to propose an optimal AV headway control framework for SO-DTA in general networks. 


\subsection{Uniqueness of SO-DTA}

The solution of SO-DTA provides valuable insights for decision-makers to develop traffic management strategies \citep{van2016user}. Most of the related works focus on the solution algorithm of SO-DTA \citep{chow2009properties,zhang2020path,qian2012system,long2019link}, and there are a few works discussing the non-uniqueness and properties of the solutions of SO-DTA. For example, \citet{SHEN20141} proves that SO-DTA may have multiple solutions with the same objective but different queue locations in networks. To explore the optimal queue distribution for the solution of SO-DTA, \citet{MA201498} proposes the solution algorithm and discusses the existence of the free-flow optimal solution of SO-DTA that all vehicles are only allowed to wait in the origin nodes, and \citet{ngoduy2016optimal} optimizes the heterogeneity of the queue lengths and the total queue lengths in traffic networks for all solutions of SO-DTA. Though both of the above studies discuss the non-uniqueness of the solutions of SO-DTA and propose the criteria to ``pick up'' the desired solution of SO-DTA, the criteria of both works are mainly related to congestion patterns. In reality, decision-makers could develop traffic management strategies considering different perspectives.
Therefore, this paper discusses the ``safest'' solution among the non-unique solutions of SO-DTA and proposes the maximin headway in this problem as the largest headway solution in SO-DTA.

\section{Modeling headway-dependent flow dynamics}
\label{sec:model}

In this section,  we propose the continuous-time form of the headway-dependent dynamic traffic flow. The continuous-time formulation is generally used in DNL, such as the continuous-time form of LTM \citep{han2016continuous,jin2015continuous}, continuous DQ \citep{MA201498} continuous-time point queue (PQ) model \citep{ban2012continuous}. To this end, we first present the headway-dependent fundamental diagram (HFD) to capture the effects of time headway on FD. Then we develop the headway-dependent double queue (HDQ) model and prove that HDQ is a generalized form of continuous-time DQ \citep{MA201498}. Lastly, the network-wide flow dynamics are modeled through a series of linear constraints.

\subsection{Notations}
The road network is represented as a directed graph \textit{G}(\textit{V},\textit{E}), where \textit{V} is the set of all nodes including dummy nodes and \textit{E} is the set of all links including O-D connectors. We denote $\widetilde{R}$ as the set of dummy origins and $\widetilde{S}$ as the set of dummy destinations. Naturally, we have $\widetilde{R}\subset\textit{V}$ and $\widetilde{S}\subset\textit{V}$. The frequently used notations are stated in Table \ref{notation_con} in \ref{app:notations}.

\subsection{Headway-dependent fundamental diagram}
\label{sec:HFD}

This section explores how different headways would affect the fundamental diagram.
Without loss of generality, it is assumed that the headway and length of AV are homogeneous. To be specific, headway can be categorized into space and time headway. The space headway refers to the distance between the rear bumper of one vehicle and the front bumper of the vehicle following it; time headway refers to the time interval that elapses from the moment the rear of one vehicle passes a fixed point until the front bumper of following vehicle passes the same point \citep{guo2012autonomous,zhu2018distributed}. The space headway is not sufficient to represent the safety between AVs, while the time headway has been widely used in the literature regarding the AV control \citep{chen2020connected,becker2022driver}. The time headway in this study is capable of measuring the probability of accidents by considering both the spacing and speed of AVs, which serves as a desirable control variable that trades off efficiency and safety.
Then we present the fundamental diagram for AVs based on the results by \citet{ZHOU2020102614}. 
Considering link $(i,j)$, we denote $\Delta_{x_{i,j}}$ as the spacing of AVs in link $(i,j)$, $\textit{v}_{i,j}$ as the speed of AVs on link $(i,j)$, $\textit{h}_{i,j}$ as the headway of AVs in link $(i,j)$ and \textit{L} as the length of each AV. The relationship between the headway $\textit{h}_{i,j}$ and spacing $\Delta_{x_{i,j}}$ can be described as follows:

\begin{equation}
\Delta_{x_{i,j}}=\textit{h}_{i,j}\textit{v}_{i,j}+\textit{L}.
\nonumber
\end{equation}

Besides, the relationship between the spacing $\Delta_{x_{i,j}}$ and density $\rho_{i,j}$ is shown as follows:
\begin{equation}
\Delta_{x_{i,j}}=\frac{1}{\rho_{i,j}}.
\nonumber
\end{equation}



The speed $\textit{v}_{i,j}$ should be limited within free-flow speed $\textit{v}_{i,j}^f$ as follows:

\begin{equation}
0\leq\textit{v}_{i,j}\leq\textit{v}_{i,j}^f. \nonumber
\end{equation}

Combining the above equations, we have Equation \ref{eq:FD} to represent the headway-dependent fundamental diagram (HFD) of AVs for $(i,j)\in E\setminus(L_R\cup L_S)$ and $t\in[0,T]$, where $L_R$ and $L_S$  represent the set of origin and destination connectors, respectively.

\begin{eqnarray}
\textit{f}_{i,j}(t) &=& \left\{
\begin{array}{lcl}
\textit{v}_{i,j}^{f}{~} \rho_{i,j}(t),  & &   {0\leq\rho_{i,j}(t)  \leq \frac{1}{\textit{h}_{i,j}(t)\textit{v}_{i,j}^f+\emph{L}}} \\
\frac{1-\rho_{i,j}(t)\emph{L}}{\textit{h}_{i,j}(t)}, &   &{\frac{1}{\textit{h}_{i,j}(t)\textit{v}_{i,j}^f+\emph{L}}   \leq \rho_{i,j}(t)\leq\frac{1}{\emph{L}}}
\end{array} \right. \nonumber\\
&=& \min\left\{\textit{v}_{i,j}^{f}{~} \rho_{i,j}(t),~ \frac{1-\rho_{i,j}(t)\emph{L}}{\textit{h}_{i,j}(t)} \right\}, \quad\quad 0 \leq \rho_{i,j}(t) \leq \frac{1}{L}
\label{eq:FD}
\end{eqnarray}
where $\rho_{i,j}(t)$ and $\textit{h}_{i,j}(t)$ are the density of the flow area and headway of link $(i,j)$ at time $t$, respectively. Figure \ref{fig:FD} presents the effect of headway on flow by HFD based on Equation~\ref{eq:FD}, where $h^{1}<h^{2}<h^{3}<h^{4}<h^{5}$. This figure is also validated by the experimental results from \citet{SHI2021279}, which empirically validates the fundamental diagram by collecting GPS trajectories of a platoon of AVs.

\begin{figure}[h]
    \centering
    \includegraphics[width=\linewidth]{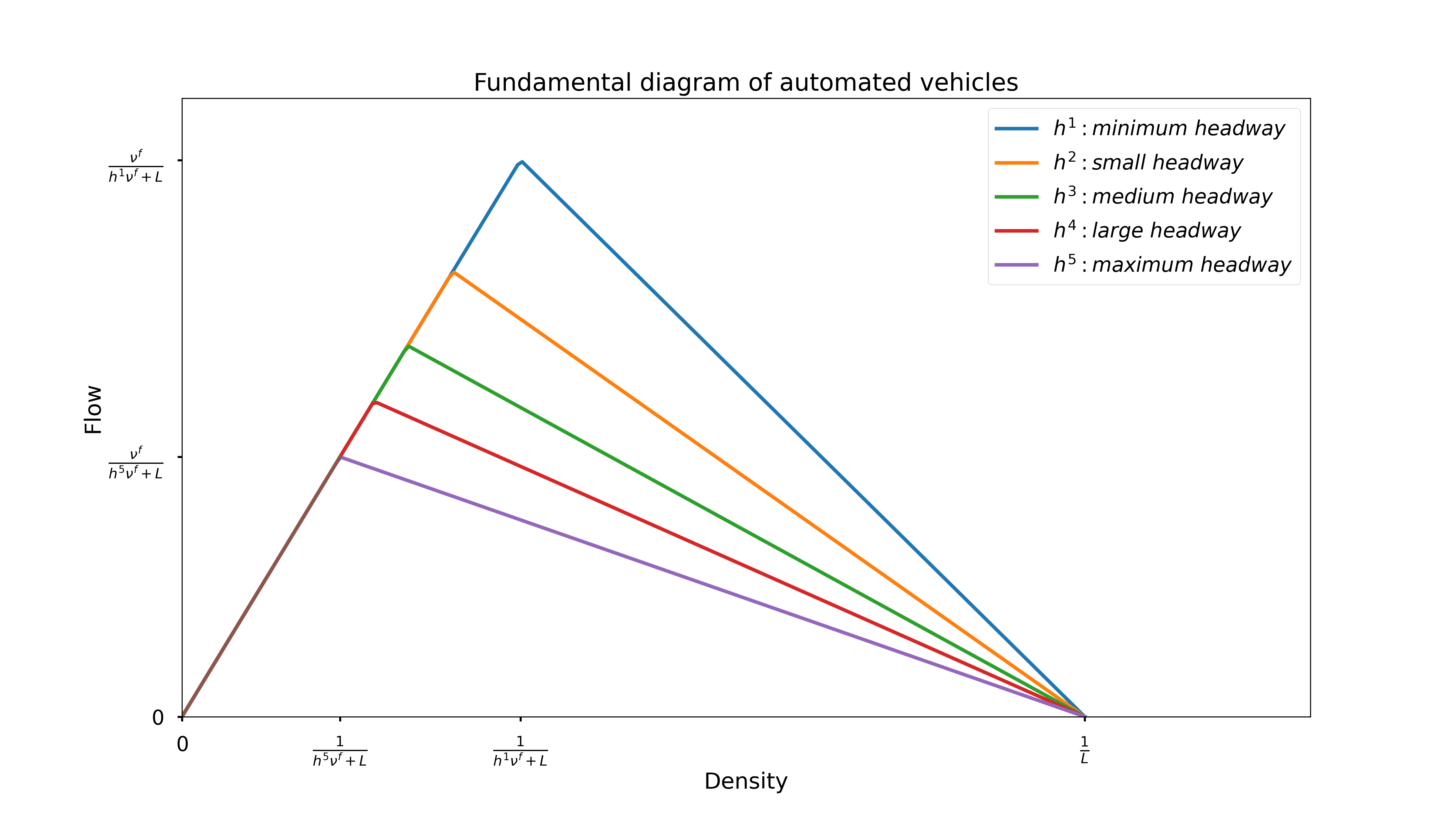}
    \caption{Fundamental diagram of AVs with different headway.}
    \label{fig:FD}
\end{figure}

As shown in Figure \ref{fig:FD}, headway could affect the HFD of AVs in the following aspects:

\begin{itemize}
    \item Scope of the free-flow region: a smaller headway corresponds to a larger scope of the free-flow region.
    \item Upper bound of flow: a smaller headway corresponds to a larger upper bound of flow.
    \item Wave travel speed in congested states: a smaller headway corresponds to a higher wave travel speed in congested states.
\end{itemize}


The wave travel time in congested states is to capture how the traffic conditions in the downstream link boundary would influence traffic conditions in the upstream link boundary \citep{LTM}.
Suppose $\tau_{i,j}^{w}(t)$ denotes the wave travel time in congested states on link $(i,j)$ at time $t$, representing the used time traveling from the exit to entrance when reaching the entrance at time $t$. If the fundamental diagram is fixed, we would have a constant wave travel time in congested states $\tau_{i,j}^{w}=\frac{L_{i,j}}{\textit{v}_{i,j}^w}$, where $\textit{v}_{i,j}^w$ is the wave travel speed in congested states on link $(i,j)$ and $L_{i,j}$ is the length of link $(i,j)$. However, in HFD, Equation \ref{eq:FD} indicates that the wave travel speed in congested states $v_{i,j}^{w}(t)$ in HFD is time-dependent given a dynamic headway $h_{i,j}(t)$ and we formulate $v_{i,j}^{w}(t)$ in the following equation as the slope of HFD in congested states. 

\begin{equation}
    v_{i,j}^{w}(t) = \frac{L}{h_{i,j}(t)}
    \nonumber
\end{equation}

Then Equation \ref{eq:shockwave_con} indicates that the wave travel time in congested states should satisfy Equation~\ref{eq:shockwave_con} for $(i,j)\in E\setminus(L_R\cup L_S),~ s^{\prime}\in \widetilde{S}$ and $t\in[0,T]$.

\begin{equation}
\int_{\textit{t}-\tau_{i,j}^{w}(t)}^{\textit{t}} v_{i,j}^{w}(\hat{t})~d\hat{t} = 
\int_{\textit{t}-\tau_{i,j}^{w}(t)}^{\textit{t}}\frac{\textit{L}}{\textit{h}_{i,j}(\hat{t})}~d\hat{t} =L_{i,j},
\label{eq:shockwave_con}
\end{equation}
where $L_{i,j}$ is the length of link $(i,j)$.

\subsection{Headway-dependent double queue model}
\label{sec:HDQ}

This section proposes a dynamic link model considering time-dependent AV headway, namely the headway-dependent double queue (HDQ) model,  We further prove that HDQ is a generalized form of the double queue (DQ) model.

In HDQ, each link is divided into two areas: flow area and buffer area. The flow area depicts the traffic flow dynamics behavior and captures the upstream queue in the link, and we suppose the traffic flow would propagate based on HFD. The buffer area captures the downstream queue in the link and the maximum number of queued vehicles is limited based on the holding capacity in the buffer area \citep{GUO202087}. An illustration of the HDQ is shown in Figure~\ref{fig:link}.  
One can see that the HDQ is inspired by the MTTM and DQ model, and it shares similarities with the Spatial Queue model \citep{gawron1998iterative}.

\begin{figure}[h]
    \centering
    \includegraphics[width=0.7\linewidth]{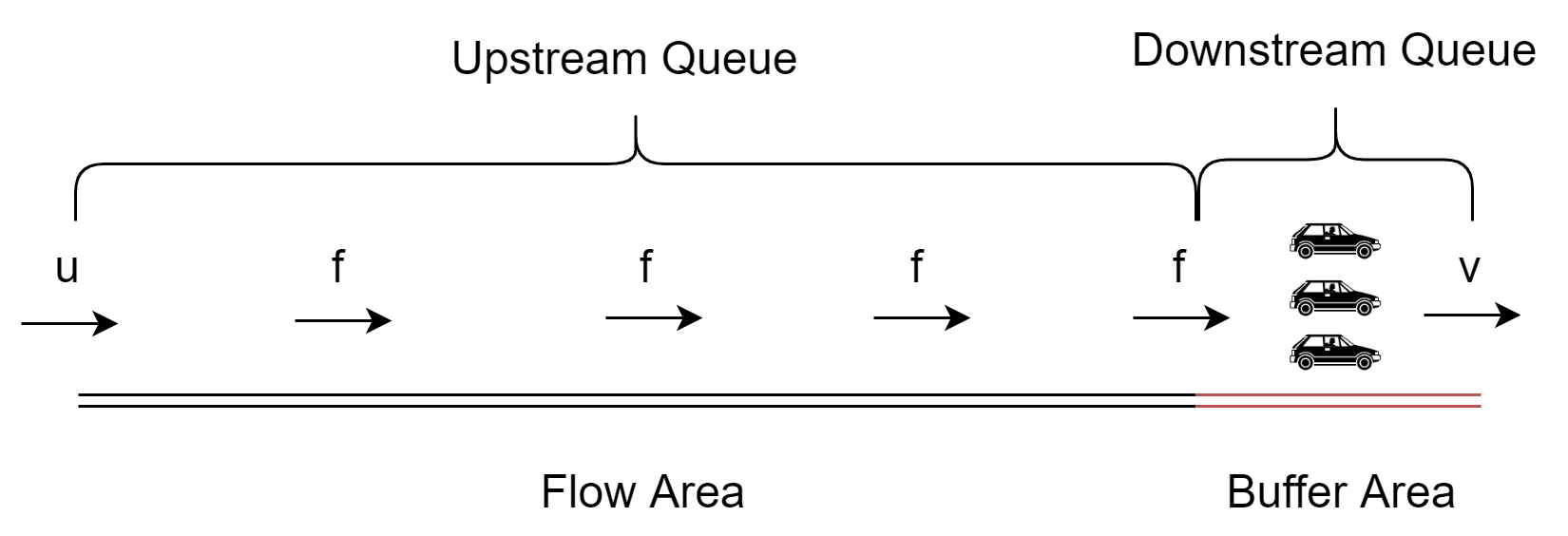}
    \caption{Flow area and buffer area in a link.}
    \label{fig:link}
\end{figure}

The HDQ has the following characteristics that are different from the DQ and MTTM: 
\begin{itemize}
\item The MTTM overlooks the possible congestion in the flow area, while the HDQ considers both upstream and downstream queues.
\item Comparing to the DQ, the HDQ makes use of the variable $f$ to represent the effect of headway control on the traffic states and network conditions.
\end{itemize}

The cumulative flow for $ (i,j)\in E,~s^{\prime}\in \widetilde{S}$ and $t\in[0,T]$ is defined as follows:
\begin{equation}
\textit{U}_{i,j}^{s^{\prime}}(t) = \int_{0}^{\textit{t}}\textit{u}_{i,j}^{s^{\prime}}(\hat{t})~d\hat{t}, \quad 
\textit{V}_{i,j}^{s^{\prime}}(t) =\int_{0}^{\textit{t}}\textit{v}_{i,j}^{s^{\prime}}(\hat{t})~d\hat{t}, \quad
\textit{F}_{i,j}^{s^{\prime}}(t) ~= ~\int_{0}^{\textit{t}}\textit{f}_{i,j}^{~s^{\prime}}(\hat{t})~d\hat{t},
\nonumber
\end{equation}
where $\textit{U}_{i,j}^{s^{\prime}}(t)$, $\textit{V}_{i,j}^{s^{\prime}}(t)$ and $\textit{F}_{i,j}^{s^{\prime}}(t)$ are the cumulative inflow of the flow area, the cumulative outflow of the buffer area and cumulative inflow at the boundary of the flow area and buffer area of link $(i,j)$ with destination $s^{\prime}$ at time $t$, respectively. $\textit{u}_{i,j}^{s^{\prime}}(t)$, $\textit{v}_{i,j}^{s^{\prime}}(t)$ and $\textit{f}_{i,j}^{~s^{\prime}}(t)$ are the inflow of the flow area, outflow of the buffer area and inflow at the boundary of the flow area and buffer area with destination $s^{\prime}$ of link $(i,j)$ at time $t$, respectively.

Then the downstream queue and upstream queue in HDQ for $ (i,j)\in E,~s^{\prime}\in \widetilde{S}$ and $t\in[0,T]$ are presented as follows:
\begin{eqnarray}
\textit{q}_{i,j}^{\mathcal{D},s^{\prime}}(t) &=& \int_{0}^{\textit{t}} \textit{f}_{i,j}^{~ s^{\prime}}(\hat{t}) ~d\hat{t} - \int_{0}^{\textit{t}} \textit{v}_{i,j}^{s^{\prime}}(\hat{t}) ~d\hat{t} + \textit{q}_{i,j}^{\mathcal{D},s^{\prime}}(0) = \textit{F}_{i,j}^{s^{\prime}}(t)-\textit{V}_{i,j}^{s^{\prime}}(t), 
\label{eq:dq_con} \\
\textit{q}_{i,j}^{\mathcal{U},s^{\prime}}(t) &=& \int_{0}^{\textit{t}} \textit{u}_{i,j}^{~ s^{\prime}}(\hat{t}) ~d\hat{t} - \int_{0}^{\textit{t}-\tau_{i,j}^{w}(t)} \textit{f}_{i,j}^{~ s^{\prime}}(\hat{t}) ~d\hat{t} + \textit{q}_{i,j}^{\mathcal{U},s^{\prime}}(0) = \textit{U}_{i,j}^{s^{\prime}}(t)-\textit{F}_{i,j}^{s^{\prime}}(t-\tau_{i,j}^{w}(t)), 
\label{eq:uq_con} 
\end{eqnarray}
where $\textit{q}_{i,j}^{\mathcal{D},s^{\prime}}(t)$ and $\textit{q}_{i,j}^{\mathcal{U},s^{\prime}}(t)$ are the downstream queue and upstream queue with destination $s^{\prime}$ of link $(i,j)$ at time $t$, respectively.

Proposition \ref{prop:HDQ&DQ} presents the relationship between the DQ and HDQ, indicating that HDQ could be reduced to DQ under its assumption.

\begin{proposition}[Reduction of HDQ to DQ]
HDQ is a generalized form of DQ by incorporating the influence of headway; HDQ could be reduced to the DQ under the assumption of DQ.
\label{prop:HDQ&DQ}
\end{proposition}
\begin{proof}
See \ref{app:generality}.
\end{proof}



We further have Equations \ref{eq:sum_1_con} and \ref{eq:sum_2_con} hold, meaning that traffic flow dynamics variables are the summation of their destination-specified variables for $ (i,j)\in E,~ s^{\prime}\in \widetilde{S}$ and $t\in[0,T]$.

\begin{alignat}{4}
\textit{q}_{i,j}^{\mathcal{D}}(t) &=\sum_{s^{\prime}}\textit{q}_{i,j}^{\mathcal{D},s^{\prime}}(t),   &&\quad
\textit{q}_{i,j}^{\mathcal{U}}(t) = \sum_{s^{\prime}} \textit{q}_{i,j}^{\mathcal{U},s^{\prime}}(t), &&&\quad
\rho_{i,j}(t) = \sum_{s^{\prime}} \rho_{i,j}^{s^{\prime}}(t),  
\label{eq:sum_1_con}\\
\textit{u}_{i,j}(t) &= \sum_{s^{\prime}} \textit{u}_{i,j}^{s^{\prime}}(t), &&\quad \textit{f}_{i,j}(t) = \sum_{s^{\prime}} \textit{f}_{i,j}^{~ s^{\prime}}(t), &&&\quad
\textit{v}_{i,j}(t) = \sum_{s^{\prime}} \textit{v}_{i,j}^{s^{\prime}}(t), 
\label{eq:sum_2_con}
\end{alignat}
where $\textit{q}_{i,j}^{\mathcal{D}}(t)$ and $\textit{q}_{i,j}^{\mathcal{U}}(t)$ are the downstream queue and upstream queue of link $(i,j)$ at time $t$, respectively. $\textit{u}_{i,j}(t)$, $\textit{v}_{i,j}(t)$ and $\textit{f}_{i,j}(t)$ are the inflow of the flow area, outflow of the buffer area and inflow at the boundary of the flow area and buffer area of link $(i,j)$ at time $t$, respectively. $\rho_{i,j}(t)$  and $\rho_{i,j}^{s^{\prime}}(t)$ are the density of the flow area and the density of the flow area with destination $s^{\prime}$ of link $(i,j)$ at time $t$, respectively.

Assuming the density is homogeneous in flow area,  the change rate of density is presented in Equation \ref{eq:density_con} for $(i,j)\in E\setminus(L_R\cup L_S),~ s^{\prime}\in \widetilde{S}$ and  $t\in[0,T]$.

\begin{equation}
\dot{\rho}_{i,j}^{s^{\prime}}(t)=\frac{\textit{u}_{i,j}^{s^{\prime}}(t)-\textit{f}_{i,j}^{~s^{\prime}}(t)}{L_{i,j}}
\label{eq:density_con}
\end{equation}

Additionally, Equation \ref{eq:ineq_1_con} and \ref{eq:ineq_2_con} indicate that the downstream queue, upstream queue, inflow rate, and outflow rate are bounded by their respective capacity for $ (i,j)\in E\setminus(L_{R}\cup L_{S}),~ s^{\prime}\in \widetilde{S}$ and $t\in[0,T]$.

\begin{alignat}{4}
\textit{q}_{i,j}^{\mathcal{D}}(t) &\leq \bar{Q}_{i,j}^{\mathcal{D}},   \quad
\textit{q}_{i,j}^{\mathcal{U}}(t) &&\leq \bar{Q}_{i,j}^{\mathcal{U}},
\label{eq:ineq_1_con}\\
\textit{u}_{i,j}(t) &\leq \bar{C}_{i,j}^{\mu}, \quad 
\textit{v}_{i,j}(t) &&\leq \bar{C}_{i,j}^{\nu},
\label{eq:ineq_2_con}
\end{alignat}
where $\bar{Q}_{i,j}^{\mathcal{D}}$ and $\bar{Q}_{i,j}^{\mathcal{U}}$ are the downstream queue capacity and upstream queue capacity, respectively; $\bar{C}_{i,j}^{\mu}$ and $\bar{C}_{i,j}^{\nu}$ are the inflow and outflow capacity of link $(i,j)$ at time $t$, respectively.
We further have Proposition \ref{prop:property} indicating that the impact to the maximal flow rate of the link is dependent on the HDQ.

\begin{proposition}[Headway-dependent flow rate bounds] The upper bounds of the inflow rate $\textit{u}_{i,j}(t)$ and outflow rate $\textit{v}_{i,j}(t)$ are dependent on the upstream queue $\textit{q}^{\mathcal{U}}_{i,j}(t)$ and downstream queue $\textit{q}^{\mathcal{D}}_{i,j}(t)$, respectively, and their relation is shown as follows.
\begin{equation}
\begin{array}{ccccc}
\qquad\qquad~\textit{u}_{i,j}(t) \leq \begin{cases}
\bar{C}_{i,j}^{\mu}, & \textit{q}_{i,j}^{\mathcal{U}}(t)<\bar{Q}_{i,j}^{u} \\
\min \left\{ \textit{f}_{i,j}(t-\tau_{i,j}^{w}(t)),~ \bar{C}_{i,j}^{\mu} \right\}, & \textit{q}_{i,j}^{\mathcal{U}}(t)=\bar{Q}_{i,j}^{u}
\end{cases} \\
\textit{v}_{i,j}(t) \leq \begin{cases}
\bar{C}_{i,j}^{\nu}, & \textit{q}_{i,j}^{\mathcal{D}}(t)>0 \\
\min \left\{ \textit{f}_{i,j}(t),~ \bar{C}_{i,j}^{\nu} \right\}, &
\textit{q}_{i,j}^{\mathcal{D}}(t)=0
\end{cases}
\end{array}  
\nonumber
\end{equation}
\label{prop:property}
\end{proposition}

\begin{proof}
See \ref{app:HDQ}.
\end{proof}

\subsection{Dynamic network traffic model}
\label{sec:Autonomous_vehicle_Headway_Control}
This section models the network-wide effects of AV headway control by developing a dynamic network model. 

The flow conservation at nodes is stated in Equation \ref{eq:flow_conser_con} for $i \in N\setminus(\widetilde{R}\cup \widetilde{S}),~ s^{\prime}\in \widetilde{S}$ and $t\in[0,T]$. 

\begin{equation}
\sum_{p}\textit{v}_{p,i}^{s^{\prime}}(t)=\sum_{j}\textit{u}_{i,j}^{s^{\prime}}(t)
\label{eq:flow_conser_con}
\end{equation}

As the node model is not the main contribution of this paper, we adopt the simplified model that only considers the conservation law, while other advanced node models can be also incorporated if necessary.
For example, node models could be developed in macroscopic levels to reduce congestion combined with AVs \citep{li2023optimal}, prevent the conflicting movement at intersections \citep{zhu2015linear} and eliminate the traffic-holding phenomenon \citep{jabari2016node}. Similarly, our proposed framework could also incorporate the node model and apply corresponding strategies at intersections to avoid conflicts, such as traffic signal control \citep{yu2018optimal}.

The link conservation is ensured through the HDQ model. Equation \ref{ineq:nonnega_con} presents the non-negative constraint of all the variables involved for $(i,j)\in E,~ s^{\prime}\in \widetilde{S}$ and $t\in[0,T]$.

\begin{equation}
\rho_{i,j}^{s^{\prime}}(t),~ \textit{u}_{i,j}^{s^{\prime}}(t), ~ \textit{f}_{i,j}^{~s^{\prime}}(t), ~ \textit{v}_{i,j}^{s^{\prime}}(t), ~ \textit{q}_{i,j}^{\mathcal{D},s^{\prime}}(t), ~
\textit{q}_{i,j}^{\mathcal{U},s^{\prime}}(t)\geq0
\label{ineq:nonnega_con}
\end{equation}

Besides, the number of AVs in link $(i,j)$  with destination $s^{\prime}$ at time $t$ could be represented as the sum of the number of AVs in the flow area and buffer area, as shown in Equation~\ref{eq:sum}.

\begin{equation}
\label{eq:sum}
\begin{array}{clll}
 \textit{N}_{i,j}^{s^{\prime}}(t) &=&  \int_{\textit{0}}^{\textit{t}}\textit{u}_{i,j}^{s^{\prime}}(\hat{t}) ~d\hat{t} - \int_{\textit{0}}^{\textit{t}}\textit{v}_{i,j}^{s^{\prime}}(\hat{t}) ~d\hat{t}  \\
&=& \int_{\textit{0}}^{\textit{t}}\textit{u}_{i,j}^{s^{\prime}}(\hat{t}) ~d\hat{t} -
\int_{\textit{0}}^{\textit{t}}\textit{f}_{i,j}^{~s^{\prime}}(\hat{t}) ~d\hat{t} +
\int_{\textit{0}}^{\textit{t}}\textit{f}_{i,j}^{~s^{\prime}}(\hat{t}) ~d\hat{t} -
\int_{\textit{0}}^{\textit{t}}\textit{v}_{i,j}^{s^{\prime}}(\hat{t}) ~d\hat{t} \\ &=& L_{i,j} ~\rho_{i,j}^{s^{\prime}}(t)+\textit{q}_{i,j}^{\mathcal{D},s^{\prime}}(t)
\end{array} 
\end{equation}

The initial condition of the network is stated in Equation \ref{eq:init_con} for $(i,j)\in E, ~ s^{\prime}\in \widetilde{S}$ and $t\in[0,T]$, which guarantees the whole traffic network is empty before loading.

\begin{equation}
\rho_{i,j}^{s^{\prime}}(0), ~
\textit{q}_{i,j}^{\mathcal{D},s^{\prime}}(0)=0
\label{eq:init_con}
\end{equation}

The end condition constraint is presented in Equation \ref{ineq:end_con} for $(i,j)\in E, ~ s^{\prime}\in \widetilde{S}$ and $t\in[0,T]$, which guarantees the whole traffic network is empty in the end. 
\begin{equation}
\textit{q}_{i,j}^{\mathcal{D},s^{\prime}}(T)=0, ~
\rho_{i,j}(T) \leq \frac{1}{L_{i,j}}
\label{ineq:end_con}
\end{equation}

Importantly, the change rate of link density can be represented by $\dot{\rho}_{i,j}(t)=-\frac{\textit{v}_{i,j}^{f}\,\rho_{i,j}(t)}{L_{i,j}}, t\geq t_{0}$ based on Equations~\ref{eq:sum_1_con} to \ref{eq:density_con} when $\textit{u}_{i,j}(t) = 0$. This represents the scenario in which no vehicles will enter the link 
when $\rho_{i,j}(t_{0})>0$. $\rho_{i,j}(t)$ has the shape of the negative exponential function, and hence it will take infinite time to make $\rho_{i,j}(t)=0$. In order to avoid this situation, we set the end constraint of the flow area in link $(i,j)$ as $\rho_{i,j}(T) \leq \frac{1}{L_{i,j}}$, which means that we suppose the flow area is empty when the number of vehicles in flow area is smaller than 1 at the end time $T$. Furthermore, we could also set a large positive constant $c$ as $\rho_{i,j}(T) \leq \frac{1}{cL_{i,j}}$ to ensure the vehicles in the total network is close to 0 at end time $T$.





Equation \ref{eq:demand_con} depicts the number of AVs waiting at origins for $r^{\prime} \in \widetilde{R},~ s^{\prime}\in \widetilde{S}$ and $t \in [0,T]$:

\begin{equation}
\textit{w}_{r^{\prime},r}^{~s^{\prime}}(t)= \int_{\textit{0}}^{\textit{t}}\textit{d}_{r^{\prime},s^{\prime}}(\hat{t}) ~d\hat{t} - \int_{\textit{0}}^{\textit{t}}\textit{v}_{r^{\prime},r}^{s^{\prime}}(\hat{t}) ~d\hat{t}
\label{eq:demand_con}
\end{equation}
where $\textit{w}_{r^{\prime},r}^{~s^{\prime}}(t)$ represents the number of AVs waiting at origins in the origin connectors  $(r^{\prime},r)$ with destination $s^{\prime}$ at time $t$.  $\textit{v}_{r^{\prime},r}^{s^{\prime}}(t)$ is the outflow of origin connectors  $(r^{\prime},r)$ with destination $s^{\prime}$ at time $t$. $\textit{d}_{r^{\prime},s^{\prime}}(t)$ represents the travel demand of O-D pair $(r^{\prime},s^{\prime})$ at time $t$. The demand is restricted in the time interval $[O,T_{1}]$, so we have $\textit{d}_{r^{\prime},s^{\prime}}(t)=0$ for $T_{1}<t \leq T$.
 

Additionally, Equation \ref{ineq:headway_con} limits the minimum and maximum headway for $(i,j)\in E\setminus(L_R\cup L_S),~ s^{\prime}\in \widetilde{S}$ and $t\in[0,T]$.

\begin{equation}
\textit{h}_{i,j;t}^{min}\leq\textit{h}_{i,j}(t)\leq
\textit{h}_{i,j;t}^{max}
\label{ineq:headway_con}
\end{equation}
where $\textit{h}_{i,j;t}^{min}$ and $\textit{h}_{i,j;t}^{max}$ are the minimum and maximum headway of link $(i,j)$ at time $t$. $\textit{h}_{i,j;t}^{min}$ represents the tolerance of safety. If the headway ${h}_{i,j}(t)$ is less than $\textit{h}_{i,j;t}^{min}$, the driving condition in a link would be highly dangerous. $\textit{h}_{i,j;t}^{max}$ depicts the bound of V2V communication. If the headway ${h}_{i,j}(t)$ is larger than $\textit{h}_{i,j;t}^{max}$, the V2V communication may be damaged.

\section{Optimal headway control}
\label{sec_system_optimal_headway_formulation}



In this section, we first formulate the continuous and discretized headway control for SO-DTA.
Then we prove that the minimum headway control could achieve SO-DTA, and it is not the unique solution. 

To search for the SO-DTA, there are also two types of formulations in the community of dynamic traffic assignment (DTA): path-based formulation and link-based formulation. The two formulations are actually equivalent, and we can solve either one to obtain the SO-DTA conditions. Specifically, the link flow and path flow can be converted to each other through the path-link incidence matrix or Dynamic Assignment Ratio (DAR) matrix \citep{ma2018estimating}. Both formulations consider the utility changes due to the change of link-specific headway. In this paper, we adopt the link-based formulation as it can facilitate analytical analysis and solution algorithms, where the headway control is directly applied in each link, then the link flow/cost would be influenced, thus each path flow/cost is also impacted under the link-specific headway control. Therefore, the network-wide effect of link-specific time headway control could be analyzed.

\subsection{Continuous-time formulation}
\label{sec:contin_SO-DTA}

The objective of the SO-DTA aims to minimize total travel time (TTT), which equals to the sum of the travel time on road networks and the waiting time at origins, as represented by Equation~\ref{eq:obj_con}. The decision variables include traffic dynamics $\textbf{x}$ and headway $\textbf{h}$. 


\begin{equation}
\min_{\textbf{x}, \textbf{h}} ~TTT= \int_{\textit{0}}^{\textit{T}}\left(\sum_{r^{\prime}}\int_{\textit{0}}^{\textit{t}}\textit{v}_{r^{\prime},r}(\hat(t)) ~d\hat(t) - \sum_{s^{\prime}}\int_{\textit{0}}^{\textit{t}}\textit{u}_{s,s^{\prime}}(\hat(t)) ~d\hat(t)\right) ~dt + \int_{0}^{T}\sum\limits_{(r^{\prime},s^{\prime})}\textit{w}_{r^{\prime},\textit{r}}^{s^{\prime}}(t) ~dt
\label{eq:obj_con}
\end{equation}
subject to:
\begin{itemize}
\item Headway-dependent fundamental diagram constraints: \textit{Eqs}. (\ref{eq:FD})-(\ref{eq:shockwave_con})
\item Headway-dependent double queue constraints: \textit{Eqs}.  (\ref{eq:dq_con})-(\ref{eq:ineq_2_con})
\item Flow conservation constraints: \textit{Eq}. (\ref{eq:flow_conser_con})
\item Non-negative constraints: \textit{Eq}. (\ref{ineq:nonnega_con})
\item Initial constraints: \textit{Eq}. (\ref{eq:init_con})
\item End constraints: \textit{Eq}. (\ref{ineq:end_con})
\item Other constraints: \textit{Eqs}. (\ref{eq:demand_con})-(\ref{ineq:headway_con})
\end{itemize}




\subsection{Discretized-time formulation}

This section discretizes the formulation proposed in Section~\ref{sec:contin_SO-DTA}. 
We set $\textit{T}=\textit{N}\Delta_{t}$, where $\textit{N}$  represents the number of time intervals and $\Delta_{t}$ represents the length of a time interval. A summary of the notations is shown in Table \ref{notation_dis}.

\begin{table}[H]
    \centering
    \resizebox{1.02\textwidth}{!}{
    \begin{tabular}{lll}
        \toprule
        Notation   & Description \\
        \midrule
         $\rho_{i,j}(k)$ & Density of the flow area of link $(i,j)$ in the  $\textit{k}$th time interval \\
         $\rho_{i,j}^{s^{\prime}}(k)$ & Density of the flow area of link $(i,j)$ with destination $s^{\prime}$ in the  $\textit{k}$th time interval \\
        $\textit{q}_{i,j}^{\mathcal{D}}(k)$ & Downstream queue of link $(i,j)$ in the  $\textit{k}$th time interval  \\
        $\textit{q}_{i,j}^{\mathcal{D},s^{\prime}}(k)$ & Downstream queue of link $(i,j)$ with destination $s^{\prime}$ in the  $\textit{k}$th time interval  \\
        $\textit{q}_{i,j}^{\mathcal{U}}(k)$ & Upstream queue of link $(i,j)$ in the  $\textit{k}$th time interval  \\
        $\textit{q}_{i,j}^{\mathcal{U},s^{\prime}}(k)$ & Upstream queue of link $(i,j)$ with destination $s^{\prime}$ in the  $\textit{k}$th time interval  \\
        $\textit{u}_{i,j}(k)$ & Inflow of the flow area of link $(i,j)$ in the  $\textit{k}$th time interval  \\
        $\textit{u}_{i,j}^{s^{\prime}}(k)$ & Inflow of the flow area of link $(i,j)$ with destination $s^{\prime}$ in the  $\textit{k}$th time interval  \\
        $\textit{f}_{i,j}(k)$ & Inflow at the boundary of the flow area and  buffer area of link $(i,j)$ in the  $\textit{k}$th time interval  \\
        $\textit{f}_{i,j}^{~s^{\prime}}(k)$ & Inflow at the boundary of the flow area and buffer area of link $(i,j)$ with destination $s^{\prime}$ in the $\textit{k}$th time interval \\
        $\textit{v}_{i,j}(k)$ & Outflow of the buffer area of link $(i,j)$ in the  $\textit{k}$th time interval  \\
        $\textit{v}_{i,j}^{s^{\prime}}(k)$ & Outflow of the buffer area of link $(i,j)$ with destination $s^{\prime}$ in the  $\textit{k}$th time interval  \\
        $\textit{h}_{i,j}(k)$ & Headway of automated vehicles on link $(i,j)$ in the $\textit{k}$th time interval  \\
        $\textit{w}_{r^{\prime},r}^{s^{\prime}}(k)$ & Number of automated vehicles waiting in the origin connectors $(r^{\prime},r)$ with destination $s^{\prime}$ in the $\textit{k}$th time interval  \\
        $n_{i,j}^{w}(k)$ & Wave travel time in congested states on link $(i,j)$ in the $\textit{k}$th time interval \\
        \bottomrule
    \end{tabular}
    }
    \caption{Notations in the discretized formulation.}
    \label{notation_dis}
\end{table}

We are now ready to present the discretized version of the headway-dependent SO-DTA formulation as follows: 
\begin{equation}
\min_{\textbf{x},\textbf{h}} ~TTT=  \quad \sum_{r^{\prime}}\sum_{\textit{k}=1}^{\textit{N}}\Delta_{t}\left(\sum_{\textit{l}=1}^{\textit{k}}\Delta_{t}\textit{v}_{r^{\prime},r}(l)\right)-\sum_{s^{\prime}}\sum_{\textit{k}=1}^{\textit{N}}\Delta_{t}\left(\sum_{\textit{l}=1}^{\textit{k}}\Delta_{t}\textit{u}_{s,s^{\prime}}(l)\right)+\sum_{(r^{\prime},s^{\prime})}\sum_{\textit{k}=1}^{\textit{N}}\Delta_{t}\textit{w}_{r^{\prime},\textit{r}}^{s^{\prime}}(k)
\label{eq:obj_dis}
\end{equation}

subject to

1. Headway-dependent fundamental diagram constraints for $(i,j)\in E\setminus(L_R\cup L_S),~ s^{\prime}\in \widetilde{S}$ and $1\leq k \leq N$. 

\begin{equation}
\textit{f}_{i,j}(k) = \min\left\{\textit{v}_{i,j}^{f}{~} \rho_{i,j}(k),~ \frac{1-\rho_{i,j}(k)\emph{L}}{\textit{h}_{i,j}(k)} \right\}, \quad\quad 0 \leq \rho_{i,j}(k) \leq \frac{1}{L}
\label{eq:FD_dis}
\end{equation}




The discretized form of Equation \ref{eq:shockwave_con} representing the wave travel time in congested states is shown as follows.

\begin{eqnarray}
\sum\limits_{l=k-n_{i,j}^{w}(k)}^{k} \Delta_{t} \frac{\textit{L}}{\textit{h}_{i,j}(l)}  &\leq& L_{i,j} 
\label{eq:shockwave_dis3}  \\
\sum\limits_{l=k-n_{i,j}^{w}(k)-1}^{k} \Delta_{t} \frac{\textit{L}}{\textit{h}_{i,j}(l)}  & \geq& L_{i,j} 
\label{eq:shockwave_dis4}
\end{eqnarray}

For the convenience of searching the maximin headway defined in Section \ref{sec_Optimal_maximin_Headway_Formulation}, we assume that wave travel time in congested states $n_{i,j}^{w}(k)$ is only determined by the headway $\textit{h}_{i,j}(k)$ in Equation \ref{eq:shockwave_dis1} and \ref{eq:shockwave_dis2}. This relaxation would not influence the properties and definition of both optimal headway and maximin headway. In the detailed proof, Propositions \ref{prop:feasibility}  to \ref{prop:unique} could also be achieved under Equations \ref{eq:shockwave_dis3} and \ref{eq:shockwave_dis4}. The relaxation is only designed to convert the proposed SO-DTA problem into two LPs to solve the maximin headway in Algorithm \ref{algo:optimal}.
\begin{eqnarray}
n_{i,j}^{w}(k) \Delta_{t} \frac{\textit{L}}{\textit{h}_{i,j}(k)}  & \leq & L_{i,j} 
\label{eq:shockwave_dis1} \\
(n_{i,j}^{w}(k)+1) \Delta_{t} \frac{\textit{L}}{\textit{h}_{i,j}(k)}  & \geq & L_{i,j} 
\label{eq:shockwave_dis2}
\end{eqnarray}

2. Headway-dependent double queue constraints constraints for $s^{\prime}\in \widetilde{S}$ and $1\leq k \leq N$.  We limit $(i,j)\in E$ in Equations \ref{eq:dq_dis} to \ref{eq:sum_2_dis} and $(i,j)\in E\setminus(L_R\cup L_S)$ in  Equations \ref{eq:ineq_1_dis} to \ref{eq:ineq_2_dis}.

\begin{eqnarray}
\textit{q}_{i,j}^{\mathcal{D},s^{\prime}}(k) &=& \sum_{l=0}^{\textit{k}}\Delta_{t} ~\textit{f}_{i,j}^{~ s^{\prime}}(l) - \sum_{l=0}^{\textit{k}}\Delta_{t} \textit{v}_{i,j}^{s^{\prime}}(l)  + \textit{q}_{i,j}^{\mathcal{D},s^{\prime}}(0)
\label{eq:dq_dis} \\
\textit{q}_{i,j}^{\mathcal{U},s^{\prime}}(k) &=& \sum_{l=0}^{\textit{k}}\Delta_{t} \textit{u}_{i,j}^{~ s^{\prime}}(l) - \sum_{l=0}^{\textit{k}-n_{i,j}^{w}(k)}\Delta_{t} ~\textit{f}_{i,j}^{~ s^{\prime}}(t)  + \textit{q}_{i,j}^{\mathcal{U},s^{\prime}}(0)
\label{eq:uq_dis} 
\end{eqnarray}

\begin{alignat}{4}
\textit{q}_{i,j}^{\mathcal{D}}(k) &=\sum_{s^{\prime}}\textit{q}_{i,j}^{\mathcal{D},s^{\prime}}(k);   &&\quad
\textit{q}_{i,j}^{\mathcal{U}}(k) = \sum_{s^{\prime}} \textit{q}_{i,j}^{\mathcal{U},s^{\prime}}(k); &&&\quad
\rho_{i,j}(k) = \sum_{s^{\prime}} \rho_{i,j}^{s^{\prime}}(k)   \label{eq:sum_1_dis}\\
\textit{u}_{i,j}(k) &= \sum_{s^{\prime}} \textit{u}_{i,j}^{s^{\prime}}(k); &&\quad \textit{f}_{i,j}(k) = \sum_{s^{\prime}} \textit{f}_{i,j}^{~s^{\prime}}(k); &&&\quad
\textit{v}_{i,j}(k) = \sum_{s^{\prime}} \textit{v}_{i,j}^{s^{\prime}}(k) 
\label{eq:sum_2_dis}
\end{alignat}

\begin{equation}
\rho_{i,j}^{s^{\prime}}(k)=\rho_{i,j}^{s^{\prime}}(k-1)+\Delta_{t}\frac{\textit{u}_{i,j}^{s^{\prime}}(k)-\textit{f}_{i,j}^{~s^{\prime}}(k)}{L_{i,j}} 
\label{eq:density_dis}
\end{equation}

\begin{alignat}{4}
\textit{q}_{i,j}^{\mathcal{D}}(k) &\leq \bar{Q}_{i,j}^{\mathcal{D}};   \quad
\textit{q}_{i,j}^{\mathcal{U}}(k) &&\leq \bar{Q}_{i,j}^{\mathcal{U}}
\label{eq:ineq_1_dis}\\
\textit{u}_{i,j}(k) &\leq \bar{C}_{i,j}^{\mu}; \quad 
\textit{v}_{i,j}(k) &&\leq \bar{C}_{i,j}^{\nu}
\label{eq:ineq_2_dis}
\end{alignat}

3. Flow conservation constraints for $i \in N\setminus(\widetilde{R}\cup \widetilde{S}),~ s^{\prime}\in \widetilde{S}$ and $1\leq k \leq N$.

\begin{equation}
\sum_{p}\textit{v}_{\emph{p},\emph{i}}^{s^{\prime}}(k)=\sum_{j}\textit{u}_{i,j}^{s^{\prime}}(k)
\label{eq:flow_conser_dis}
\end{equation}

4. Non-negative Constraint for $(i,j)\in E,~ s^{\prime}\in \widetilde{S}$ and $1\leq k \leq N $.

\begin{equation}
\rho_{i,j}^{s^{\prime}}(k), \textit{u}_{i,j}^{s^{\prime}}(k), ~ \textit{f}_{i,j}^{~s^{\prime}}(k), ~ \textit{v}_{i,j}^{s^{\prime}}(k), ~ \textit{q}_{i,j}^{\mathcal{D},s^{\prime}}(k), ~
\textit{q}_{i,j}^{\mathcal{U},s^{\prime}}(k)\geq0
\label{ineq:nonnega_dis}
\end{equation}

4. Initial condition constraint for $(i,j)\in E$ and $s^{\prime}\in \widetilde{S}$.

\begin{equation}
\rho_{i,j}(0), ~
\textit{q}_{i,j}^{\mathcal{D},s^{\prime}}(0)=0
\label{eq:init_dis}
\end{equation}

5. End condition constraint for $(i,j)\in E$ and $s^{\prime}\in \widetilde{S}$.

\begin{equation}
\textit{q}_{i,j}^{\mathcal{D},s^{\prime}}(N)=0, ~
\rho_{i,j}(N) \leq \frac{1}{L_{i,j}}
\label{ineq:end_dis}
\end{equation}

6. Other constraints for $(i,j)\in E\setminus(L_R\cup L_S),~ r^{\prime} \in \widetilde{R},~ s^{\prime} \in \widetilde{S}$ and $1\leq k \leq N $, where $\textit{d}_{r^{\prime},s^{\prime}}(k)$ represents the travel demand of O-D pair $(r^{\prime},s^{\prime})$ in the $k$th time interval and $N_{1}=\frac{T_{1}}{\Delta_{t}}$.



\begin{equation}
\textit{w}_{r^{\prime},r}^{~s^{\prime}}(t)= \sum_{l=0}^{\textit{k}}\Delta_{t}\textit{d}_{r^{\prime},s^{\prime}}(l) - \sum_{l=0}^{\textit{k}}\Delta_{t} \textit{v}_{r^{\prime},r}^{s^{\prime}}(l) 
\label{eq:demand_dis}
\end{equation}


\begin{equation}
\textit{h}_{i,j;k}^{min}\leq\textit{h}_{i,j}(k)\leq
\textit{h}_{i,j;k}^{max} 
\label{ineq:headway_dis}
\end{equation}

Proposition \ref{prop:feasibility} presents the feasibility of the proposed optimal headway control problem in Equation \ref{eq:obj_dis}, where $\textit{N}_\textit{1}$, $N^{r^{\prime},s^{\prime}}$, $P^{r^{\prime},s^{\prime}}$, $N_{i,j}^{r^{\prime},s^{\prime}}$ and $\delta_{i,j}^{P^{r^{\prime},s^{\prime}}}$ are defined in \ref{app:feas}.

\begin{proposition}[Feasibility]
    If $\textit{T}>\textit{T}_\textit{1}+\sum\limits_{(r^{\prime},s^{\prime})}n^{r^{\prime},s^{\prime}}\sum\limits_{(i,j)}T_{i,j}^{r^{\prime},s^{\prime}}\delta_{i,j}^{P^{r^{\prime},s^{\prime}}}$, the optimal headway control problem in Equation~\ref{eq:obj_dis} is always feasible.
    \label{prop:feasibility}
\end{proposition}
\begin{proofsketch}
It is obvious that if we choose a large enough time horizon $N$, the problem is always feasible, which means all AVs would arrive at their destinations before $N$. So the basic idea for the proof is to find a large enough $N$ and construct a feasible solution under $N$. Details are presented in \ref{app:feas}.
\end{proofsketch}

We define the minimum total travel time and optimal headway of the discretized form of SO-DTA for Equations \ref{eq:obj_dis} to \ref{ineq:headway_dis} in Definition \ref{def:optimal_headway}. $\Omega_{\textbf{h}}$ represents the set of flow dynamics $\textbf{x}$ such that $\textbf{x}$ and $\textbf{h}$ constitute a feasible solution of SO-DTA. $TTT_{\textbf{h}}$ represents the minimum $TTT$ if we set the headway variable $\textbf{h}$ as an exogenous input.

\begin{definition}
\label{def:optimal_headway}
Suppose $\Omega_{\textbf{h}}=\{\textbf{x}\,|\,\textbf{x}~and~\textbf{h}~satisfy~Equations~\ref{eq:FD_dis}; \ref{eq:shockwave_dis1}~to~\ref{ineq:headway_dis}\}$, and $TTT(\textbf{x})$ represents the total travel time under $\textbf{x}$ as Equation \ref{eq:obj_dis} is only related with flow dynamics $\textbf{x}$. Then we denote $TTT_{\textbf{h}}=min\{TTT(\textbf{x})|x \in \Omega_{\textbf{h}}\}$ and  define the minimum total travel time and optimal headway  as follows:

\begin{itemize}
    \item Minimum total travel time (minimum TTT). $TTT^{*}$ is the minimum total travel time if  $TTT^{*}=min\{TTT_{\textbf{h}}\,|\,\textit{h}_{i,j;k}^{min} 
    \leq \textit{h}_{i,j}(k) \leq
\textit{h}_{i,j;k}^{max}; (i,j) \in E\setminus(L_R\cup L_S); 1\leq k \leq N\}$.
    \item Optimal headway. $\textbf{h}$ is the optimal headway if $TTT_{\textbf{h}}=TTT^{*}$; $\textbf{h}$ and $\textbf{x}$ is the optimal solution of SO-DTA problem if  $\textbf{x} \in \{\textbf{x}\,|\,TTT(\textbf{x})=TTT_{\textbf{h}}=TTT^{*};x \in \Omega_{\textbf{h}}\}$.  
\end{itemize}
\end{definition}

The proposed SO-DTA problem is nonlinear programming due to Equations \ref{eq:FD_dis} and Equations \ref{eq:shockwave_dis1} to   \ref{eq:shockwave_dis2}. However, if we set the headway variable $\textbf{h}$ as an exogenous input, the problem would become a linear programming (LP) in \ref{app:MILP}. 
The corresponding sensitivity-based solution algorithm is presented in \ref{app:SA}. 

However, for the SO-DTA problem, the sensitivity analysis approach has a large computation cost and it is hard to guarantee the convergence to the global optimal solution. Proposition \ref{prop:minimum} states that the minimum TTT can be achieved under the minimum headway setting, providing us with a direct way to reach the state of SO-DTA. On top of that, Proposition \ref{prop:nonuni} states that it is possible to have multiple AV headway settings that yield the same SO-DTA conditions.


\begin{proposition}[Minimum headway achieves SO-DTA]
Supposing that $\textbf{h}^{min}=\{h_{i,j}(k)=h_{i,j;k}^{min}~|~(i,j) ~\in E\setminus(L_R\cup L_S); 1\leq k \leq N\}$, for $\forall \textbf{h}$ satisfies Equation \ref{ineq:headway_dis}, we have 
\begin{equation}
TTT_{\textbf{h}^{min}}=TTT^{*} \leq TTT_{\textbf{h}}.
\nonumber
\end{equation}
\label{prop:minimum}
\end{proposition}

\begin{proofsketch}
For any feasible solution $\textbf{x}$ and $\textbf{h}$ of the proposed SO-DTA formulation, we only change a smaller $\textit{h}_{i,j}(k)$ in $\textbf{h}$ and keep all other headway in $\textbf{h}$ unchanged. The new headway is denoted as $\textbf{h}^{*}$. Then we prove that we could always find a feasible $\textbf{x}^{*}$ under $\textbf{h}^{*}$ and the total travel time of $\textbf{x}^{*}$ is less than or equal to the total travel time of $\textbf{x}$
under original headway $\textbf{h}$. When we traverse all links and time intervals, it is proved that SO-DTA could be achieved under the minimum headway.
Details are presented in \ref{app:minimum}.
\end{proofsketch}

\begin{proposition}[Non-uniqueness]

\begin{eqnarray}
& \forall \textbf{h} \text{ s.t. } TTT_{\textbf{h}}=TTT^{*}, if~ \exists~ \textbf{x} \in \Omega_{\textbf{h}} \text{ s.t. } TTT(\textbf{x})=TTT_{\textbf{h}}=TTT^{*} ~and~ \exists ~(i,j) \in E\setminus(L_R\cup L_S)~and~k \in [1,N] 
\nonumber\\
& \text{ s.t. }  \textit{v}_{i,j}^{f}{~} \rho_{i,j}(k)< \frac{1-\rho_{i,j}(k)\emph{L}}{\textit{h}_{i,j}(k)},~then~ \exists~ \textbf{h}^{1} \neq \textbf{h}  \text{ s.t. } \textbf{x} \in \Omega_{\textbf{h}^{1}} ~and~ TTT_{\textbf{h}^{1}}=TTT(\textbf{x})=TTT_{\textbf{h}}=TTT^{*}.
\nonumber
\end{eqnarray}

\label{prop:nonuni}
\end{proposition}
\begin{proof}
See \ref{app:nonuniq}.
\end{proof}

Proposition~\ref{prop:minimum} aligns with Observation~\ref{ob:1}, and Proposition~\ref{prop:nonuni} proves Observation~\ref{ob:nonunique}. Then we have two natural questions arise:

\begin{itemize}
    \item If we decide to control AVs with minimum headway in reality, would it cause other concerns ({\em e.g.}. safety)?
    \item Is there a ``best'' headway among all the optimal headway given the non-uniqueness of SO-DTA solutions?
\end{itemize}

Both questions motivate us to search for a better headway control setting that still achieves the SO-DTA, which is presented in the following section.

\section{Maximin headway control}
\label{sec_Optimal_maximin_Headway_Formulation}


In this section, we define the maximin headway as the ``largest'' headway setting among all the optimal headway settings that achieve the SO-DTA.
The problem of maximin headway control is formulated and solved analytically. 



\subsection{Maximin headway control formulation}
\label{subsec:Optimal_maximin_Headway_Definition_and_Formulation}

We first rigorously define the maximin headway in Definition~\ref{def:optimal}, where $\textbf{h}^{*}$ is a vector consisting of the time-dependent headway in different links. For the purpose of further converting the proposed SO-DTA problem into a LP, we use the 1-norm to represent the largest headway. The definition of maximin headway could also be extended by using other norms for future study.

\begin{definition}[Maximin Headway]
$\textbf{h}^{*}$ is called the maximin headway if it satisfies the following two conditions:

\begin{itemize} 
\item SO-DTA: $TTT_{\textbf{h}^{*}} = TTT^{*}$, meaning that $\textbf{h}^{*}$ is the optimal headway for SO-DTA.
\item Largest headway:
$\forall~\textbf{h}$ \textit{s.t.} $TTT_{\textbf{h}} = TTT^{*}$, we have $||~ \textbf{h}~ ||_{1} \leq ||~ \textbf{h}^{*} ||_{1}$, where $||\cdot||_{1}$ is the 1-norm.
\end{itemize}
\label{def:optimal}
\end{definition}

Definition~\ref{def:optimal} aligns with Definition~\ref{def:maximin} with more rigor and mathematical details. We formulate a nonlinear optimization to search for the maximin headway in Equation \ref{eq:maximin}, where $TTT^{*}=TTT_{\textbf{h}^{min}}$ represents the optimal value of $TTT$ under minimum headway as denoted in Proposition \ref{prop:minimum}.

\begin{equation}
\begin{aligned}
\max\limits_{\textbf{h}} \quad & \mathbf{1}^T~\textbf{h}\\
\textrm{s.t.} \quad & TTT = TTT^{*} \\
& \textit{Headway-dependent fundamental diagram constraints}\, (\ref{eq:FD_dis}); (\ref{eq:shockwave_dis1})-(\ref{eq:shockwave_dis2}) \\
&  \textit{Headway-dependent double queue constraints}\, (\ref{eq:dq_dis})-(\ref{eq:ineq_2_dis})   \\
&  \textit{Flow conservation constraints}\, 
(\ref{eq:flow_conser_dis}) \\
&  \textit{Non-negative Constraints}\, 
(\ref{ineq:nonnega_dis}) \\
&  \textit{Initial condition constraints}\, 
(\ref{eq:init_dis}) \\
&  \textit{Initial condition constraints}\, 
(\ref{ineq:end_dis}) \\
&  \textit{Other constraints}\,
(\ref{eq:demand_dis})-(\ref{ineq:headway_dis})
\end{aligned}
\label{eq:maximin}
\end{equation}


We note it is a direct extension to replace $\mathbf{1}^T$ with a weight vector $\mathbf{w}^T$ to represent the importance of headway on different roads. Proposition \ref{prop:unique} discusses the condition of the uniqueness of maximin headway of Formulation \ref{eq:maximin}.


\begin{proposition}[Uniqueness of maximin headway]
Suppose $\Phi =  \mathop{\cup}\limits_{\textbf{h}}\Omega_{\textbf{h}}$ is the set of feasible traffic dynamics $\textbf{x}$,
\begin{eqnarray}
& If ~\exists ~\textbf{x}^{*} \in \Phi ~\textit{s.t.}~ TTT(\textbf{x}^{*})<TTT(\textbf{x}), ~for~\forall \textbf{x}\in \Phi ~and~ \textbf{x} \neq \textbf{x}^{*}, 
\nonumber \\
& then ~\exists ~\textbf{h}^{*} ~satisfying~
TTT_{\textbf{h}^{*}}=TTT(\textbf{x}^{*}), ~we~ have ~
||~ \textbf{h}^{*}~ ||_{1} > ||~ \textbf{h} ||_{1}, ~\forall \textbf{h}~ satisfying~ 
TTT_{\textbf{h}}=TTT(\textbf{x}^{*})~and~ \textbf{h}^{*}\neq\textbf{h}. 
\nonumber \\
& \Longleftrightarrow If ~\textbf{x}^{*} ~with~ TTT(\textbf{x}^{*})=TTT^{*} ~is~ unique, the~ optimal~ solution~ of~ Formulation~\ref{eq:maximin} ~is ~unique.
\nonumber
\end{eqnarray}
\label{prop:unique}
\end{proposition}

\begin{proof}
See \ref{app:unique}.
\end{proof}


\subsection{Maximin headway control solution algorithm}
\label{subsub:analytics}

This section proposes a maximin headway control solving algorithm with the idea of exploring the maximum increase of headway to maintain the states of SO-DTA achieved by minimum headway setting.
Proposition \ref{prop:optimal} provides us a method to derive a ``larger'' optimal headway given the current optimal headway.


\begin{proposition}[Larger headway for SO-DTA]
Suppose $\textbf{h}^{*}$ satisfies Equation \ref{ineq:headway_dis} and $ \textbf{x}^{*} \in \Omega_{\textbf{h}^{*}}$ \textit{s.t.} $TTT(\textbf{x}^{*})=TTT_{\textbf{h}^{*}}=TTT^{*}$. If~ $\textbf{h}$ satisfies following constraints:

\begin{eqnarray}
\rho_{i,j}^{*}(k)\textit{v}_{i,j}^f\textit{h}_{i,j}(k) & \leq& 1-\rho_{i,j}^{*}(k)\textit{L}  \qquad i,j,k \in \left\{(i,j;k)~|~ \textit{v}_{i,j}^{f}{~} \rho_{i,j}^{*}(k)  \leq  \frac{1-\rho_{i,j}^{*}(k)\emph{L}}{\textit{h}_{i,j}^{*}(k)} \right\} \label{eq:new1}\\
L_{i,j}\textit{h}_{i,j}(k) &\geq& \Delta_{t}\textit{L}\textit{n}_{i,j}^{w,*}(k)
\qquad i,j,k \in \left\{(i,j;k)~|~ \textit{v}_{i,j}^{f}{~} \rho_{i,j}^{*}(k)  \leq  \frac{1-\rho_{i,j}^{*}(k)\emph{L}}{\textit{h}_{i,j}^{*}(k)} \right\} \\
 L_{i,j}\textit{h}_{i,j}(k) & \leq &  \Delta_{t}\textit{L}\textit{n}_{i,j}^{w,*}(k)+\Delta_{t}\textit{L} \qquad i,j,k \in \left\{(i,j;k)~|~ \textit{v}_{i,j}^{f}{~} \rho_{i,j}^{*}(k) \leq  \frac{1-\rho_{i,j}^{*}(k)\emph{L}}{\textit{h}_{i,j}^{*}(k)} \right\} \\
 \textit{h}_{i,j}^{*}(k) \leq \textit{h}_{i,j}(k) & \leq&  \textit{h}_{i,j;k}^{max} 
\qquad i,j,k \in \left\{(i,j;k)~|~ \textit{v}_{i,j}^{f}{~} \rho_{i,j}^{*}(k)  \leq  \frac{1-\rho_{i,j}^{*}(k)\emph{L}}{\textit{h}_{i,j}^{*}(k)} \right\} \\
 \textit{h}_{i,j}(k) & =&  \textit{h}_{i,j}^{*}(k) \qquad i,j,k \in \left\{(i,j;k)~|~ \textit{v}_{i,j}^{f}{~} \rho_{i,j}^{*}(k) \geq \frac{1-\rho_{i,j}^{*}(k)\emph{L}}{\textit{h}_{i,j}^{*}(k)} \right\} \label{eq:new2}
\end{eqnarray}

where $\rho_{i,j}^{*}(k)$ and $\textit{n}_{i,j}^{w,*}(k) \in \textbf{x}^{*}$, then we have $\textbf{x}^{*} \in \Omega_{\textbf{h}}$, $TTT_{\textbf{h}}=TTT(\textbf{x}^{*})=TTT^{*}$ and $||~ \textbf{h}^{*}~ ||_{1} < ||~ \textbf{h} ||_{1}$. 

\label{prop:optimal}
\end{proposition}

\begin{proof}
It is obvious that if headway variables $\textbf{h}$ satisfies the Equations \ref{eq:new1} to \ref{eq:new2}, then $\textbf{h}$ and $\textbf{x}^{*}$ would constitute the optimal solution of SO-DTA problem in Equation \ref{eq:obj_dis} to \ref{ineq:headway_dis}, so we have $TTT_{\textbf{h}}=TTT(\textbf{x}^{*})=TTT^{*}$ and $||~ \textbf{h}^{*}~ ||_{1} < ||~ \textbf{h} ||_{1}$.
\end{proof}

Proposition \ref{prop:minimum} and \ref{prop:optimal} inspire us a way to derive the maximin headway control and it mainly consists of the following two steps: 1) derive the flow dynamics $\textbf{x}^{*}$ for SO-DTA under the minimum headway $\textbf{h}^{min}$; 2) explore the ``largest'' headway setting based on $\textbf{x}^{*}$ by remaining the same TTT.


To this end, we propose an algorithm to solve the maximin headway control in Algorithm \ref{algo:optimal} and prove the correctness of Algorithm \ref{algo:optimal} under the condition of the existence of unique optimal flow dynamics in Proposition \ref{prop:correctness}. 
One can see from Algorithm \ref{algo:optimal} that, once the traffic condition $\textbf{x}^{*}$ is determined, we can run a linear programming (LP) to solve for the maximin headway, and hence the solution process is computationally efficient.

\begin{breakablealgorithm}
\caption{Maximin Headway Control Solving Algorithm}
{\textbf{Input:} parameters in Table \ref{notation_con}}.\qquad\qquad\qquad\qquad\qquad\qquad\qquad\qquad~
\begin{algorithmic}[1]
\STATE Solve the discretized SO-DTA formulation (Equations \ref{eq:obj_dis} to \ref{eq:FD_dis}; Equations \ref{eq:shockwave_dis1} to \ref{ineq:headway_dis}) under minimum headway $\textbf{h}^{min}$ to derive the optimal flow dynamics $\textbf{x}^{*}$ with $TTT(\textbf{x}^{*})=TTT_{\textbf{h}^{min}}=TTT^{*}$. To be specific, we can solve Formulation~\ref{eq:milp} to obtain $\textbf{x}^{*}$, and details of the formulation is presented in \ref{app:MILP}.
\begin{eqnarray}
\begin{aligned}
    \textbf{x}^{*} =& \arg\min \limits_{\mathbf{x}} ~~  TTT =\textbf{P}^{T}\textbf{x}\\
    & \, s.t.  \quad \textbf{A}(\textbf{h}^{min})\textbf{x}=\textbf{B}(\textbf{h}^{min}) \\
    & \quad \quad ~~ \textbf{C}(\textbf{h}^{min})\textbf{x}\leq\textbf{D}(\textbf{h}^{min}) \\
\end{aligned}
\label{eq:milp}
\end{eqnarray}
\STATE Solve Formulation~\ref{eq:swap} under the  flow dynamics $\textbf{x}^{*}$ and derive the maximin headway  $\textbf{h}^{*}$, where $\rho_{i,j}^{*}(k)$ and $\textit{n}_{i,j}^{w,*}(k) \in \textbf{x}^{*}$.

\begin{eqnarray}
\max \limits_{\mathbf{h}} ~~  \mathbf{1}^T~\textbf{h} \qquad\qquad \quad &\,&
\label{eq:swap}\\
s.t. \quad  \rho_{i,j}^{*}(k)\textit{v}_{i,j}^f\textit{h}_{i,j}(k) & \leq&  1-\rho_{i,j}^{*}(k)\textit{L}  \qquad i,j,k \in \left\{(i,j;k)~|~ \textit{v}_{i,j}^{f}{~} \rho_{i,j}^{*}(k) \leq \frac{1-\rho_{i,j}^{*}(k)\emph{L}}{\textit{h}_{i,j}^{*}(k)} \right\} \nonumber \\
 L_{i,j}\textit{h}_{i,j}(k) & \geq&  \Delta_{t}\textit{L}\textit{n}_{i,j}^{w,*}(k)
\qquad i,j,k \in \left\{(i,j;k)~|~ \textit{v}_{i,j}^{f}{~} \rho_{i,j}^{*}(k)  \leq  \frac{1-\rho_{i,j}^{*}(k)\emph{L}}{\textit{h}_{i,j}^{*}(k)} \right\} \nonumber \\
 L_{i,j}\textit{h}_{i,j}(k) & \leq&  \Delta_{t}\textit{L}\textit{n}_{i,j}^{w,*}(k)+\Delta_{t}\textit{L} \qquad i,j,k \in \left\{(i,j;k)~|~ \textit{v}_{i,j}^{f}{~} \rho_{i,j}^{*}(k)  \leq  \frac{1-\rho_{i,j}^{*}(k)\emph{L}}{\textit{h}_{i,j}^{*}(k)} \right\} \nonumber \\
 \textit{h}_{i,j}^{*}(k) \leq \textit{h}_{i,j}(k) & \leq&  \textit{h}_{i,j;k}^{max} 
\qquad i,j,k \in \left\{(i,j;k)~|~ \textit{v}_{i,j}^{f}{~} \rho_{i,j}^{*}(k)  \leq  \frac{1-\rho_{i,j}^{*}(k)\emph{L}}{\textit{h}_{i,j}^{*}(k)} \right\} \nonumber \\
 \textit{h}_{i,j}(k) & =&  \textit{h}_{i,j}^{*}(k) \qquad i,j,k \in \left\{(i,j;k)~|~ \textit{v}_{i,j}^{f}{~} \rho_{i,j}^{*}(k) \geq \frac{1-\rho_{i,j}^{*}(k)\emph{L}}{\textit{h}_{i,j}^{*}(k)} \right\} \nonumber
\end{eqnarray}


\end{algorithmic}
{\textbf{Output:} maximin headway $\textbf{h}^{*}$ and optimal flow dynamics $\textbf{x}^{*}$ with $\textbf{x}^{*} \in \Omega_{\textbf{h}^{*}}$ and $TTT_{\textbf{h}^{*}}=TTT(\textbf{x}^{*})=TTT^{*}$}.\quad~~
\label{algo:optimal}
\end{breakablealgorithm}

\begin{proposition}[Correctness of Algorithm~\ref{algo:optimal}]
The output $\textbf{h}^{*}$ of Algorithm~\ref{algo:optimal} is the optimal solution of Formulation~\ref{eq:maximin} if $~\exists ~\textbf{x}^{*} \in \Phi ~\textit{s.t.}~ TTT(\textbf{x}^{*})<TTT(\textbf{x}), \forall \textbf{x}\in \Phi ~and~ \forall  \textbf{x} \neq \textbf{x}^{*}$.
\label{prop:correctness}
\end{proposition}

\begin{proof}
Proposition \ref{prop:correctness} can be directly obtained by combining Proposition \ref{prop:minimum}, \ref{prop:unique}, and \ref{prop:optimal}.
\end{proof}


Moreover, as all the variables in different time intervals are simultaneously obtained by Algorithm \ref{algo:optimal} in an off-line manner, which cannot be used in the real-time control. To address this issue, we further propose an online algorithm to solve the SO-DTA problem incrementally. In each step of the online algorithm, we only obtain the flow dynamic variables and the maximin headway in the current time interval. Details of the online algorithm are presented in \ref{app:online}.

\section{Numerical experiments}
\label{sec:Numerical_Results}

The proposed model and propositions are validated on both a small and a large-scale Hong Kong network in this section. 

\subsection{A small network}

In the small network, we first solve for the optimal headway control and maximin headway control to demonstrate that the maximin headway control achieves the SO-DTA with large headway settings. Additionally, the sensitivity analysis for the maximin headway control is also conducted.


\subsubsection{Settings}

The small network consists of 5 nodes and 6 links, Nodes 1 and 2 are connected to the origins, and Node 5 is connected with the destination.  We further set the total time horizon $\textit{T}$ as 90 minutes and the length of each time interval $\Delta_{t}$ as 5 minutes. It is worth noting that, in simulation-based microscopic models, small time steps (e.g., 5 seconds) could be used  \citep{ben2012dynamic,sumalee2011stochastic,zhang2020path}. However, for the DTA models using dynamic link models, numerous studies opt for a time interval of five minutes or greater \citep{chiu2011dynamic,bellei2005within,janson1991dynamic,chan2021quasi}. Besides, we would conduct the sensitivity analysis for  $\Delta_{t}$ on the proposed maximin headway control framework in \ref{sec:maxmin_small}. The travel demand for each O-D pair is 50 $\textit{vehicles}/\textit{mins}$ and the time horizon of demand $\textit{T}_{1}$ is 40 minutes. Other parameters are illustrated in \ref{app:parameter}.

\begin{figure}[h]
    \centering
    \includegraphics[width=0.6\linewidth]{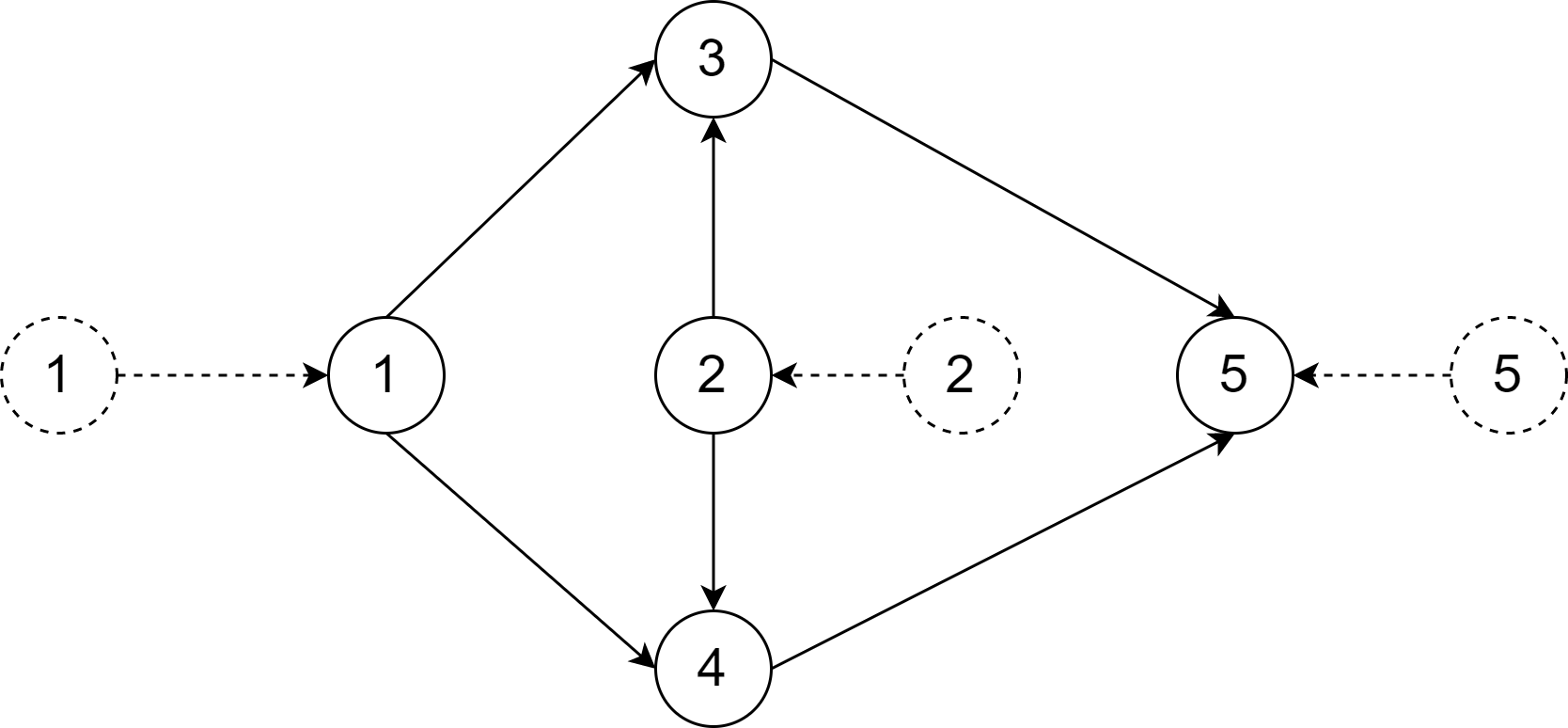}
    \caption{The small network.}
    \label{fig:network}
\end{figure}

\subsubsection{Results for maximin headway control in small network}
\label{sec:maxmin_small}


The three different scenarios are defined for comparison as follows:
\begin{itemize}
\item \textbf{Min-HW}: the headway of each link is directly set as the minimum headway. 
\item \textbf{SO-HW}: we directly solve for the SO-DTA by searching for the optimal headway, and the solution algorithm is based on the sensitivity-based algorithm (Algorithm \ref{algo:sen}) described in \ref{app:SA}.
\item \textbf{Maximin-HW}: we solve the maximin headway by Algorithm~\ref{algo:optimal}.
\end{itemize}

Based on Proposition~\ref{prop:minimum}, we already know that the minimum headway setting would yield the SO-DTA, and the total travel time $TTT(\textbf{h}^{min})$ = 25,210. Additionally, the dynamic inflow and outflow of each link in the SO-DTA are presented in Figure~\ref{fig:Flow_in_links}. The path-based and link-based formulations are equivalent to SO-DTA. Utilizing the path-link incidence matrix or the Dynamic Assignment Ratio (DAR) matrix, we can obtain the path flow under SO-DTA from the optimal link flow as determined by Algorithm \ref{algo:optimal}.






\begin{figure}[h]
    \centering
    \includegraphics[width=0.8\linewidth]{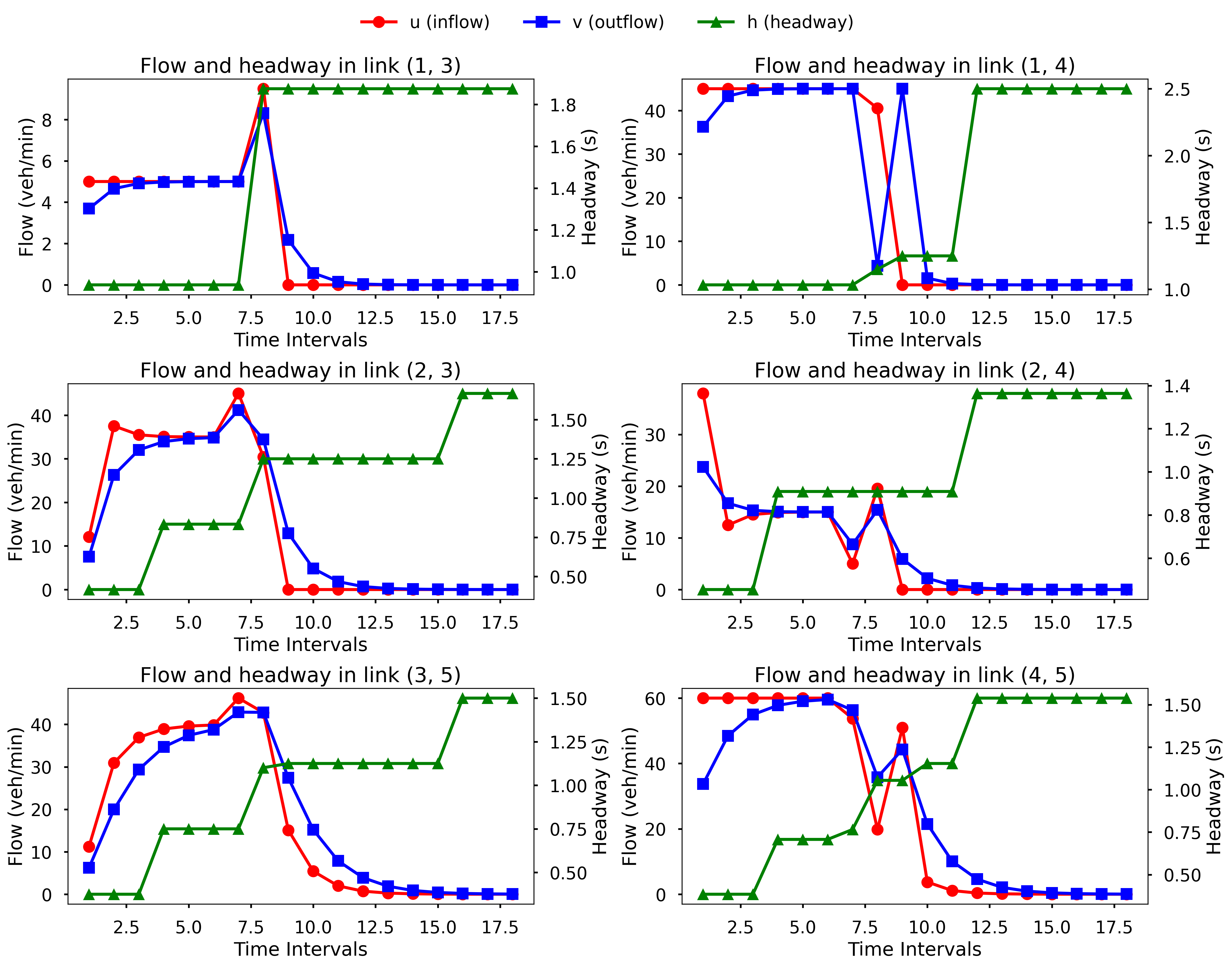}
    \caption{Inflow, outflow, and maximin headway of each link under SO-DTA.}
    \label{fig:Flow_in_links}
\end{figure}

The convergence curve of Algorithm~\ref{algo:sen} in \textbf{SO-HW} is shown in Figure~\ref{fig:FTTT_algorithm2}, where the convergence criterion is that the difference of TTT in the latest two iterations is less than 1 minute. The sensitivity-based Algorithm \ref{algo:sen} may converge to the local optimal and it does not guarantee to find the global optimal solution by Algorithm \ref{algo:sen}. However, it is proved that our proposed Algorithm \ref{algo:optimal} could analytically solve the SO-DTA problem to find the global optimal solution. Therefore, one can see that the sensitivity-based algorithm converges to $TTT$= 31,342 minutes, which is larger than 25,210 derived by Algorithm \ref{algo:optimal} in \textbf{Maximin-HW}. This verifies that the sensitivity-based method might not converge to the optimal solution of SO-DTA and demonstrates the merits of the proposed analytical algorithm in Algorithm~\ref{algo:optimal}.


\begin{figure}[h]
    \centering
    \includegraphics[width=0.7\linewidth]{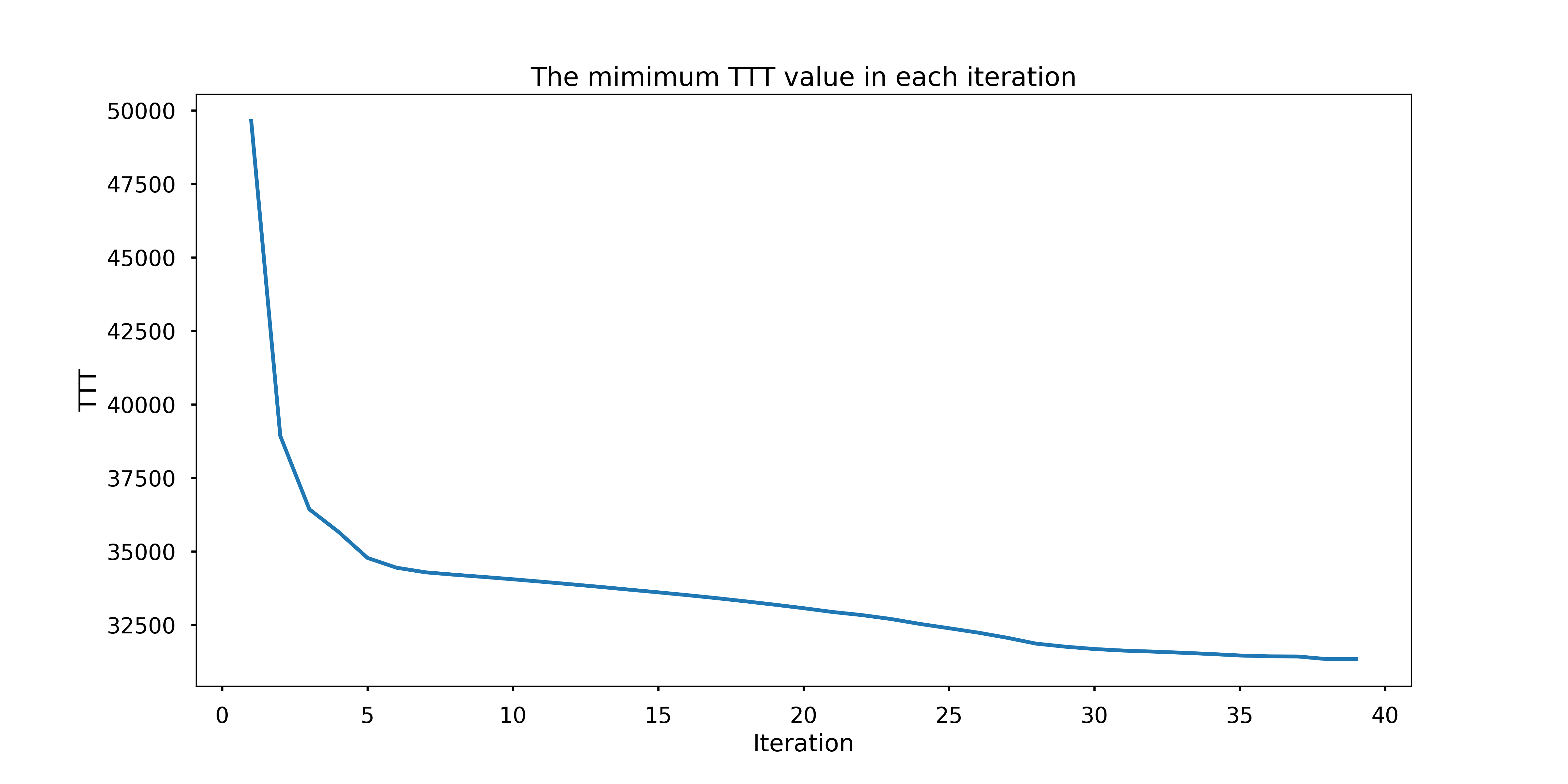}
    \caption{Convergence curve in terms of TTT by Algorithm \ref{algo:sen}.}
    \label{fig:FTTT_algorithm2}
\end{figure}

As for the \textbf{Maximin-HW}, the TTT remains 25,210 minutes, meaning that the maximin headway still achieves the SO-DTA. Besides, to measure the improvement of maximin headway compared with minimum headway when considering the safety aspect, we define the ratio of maximin headway $\bf{rmh}$ in Equation \ref{eq:ratio}, where $\textbf{h}^{*}=\{h_{i,j}^{*}(k)~|~(i,j) ~\in E\setminus(L_R\cup L_S); 1\leq k \leq N\}$ is the vector of maximin headway solved by Algorithm \ref{algo:optimal} and $\textbf{h}^{min}=\{h_{i,j;k}^{min}~|~(i,j) ~\in E\setminus(L_R\cup L_S); 1\leq k \leq N\}$ is the vector of minimum headway.

\begin{equation}
    \bf{rmh} = \frac{||\textbf{h}^{*}||_{1}}{||\textbf{h}^{min}||_{1}} 
    \label{eq:ratio}
\end{equation}

The ratio of maximin headway $\bf{rmh}$ in the small network is 1.43, which indicates that the maximin headway is about 1.43 times as large as the minimum headway on average, demonstrating the merits of the maximin headway control. Table \ref{table2} presents the average maximin headway and minimum headway of each link over time. One can see that the gap between the maximin and minimum headway is significant, and the gap could also reflect the safety margin of each link. The time-dependent maximin headway is illustrated in Figure  \ref{fig:Flow_in_links} and it could infer a rough relationship between the optimal flow dynamics and maximin headway. The link flow under SO-DTA is usually large when the maximin headway is small and vice versa. This result aligns with the real-world traffic phenomena. When the flow in a road or highway is small, it is reasonable to set a relatively large headway.

\begin{table}[H]
\centering
\renewcommand\arraystretch{1.5}
\tabcolsep=0.7cm
\begin{tabular}{ccc}
		\hline
		 $\textbf{Link}$ & $\textbf{Average~Minimum~Headway}$ &  $\textbf{Average~Maximin~Headway}$ \\
		\hline
		1 $\to$ 3 & 0.969s  & 1.510s \\
	  1 $\to$ 4  & 0.997s  & 1.646s   \\
        2 $\to$ 3  & 0.783s  & 1.088s   \\
        2 $\to$ 4  & 0.742s  & 1.088s   \\
        3 $\to$ 5 & 0.753s  & 0.978s   \\
        4 $\to$ 5 & 0.828s  & 1.068s  \\
		\hline
	\end{tabular}
	\caption{Maximin and minimum headway in the small network.}
    \label{table2}
\end{table}




\subsubsection{Sensitivity analysis for the maximin headway control}

This section explores how different network configurations would affect the ratio of maximin headway. In particular, we study the sensitivity regarding the length of time interval, travel demand, and minimum headway.

{\bf Time interval}. We first explore the influence of the length of time interval on the ratio of maximin headway. The length of time interval $\Delta_{t}$ varies from 2 to 12 minutes. 
Figure \ref{fig:time_interval} illustrates that the ratio of maximin headway remains stable when the length of time interval increases, implying that the current length of time interval is sufficient for the case study.

\begin{figure}[h]
    \centering
    \includegraphics[width=0.7\linewidth]{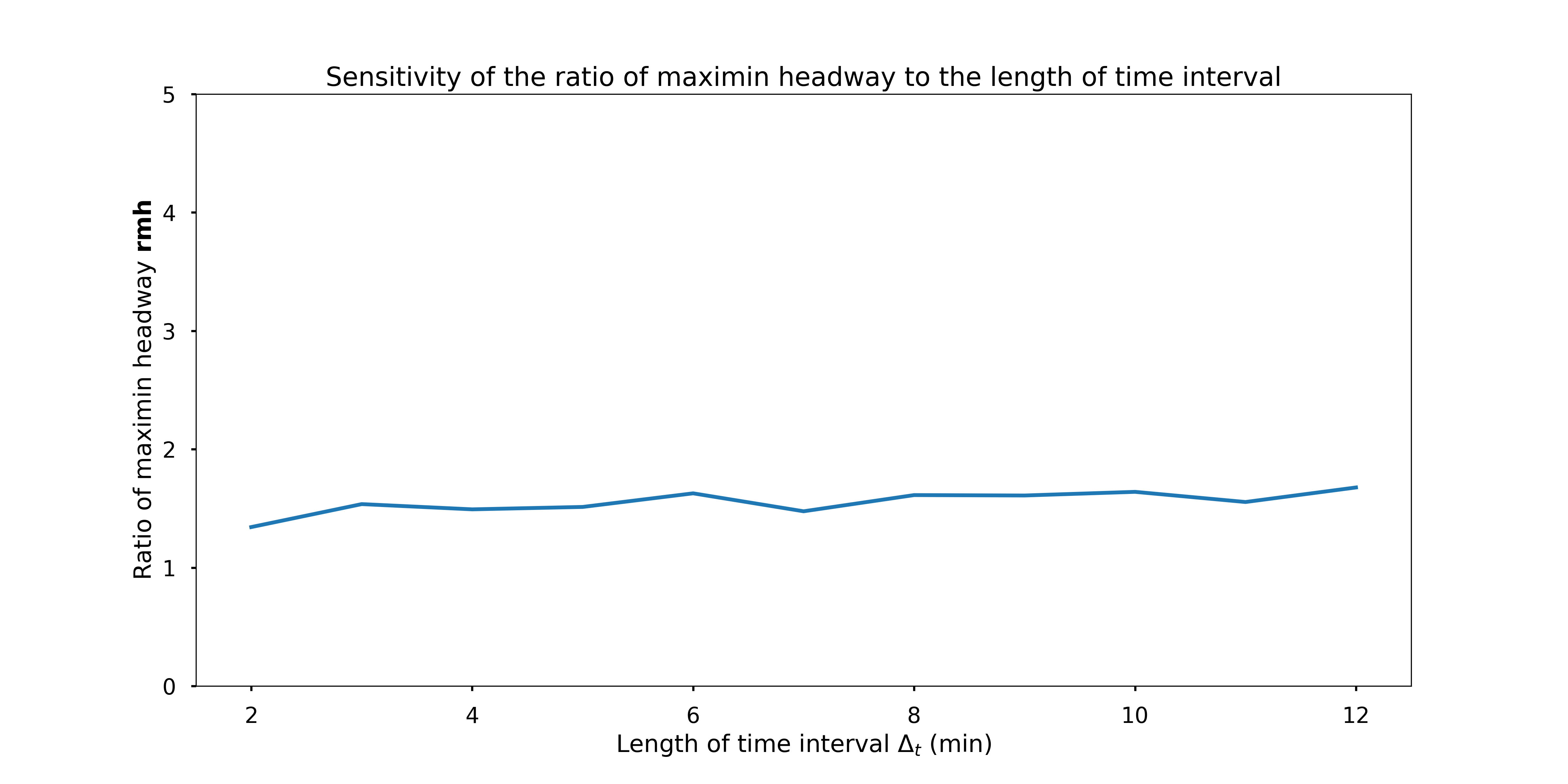}
    \caption{The ratio of maximin headway under different lengths of time intervals.}
    \label{fig:time_interval}
\end{figure}

{\bf Travel demand.} We conduct the sensitivity analysis regarding travel demand under different minimum headway requirements. We set a constant minimum headway for all links and vary the demand of each O-D pair from 30 to 70 veh/min under different minimum headway requirements, then we solve  Algorithm \ref{algo:optimal} and derive the ratio of maximin headway, and the results are shown in  Figure \ref{fig:Gap_demand}. 

\begin{figure}[h]
    \centering
    \includegraphics[width=0.7\linewidth]{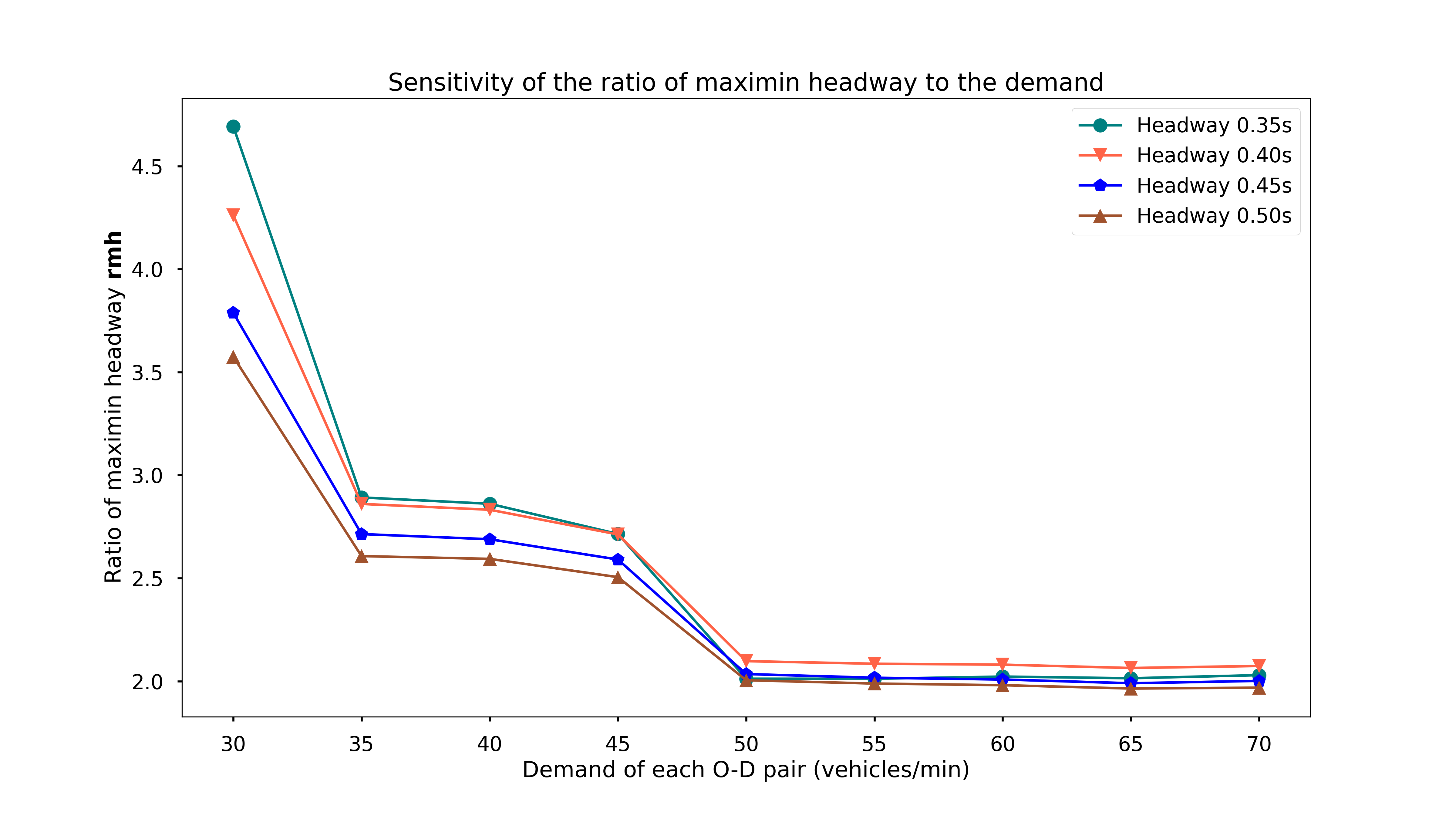}
    \caption{The ratio of maximin headway under different traffic demands.}
    \label{fig:Gap_demand}
\end{figure}

One can see that the ratio of maximin headway decreases when the demand increases, indicating that the safety margin for SO-DTA will become marginal when more people travel.



{\bf Minimum headway.} 
Then the sensitivity analysis for the minimum headway value is conducted. The minimum headway is constant for all links and we vary the value of minimum headway from 0.20s to 1.10s under different O-D demands and keep other parameters unchanged. Then we solve the Algorithm \ref{algo:optimal} and derive the ratio of maximin headway under different minimum headway values. Figure \ref{fig:Gap_minimum_headway} shows how the the ratio of maximin headway would change corresponding to different minimum headway values. 

\begin{figure}[H]
    \centering
    \includegraphics[width=0.7\linewidth]{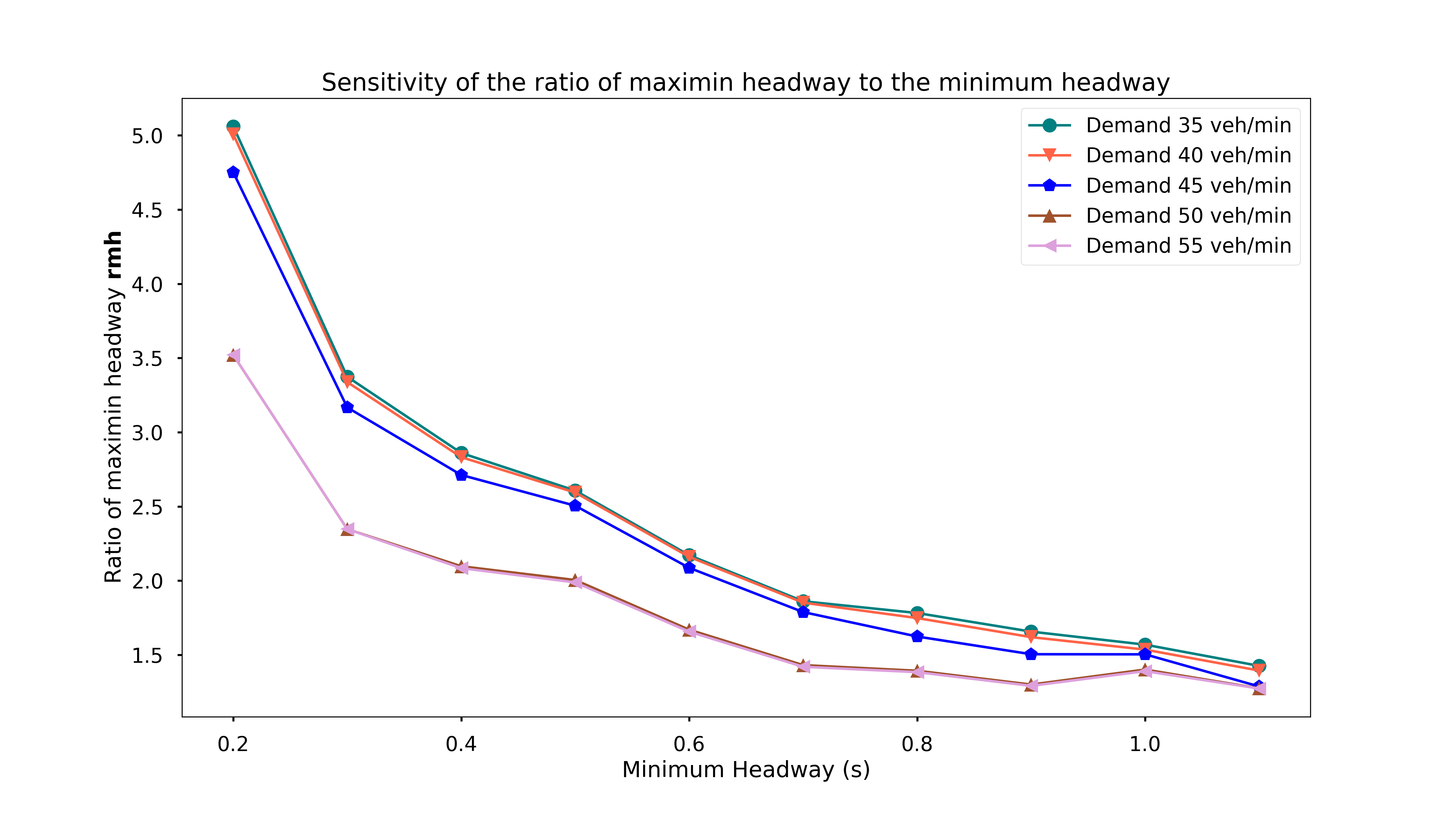}
    \caption{The ratio of maximin headway under different minimum headway settings.}
    \label{fig:Gap_minimum_headway}
\end{figure}

One can see that in general, the ratio of maximin headway is relatively stabilized under different minimum headway, indicating that there is always room to enlarge the headway after achieving SO-DTA under different safety requirements ({\em i.e}, minimum headway).

\subsection{The Hong Kong network}
\label{sec:hk_set}
In this section, we explore the effectiveness of the maximin headway control in real traffic conditions in the Hong Kong network.
\subsubsection{Settings}
We study the urban area of Kowloon Peninsula in Hong Kong to capture the real traffic conditions under headway control. The Hong Kong network consists of 55 nodes and 192 links, and the number of O-D pairs is 72, as referred to Figure \ref{fig:hk_network}. The total time horizon is 250 minutes and the length of time interval is set as 5 minutes. The minimum and maximum headway is set as a constant 0.5 seconds and 2.5 seconds, respectively. 
The link parameters and hour OD demand profile for the Hong Kong network are illustrated in \ref{app:hk}.

\begin{figure}[h]
    \centering
    \includegraphics[width=0.5\linewidth]{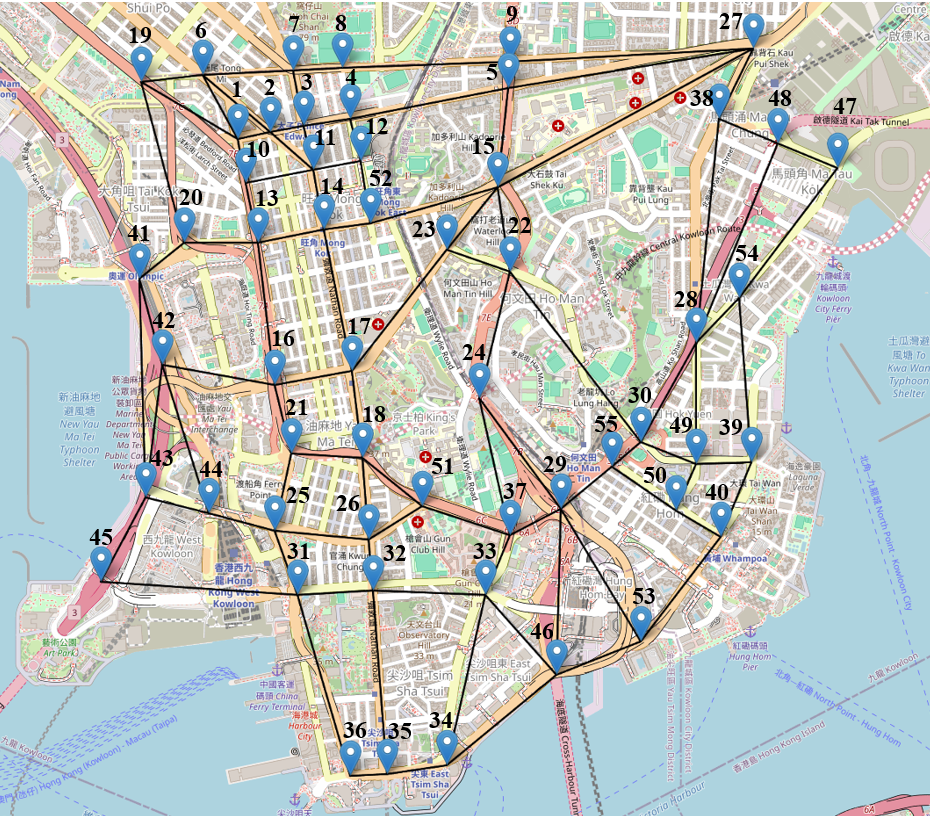}
    \caption{Overview of the Hong Kong network.}
    \label{fig:hk_network}
\end{figure}


\subsubsection{Experimental results}

The results of \textbf{Min-HW} and \textbf{Maximin-HW} in Hong Kong network are illustrated as follows. Besides, we also present the high computation efficiency of Algorithm \ref{algo:optimal} in this large-scale network.

\begin{itemize}
\item \textbf{Min-HW}: we directly set the minimum headway for each link, and we have  $TTT(\textbf{h}_{min})$ = 639,138 minutes.
\item \textbf{Maximin-HW}: The $TTT$ is the same as that of the \textbf{Min-HW}, and the ratio of maximin headway $\bf{rmh}$ is 4.15, which indicates that the maximin headway control could significantly improve the minimum headway and also keeps SO-DTA.
\end{itemize}

In terms of the computation efficiency, the computation time of Algorithm \ref{algo:optimal} is about 2,802 seconds to solve the maximin headway. However, the sensitivity-based Algorithm \ref{algo:sen} fails to solve the maximin headway as it requires impractical memory usage to compute the generalized inverse matrix $\textbf{Q}^{-1}$ defined in Equation \ref{matrix:2}. Therefore, our proposed framework of maximin control reveals its applicability in large-scale networks. 

Then we calculate the average value of maximin headway and link flow under SO-DTA for all links and time intervals to explore the network-wide effects of maximin headway in the large-scale network. Figures \ref{fig:network_headway} and \ref{fig:network_flow} present the network-wide average maximin headway and the average link flow under maximin headway control, where we calculate the average value of time-dependent maximin headway, inflow, and outflow for all links and time intervals. 
One can observe a negative correlation between maximin headway and link flow, i.e., a large maximin headway flow in other links always corresponds to a small link flow.



\begin{figure}[h]
  \centering
  \begin{subfigure}[b]{0.45\textwidth}
    \includegraphics[width=\textwidth]{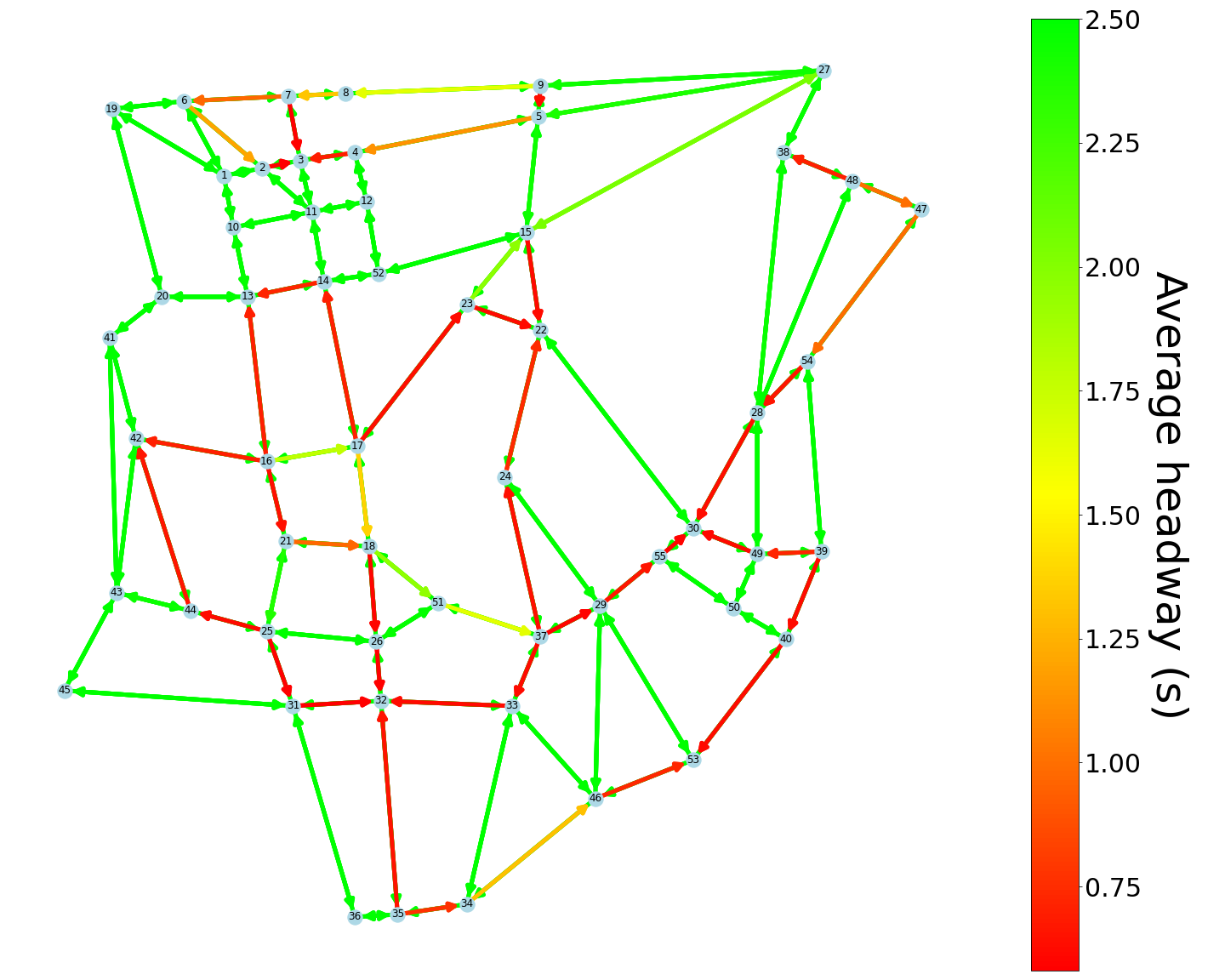}
    \caption{The average maximin headway.}
    \label{fig:network_headway}
  \end{subfigure}
  \hfill
  \begin{subfigure}[b]{0.45\textwidth}
    \includegraphics[width=\textwidth]{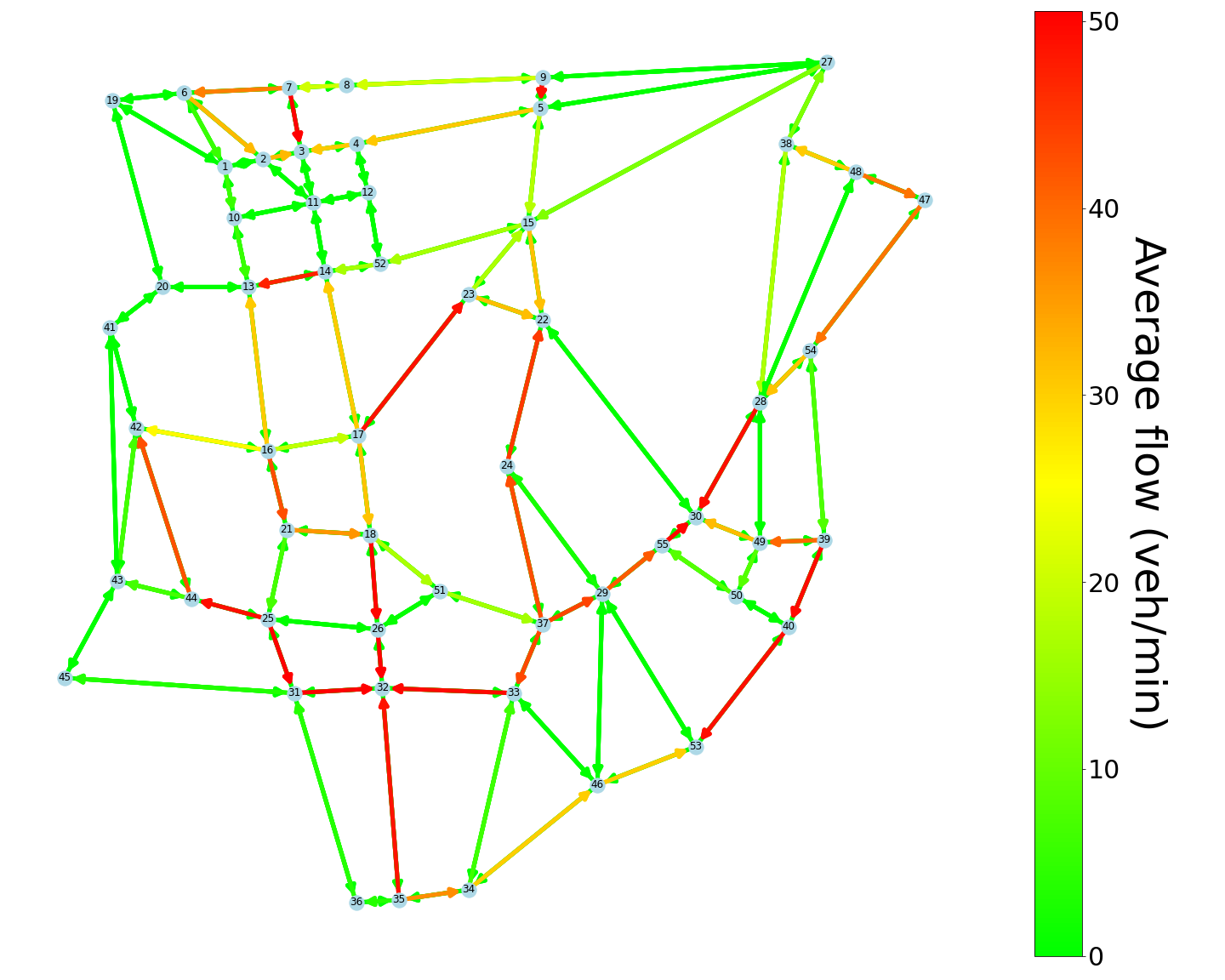}
    \caption{The average link flow.}
    \label{fig:network_flow}
  \end{subfigure}
  \caption{Overview of link flow and maximin headway of each link in Hong Kong network.}
\label{fig:network_plot}
\end{figure}

To this end, we plot the maximin headway against the link flow for all time intervals and links in Figure \ref{fig:distribution}, where we set a homogeneous minimum headway as 0.5s and maximum headway and 2.5s for all time intervals and links as presented in Section \ref{sec:hk_set}. 
Due to different link parameters and queue capacity restrictions, the detailed free-flow and congestion states of each link vary. 
The maximin headway control framework is more effective in the free-flow states when the link flow is relatively small, where the maximin headway could be approximating the maximum headway under SO-DTA. 
However, when the link flow increases, the maximin headway would decrease, but it is still larger than the minimum headway, representing the merits of maximin headway control in ensuring a higher safety level under SO-DTA. 
When the link flow is very high, the maximin headway would equal the minimum headway, indicating that there is no room to improve safety that still achieves SO-DTA. Overall, the maximin headway is highly dependent on the link flow under SO-DTA and, in general,  maximin headway is higher under less congested links.

\begin{figure}[h]
    \centering
    \includegraphics[width=0.8\linewidth]{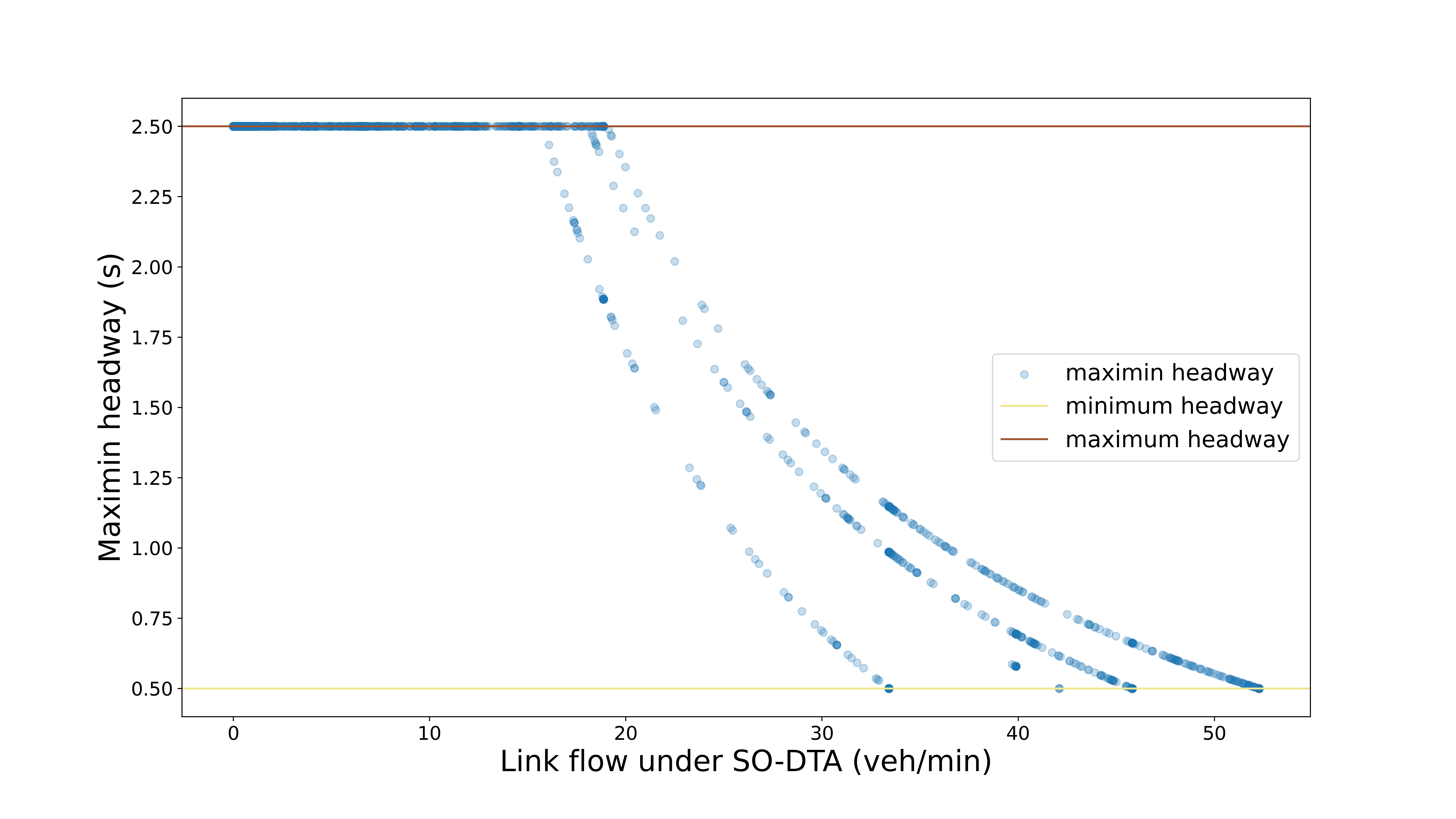}
    \caption{The distribution of maximin headway and link flow under SO-DTA.}
    \label{fig:distribution}
\end{figure}

Additionally, we also run the the online algorithm in \ref{app:online}, and the length of time interval $\Delta_{t}$ is set as 5 minutes. The total travel time $TTT$ is 660,022 and the ratio of maximin headway $\bf{rmh}$ is 4.11. The average time of conducting the online algorithm to obtain the real-time headway control in each iteration is less than 1 second on a personal computer. This demonstrates the effectiveness and efficiency of the proposed method in the online settings.

\section{Conclusion}
\label{sec:Conclusion}

In this paper, we propose a maximin headway control framework for SO-DTA on general networks in a fully automated environment. 
Both HFD and HDQ are developed to capture the effects of dynamic AV headway control. Particularly, the HDQ is proved to be a generalized version of the DQ. 
It is rigorously proved that the minimum headway could achieve SO-DTA. On top of this, the maximin headway of AVs is innovatively defined as the ``largest'' headway setting that still achieves SO-DTA. 
We further propose an efficient solution algorithm to derive the maximin headway. Numerical experiments demonstrate the correctness and effectiveness of the proposed framework. 
Sensitivity analysis demonstrates that the increase in travel demand would normally correspond to a decrease in gap values between the maximin and minimum headway, which helps the policymakers understand the safety margin of AVs on different links. We believe this study provides novel angles and mathematical techniques to derive and analyze the ``desired solution'' among the non-unique solutions in SO-DTA.


Future studies can be conducted to consider mixed traffic conditions, heterogeneous travel behaviors, and stochastic traffic demand and road supply. 
For example, this study is based on the HFD and HDQ in a fully automated environment so that we could control the headway of all AVs for SO-DTA. However, in mixed traffic conditions, the route choice of human-driving vehicles (HVs) may follow dynamic user equilibrium (DUE), making the problem more complex \citep{ZHANG201875}. It is also interesting to explore the headway control on a limited number of roads or with limited control capacities \citep{battifarano2023impact}. Additionally, the enabling technologies and environmental impacts of dynamic headway control are also interesting topics to be studied.

\section*{Acknowledgments}
The work described in this paper was supported by the National Natural Science Foundation of China (No. 52102385),  grants from the Research Grants Council of the Hong Kong Special Administrative Region, China (Project No. PolyU/25209221 and PolyU/15206322), and grants from the Otto Poon Charitable Foundation Smart Cities Research Institute (SCRI) at the Hong Kong Polytechnic University (Project No. P0043552 \& P0036472).

\bibliography{ref}

@article{GUO202087,
title = {Macroscopic fundamental diagram based perimeter control considering dynamic user equilibrium},
journal = {Transportation Research Part B: Methodological},
volume = {136},
pages = {87-109},
year = {2020},
issn = {0191-2615},
doi = {https://doi.org/10.1016/j.trb.2020.03.004},
author = {Qiangqiang Guo and Xuegang (Jeff) Ban}
}

@article{MA201498,
title = {Continuous-time dynamic system optimum for single-destination traffic networks with queue spillbacks},
journal = {Transportation Research Part B: Methodological},
volume = {68},
pages = {98-122},
year = {2014},
issn = {0191-2615},
doi = {https://doi.org/10.1016/j.trb.2014.06.003},
author = {Rui Ma and Xuegang (Jeff) Ban and Jong-Shi Pang}
}

@article{SHI2021279,
title = {Constructing a fundamental diagram for traffic flow with automated vehicles: Methodology and demonstration},
journal = {Transportation Research Part B: Methodological},
volume = {150},
pages = {279-292},
year = {2021},
issn = {0191-2615},
doi = {https://doi.org/10.1016/j.trb.2021.06.011},
author = {Xiaowei Shi and Xiaopeng Li}}

@article{ZHOU2020102614,
title = {Modeling the fundamental diagram of mixed human-driven and connected automated vehicles},
journal = {Transportation Research Part C: Emerging Technologies},
volume = {115},
pages = {102614},
year = {2020},
issn = {0968-090X},
doi = {https://doi.org/10.1016/j.trc.2020.102614},
author = {Jiazu Zhou and Feng Zhu}
}

@article{GONG2016314,
title = {Constrained optimization and distributed computation based car following control of a connected and autonomous vehicle platoon},
journal = {Transportation Research Part B: Methodological},
volume = {94},
pages = {314-334},
year = {2016},
issn = {0191-2615},
doi = {https://doi.org/10.1016/j.trb.2016.09.016},
author = {Siyuan Gong and Jinglai Shen and Lili Du}
}

@article{YU2021103101,
title = {Automated vehicle-involved traffic flow studies: A survey of assumptions, models, speculations, and perspectives},
journal = {Transportation Research Part C: Emerging Technologies},
volume = {127},
pages = {103101},
year = {2021},
issn = {0968-090X},
doi = {https://doi.org/10.1016/j.trc.2021.103101},
author = {Haiyang Yu and Rui Jiang and Zhengbing He and Zuduo Zheng and Li Li and Runkun Liu and Xiqun Chen}
}

@inproceedings{hatipoglu1996longitudinal,
  title={Longitudinal headway control of autonomous vehicles},
  author={Hatipoglu, Cem and Ozguner, U and Sommerville, Martin},
  booktitle={Proceeding of the 1996 IEEE International Conference on Control Applications IEEE International Conference on Control Applications held together with IEEE International Symposium on Intelligent Contro},
  pages={721--726},
  year={1996},
  organization={IEEE}
}

@article{LTM,
author = {Yperman, Isaak},
year = {2007},
month = {01},
pages = {},
title = {The Link Transmission Model for dynamic network loading}
}

@article{DAGANZO199579,
title = {The cell transmission model, part II: Network traffic},
journal = {Transportation Research Part B: Methodological},
volume = {29},
number = {2},
pages = {79-93},
year = {1995},
issn = {0191-2615},
doi = {https://doi.org/10.1016/0191-2615(94)00022-R},
author = {Carlos F. Daganzo}
}

@article{SHEN20141,
title = {System optimal dynamic traffic assignment: Properties and solution procedures in the case of a many-to-one network},
journal = {Transportation Research Part B: Methodological},
volume = {65},
pages = {1-17},
year = {2014},
issn = {0191-2615},
doi = {https://doi.org/10.1016/j.trb.2014.02.002},
author = {Wei Shen and H.M. Zhang}
}

@article{CHEN2016143,
title = {Optimal deployment of autonomous vehicle lanes with endogenous market penetration},
journal = {Transportation Research Part C: Emerging Technologies},
volume = {72},
pages = {143-156},
year = {2016},
issn = {0968-090X},
doi = {https://doi.org/10.1016/j.trc.2016.09.013},
author = {Zhibin Chen and Fang He and Lihui Zhang and Yafeng Yin}
}

@article{Bagloee,
author = {Bagloee, Saeed and Sarvi, Majid and Patriksson, Michael and Rajabifard, Abbas},
year = {2017},
month = {04},
pages = {},
title = {A Mixed User-Equilibrium and System-Optimal Traffic Flow for Connected Vehicles Stated as a Complementarity Problem: Mixed user-equilibrium and system-optimal traffic flow},
volume = {32},
journal = {Computer-Aided Civil and Infrastructure Engineering},
doi = {10.1111/mice.12261}
}

@article{LEVIN2016103,
title = {A multiclass cell transmission model for shared human and autonomous vehicle roads},
journal = {Transportation Research Part C: Emerging Technologies},
volume = {62},
pages = {103-116},
year = {2016},
issn = {0968-090X},
doi = {https://doi.org/10.1016/j.trc.2015.10.005},
author = {Michael W. Levin and Stephen D. Boyles}
}

@article{CHEN201744,
title = {Optimal design of autonomous vehicle zones in transportation networks},
journal = {Transportation Research Part B: Methodological},
volume = {99},
pages = {44-61},
year = {2017},
issn = {0191-2615},
doi = {https://doi.org/10.1016/j.trb.2016.12.021},
author = {Zhibin Chen and Fang He and Yafeng Yin and Yuchuan Du}
}

@article{ZHANG201875,
title = {Mitigating the impact of selfish routing: An optimal-ratio control scheme (ORCS) inspired by autonomous driving},
journal = {Transportation Research Part C: Emerging Technologies},
volume = {87},
pages = {75-90},
year = {2018},
issn = {0968-090X},
doi = {https://doi.org/10.1016/j.trc.2017.12.011},
author = {Kenan Zhang and Yu (Marco) Nie}
}

@article{WANG2019139,
title = {Multiclass traffic assignment model for mixed traffic flow of human-driven vehicles and connected and autonomous vehicles},
journal = {Transportation Research Part B: Methodological},
volume = {126},
pages = {139-168},
year = {2019},
issn = {0191-2615},
doi = {https://doi.org/10.1016/j.trb.2019.05.022},
author = {Jian Wang and Srinivas Peeta and Xiaozheng He}
}

@article{WANG2020227,
title = {A mixed behaviour equilibrium model with mode choice and its application to the endogenous demand of automated vehicles},
journal = {Journal of Management Science and Engineering},
volume = {5},
number = {4},
pages = {227-248},
year = {2020},
note = {Special Issue on Traffic and Transportation in the Era of Big Data and Shared Economy},
issn = {2096-2320},
doi = {https://doi.org/10.1016/j.jmse.2020.05.003},
author = {Guangchao Wang and Hang Qi and Huiling Xu and Seungkyu Ryu}
}

@ARTICLE{9408374,
  author={Seo, Toru and Asakura, Yasuo},
  journal={IEEE Transactions on Intelligent Transportation Systems}, 
  title={Multi-Objective Linear Optimization Problem for Strategic Planning of Shared Autonomous Vehicle Operation and Infrastructure Design}, 
  year={2022},
  volume={23},
  number={4},
  pages={3816-3828},
  doi={10.1109/TITS.2021.3071512}}

@article{NGODUY202156,
title = {Multiclass dynamic system optimum solution for mixed traffic of human-driven and automated vehicles considering physical queues},
journal = {Transportation Research Part B: Methodological},
volume = {145},
pages = {56-79},
year = {2021},
issn = {0191-2615},
doi = {https://doi.org/10.1016/j.trb.2020.12.008},
author = {Dong Ngoduy and N.H. Hoang and H.L. Vu and D. Watling}
}

@article{wadud2016help,
  title={Help or hindrance? The travel, energy and carbon impacts of highly automated vehicles},
  author={Wadud, Zia and MacKenzie, Don and Leiby, Paul},
  journal={Transportation Research Part A: Policy and Practice},
  volume={86},
  pages={1--18},
  year={2016},
  publisher={Elsevier}
}

@article{li2017vehicle,
  title={Vehicle headway modeling and its inferences in macroscopic/microscopic traffic flow theory: A survey},
  author={Li, Li and Chen, Xiqun Michael},
  journal={Transportation Research Part C: Emerging Technologies},
  volume={76},
  pages={170--188},
  year={2017},
  publisher={Elsevier}
}

@article{zhang2020path,
  title={Path-based system optimal dynamic traffic assignment: A subgradient approach},
  author={Zhang, Pinchao and Qian, Sean},
  journal={Transportation Research Part B: Methodological},
  volume={134},
  pages={41--63},
  year={2020},
  publisher={Elsevier}
}

@article{qian2012system,
  title={System-optimal dynamic traffic assignment with and without queue spillback: Its path-based formulation and solution via approximate path marginal cost},
  author={Qian, Zhen Sean and Shen, Wei and Zhang, HM},
  journal={Transportation research part B: methodological},
  volume={46},
  number={7},
  pages={874--893},
  year={2012},
  publisher={Elsevier}
}

@article{ma2017emission,
  title={Emission modeling and pricing on single-destination dynamic traffic networks},
  author={Ma, Rui and Ban, Xuegang Jeff and Szeto, WY},
  journal={Transportation Research Part B: Methodological},
  volume={100},
  pages={255--283},
  year={2017},
  publisher={Elsevier}
}

@article{long2018link,
  title={Link-based system optimum dynamic traffic assignment problems with environmental objectives},
  author={Long, Jiancheng and Chen, Jiaxu and Szeto, WY and Shi, Qin},
  journal={Transportation Research Part D: Transport and Environment},
  volume={60},
  pages={56--75},
  year={2018},
  publisher={Elsevier}
}

@article{tan2021emission,
  title={Emission exposure optimum for a single-destination dynamic traffic network},
  author={Tan, Yu and Ma, Rui and Sun, Zhanbo and Zhang, Peitong},
  journal={Transportation Research Part D: Transport and Environment},
  volume={94},
  pages={102817},
  year={2021},
  publisher={Elsevier}
}

@article{lu2016eco,
  title={Eco-system optimal time-dependent flow assignment in a congested network},
  author={Lu, Chung-Cheng and Liu, Jiangtao and Qu, Yunchao and Peeta, Srinivas and Rouphail, Nagui M and Zhou, Xuesong},
  journal={Transportation Research Part B: Methodological},
  volume={94},
  pages={217--239},
  year={2016},
  publisher={Elsevier}
}

@article{levin2017congestion,
  title={Congestion-aware system optimal route choice for shared autonomous vehicles},
  author={Levin, Michael W},
  journal={Transportation Research Part C: Emerging Technologies},
  volume={82},
  pages={229--247},
  year={2017},
  publisher={Elsevier}
}

@article{liu2020integrated,
  title={Integrated vehicle assignment and routing for system-optimal shared mobility planning with endogenous road congestion},
  author={Liu, Jiangtao and Mirchandani, Pitu and Zhou, Xuesong},
  journal={Transportation Research Part C: Emerging Technologies},
  volume={117},
  pages={102675},
  year={2020},
  publisher={Elsevier}
}

@article{samaranayake2018discrete,
  title={Discrete-time system optimal dynamic traffic assignment (SO-DTA) with partial control for physical queuing networks},
  author={Samaranayake, Samitha and Krichene, Walid and Reilly, Jack and Monache, Maria Laura Delle and Goatin, Paola and Bayen, Alexandre},
  journal={Transportation Science},
  volume={52},
  number={4},
  pages={982--1001},
  year={2018},
  publisher={INFORMS}
}

@article{levin2019linear,
  title={Linear program for system optimal parking reservation assignment},
  author={Levin, Michael W},
  journal={Journal of Transportation Engineering, Part A: Systems},
  volume={145},
  number={12},
  pages={04019049},
  year={2019},
  publisher={American Society of Civil Engineers}
}

@article{qian2014optimal,
  title={Optimal dynamic parking pricing for morning commute considering expected cruising time},
  author={Qian, Zhen Sean and Rajagopal, Ram},
  journal={Transportation Research Part C: Emerging Technologies},
  volume={48},
  pages={468--490},
  year={2014},
  publisher={Elsevier}
}

@article{han2016robust,
  title={A robust optimization approach for dynamic traffic signal control with emission considerations},
  author={Han, Ke and Liu, Hongcheng and Gayah, Vikash V and Friesz, Terry L and Yao, Tao},
  journal={Transportation Research Part C: Emerging Technologies},
  volume={70},
  pages={3--26},
  year={2016},
  publisher={Elsevier}
}

@article{lo2001cell,
  title={A cell-based traffic control formulation: strategies and benefits of dynamic timing plans},
  author={Lo, Hong K},
  journal={Transportation Science},
  volume={35},
  number={2},
  pages={148--164},
  year={2001},
  publisher={INFORMS}
}

@article{szeto2006dynamic,
  title={Dynamic traffic assignment: properties and extensions},
  author={Szeto, WY and Lo, Hong K},
  journal={Transportmetrica},
  volume={2},
  number={1},
  pages={31--52},
  year={2006},
  publisher={Taylor \& Francis}
}

@article{wardrop1952road,
  title={Road paper. some theoretical aspects of road traffic research.},
  author={Wardrop, John Glen},
  journal={Proceedings of the institution of civil engineers},
  volume={1},
  number={3},
  pages={325--362},
  year={1952},
  publisher={Thomas Telford-ICE Virtual Library}
}

@article{waller2001stochastic,
  title={Stochastic dynamic network design problem},
  author={Waller, S Travis and Ziliaskopoulos, Athanasios K},
  journal={Transportation Research Record},
  volume={1771},
  number={1},
  pages={106--113},
  year={2001},
  publisher={SAGE Publications Sage CA: Los Angeles, CA}
}

@article{waller2006linear,
  title={A linear model for the continuous network design problem},
  author={Waller, S Travis and Mouskos, Kyriacos C and Kamaryiannis, Dimitrios and Ziliaskopoulos, Athanasios K},
  journal={Computer-Aided Civil and Infrastructure Engineering},
  volume={21},
  number={5},
  pages={334--345},
  year={2006},
  publisher={Wiley Online Library}
}

@article{chiu2007modeling,
  title={Modeling no-notice mass evacuation using a dynamic traffic flow optimization model},
  author={Chiu, Yi-Chang and Zheng, Hong and Villalobos, Jorge and Gautam, Bikash},
  journal={Iie Transactions},
  volume={39},
  number={1},
  pages={83--94},
  year={2007},
  publisher={Taylor \& Francis}
}

@article{liu2006cell,
  title={Cell-based network optimization model for staged evacuation planning under emergencies},
  author={Liu, Ying and Lai, Xiaorong and Chang, Gang-Len},
  journal={Transportation Research Record},
  volume={1964},
  number={1},
  pages={127--135},
  year={2006},
  publisher={SAGE Publications Sage CA: Los Angeles, CA}
}

@article{osorio2011dynamic,
  title={Dynamic network loading: a stochastic differentiable model that derives link state distributions},
  author={Osorio, Carolina and Fl{\"o}tter{\"o}d, Gunnar and Bierlaire, Michel},
  journal={Procedia-Social and Behavioral Sciences},
  volume={17},
  pages={364--381},
  year={2011},
  publisher={Elsevier}
}

@article{wang2018dynamic,
  title={Dynamic traffic assignment: A review of the methodological advances for environmentally sustainable road transportation applications},
  author={Wang, Yi and Szeto, Wai Y and Han, Ke and Friesz, Terry L},
  journal={Transportation Research Part B: Methodological},
  volume={111},
  pages={370--394},
  year={2018},
  publisher={Elsevier}
}

@article{van2016user,
  title={From user equilibrium to system optimum: a literature review on the role of travel information, bounded rationality and non-selfish behaviour at the network and individual levels},
  author={van Essen, Mariska and Thomas, Tom and van Berkum, Eric and Chorus, Caspar},
  journal={Transport reviews},
  volume={36},
  number={4},
  pages={527--548},
  year={2016},
  publisher={Taylor \& Francis}
}

@article{chow2009properties,
  title={Properties of system optimal traffic assignment with departure time choice and its solution method},
  author={Chow, Andy HF},
  journal={Transportation Research Part B: Methodological},
  volume={43},
  number={3},
  pages={325--344},
  year={2009},
  publisher={Elsevier}
}

@article{long2019link,
  title={Link-based system optimum dynamic traffic assignment problems in general networks},
  author={Long, Jiancheng and Szeto, Wai Yuen},
  journal={Operations Research},
  volume={67},
  number={1},
  pages={167--182},
  year={2019},
  publisher={INFORMS}
}

@article{ngoduy2016optimal,
  title={Optimal queue placement in dynamic system optimum solutions for single origin-destination traffic networks},
  author={Ngoduy, Dong and Hoang, NH and Vu, Hai L and Watling, D},
  journal={Transportation Research Part B: Methodological},
  volume={92},
  pages={148--169},
  year={2016},
  publisher={Elsevier}
}

@article{gawron1998iterative,
  title={An iterative algorithm to determine the dynamic user equilibrium in a traffic simulation model},
  author={Gawron, Christian},
  journal={International Journal of Modern Physics C},
  volume={9},
  number={03},
  pages={393--407},
  year={1998},
  publisher={World Scientific}
}

@inproceedings{penrose1955generalized,
  title={A generalized inverse for matrices},
  author={Penrose, Roger},
  booktitle={Mathematical proceedings of the Cambridge philosophical society},
  volume={51},
  number={3},
  pages={406--413},
  year={1955},
  organization={Cambridge University Press}
}

@article{nguyen2021system,
  title={A system optimal speed advisory framework for a network of connected and autonomous vehicles},
  author={Nguyen, Cuong HP and Hoang, Nam H and Lee, Seunghyeon and Vu, Hai L},
  journal={IEEE Transactions on Intelligent Transportation Systems},
  volume={23},
  number={6},
  pages={5727--5739},
  year={2021},
  publisher={IEEE}
}

@article{battifarano2023impact,
  title={The Impact of Optimized Fleets in Transportation Networks},
  author={Battifarano, Matthew and Qian, Sean},
  journal={Transportation Science},
  year={2023},
  publisher={INFORMS}
}

@article{xiao2023adaptive,
  title={Adaptive Headway Control Algorithm for Mixed-Traffic Stabilization and Optimization with Automated Cars and Trucks},
  author={Xiao, Xiao and Zhang, Yunlong and Wang, Xiubin B and Guo, Xiaoyu},
  journal={Transportation Research Record},
  pages={03611981231156587},
  year={2023},
  publisher={SAGE Publications Sage CA: Los Angeles, CA}
}

@article{elmorshedy2023freeway,
  title={Freeway Congestion Management With Reinforcement Learning Headway Control of Connected and Autonomous Vehicles},
  author={Elmorshedy, Lina and Smirnov, Ilia and Abdulhai, Baher},
  journal={Transportation Research Record},
  pages={03611981231152459},
  year={2023},
  publisher={SAGE Publications Sage CA: Los Angeles, CA}
}

@article{li2017dynamical,
  title={Dynamical modeling and distributed control of connected and automated vehicles: Challenges and opportunities},
  author={Li, Shengbo Eben and Zheng, Yang and Li, Keqiang and Wu, Yujia and Hedrick, J Karl and Gao, Feng and Zhang, Hongwei},
  journal={IEEE Intelligent Transportation Systems Magazine},
  volume={9},
  number={3},
  pages={46--58},
  year={2017},
  publisher={IEEE}
}

@article{chiu1977vehicle,
  title={Vehicle-follower control with variable-gains for short headway automated guideway transit systems},
  author={Chiu, Harry Y and Stupp Jr, GB and Brown Jr, SJ},
  year={1977}
}

@article{shladover1991automated,
  title={Automated vehicle control developments in the PATH program},
  author={Shladover, Steven E and Desoer, Charles A and Hedrick, J Karl and Tomizuka, Masayoshi and Walrand, Jean and Zhang, W-B and McMahon, Donn H and Peng, Huei and Sheikholeslam, Shahab and McKeown, Nick},
  journal={IEEE Transactions on vehicular technology},
  volume={40},
  number={1},
  pages={114--130},
  year={1991},
  publisher={IEEE}
}

@article{xiao2011practical,
  title={Practical string stability of platoon of adaptive cruise control vehicles},
  author={Xiao, Lingyun and Gao, Feng},
  journal={IEEE Transactions on intelligent transportation systems},
  volume={12},
  number={4},
  pages={1184--1194},
  year={2011},
  publisher={IEEE}
}

@article{middleton2010string,
  title={String instability in classes of linear time invariant formation control with limited communication range},
  author={Middleton, Richard H and Braslavsky, Julio H},
  journal={IEEE Transactions on Automatic Control},
  volume={55},
  number={7},
  pages={1519--1530},
  year={2010},
  publisher={IEEE}
}

@article{chen2020connected,
  title={Connected automated vehicle platoon control with input saturation and variable time headway strategy},
  author={Chen, Jianzhong and Liang, Huan and Li, Jing and Lv, Zekai},
  journal={IEEE Transactions on Intelligent Transportation Systems},
  volume={22},
  number={8},
  pages={4929--4940},
  year={2020},
  publisher={IEEE}
}

@article{fenton1979headway,
  title={A headway safety policy for automated highway operations},
  author={Fenton, Robert E},
  journal={IEEE Transactions on Vehicular Technology},
  volume={28},
  number={1},
  pages={22--28},
  year={1979},
  publisher={IEEE}
}

@inproceedings{ayres2001preferred,
  title={Preferred time-headway of highway drivers},
  author={Ayres, TJ and Li, L and Schleuning, David and Young, D},
  booktitle={ITSC 2001. 2001 IEEE Intelligent Transportation Systems. Proceedings (Cat. No. 01TH8585)},
  pages={826--829},
  year={2001},
  organization={IEEE}
}

@article{vogel2003comparison,
  title={A comparison of headway and time to collision as safety indicators},
  author={Vogel, Katja},
  journal={Accident analysis \& prevention},
  volume={35},
  number={3},
  pages={427--433},
  year={2003},
  publisher={Elsevier}
}

@phdthesis{biswas2022drivers,
  title={How do drivers avoid crashes: the role of driving headway},
  author={Biswas, Raaj Kishore},
  year={2022},
  school={UNSW Sydney}
}

@article{chen2020path,
  title={Path controlling of automated vehicles for system optimum on transportation networks with heterogeneous traffic stream},
  author={Chen, Zhibin and Lin, Xi and Yin, Yafeng and Li, Meng},
  journal={Transportation Research Part C: Emerging Technologies},
  volume={110},
  pages={312--329},
  year={2020},
  publisher={Elsevier}
}

@inproceedings{de2015autonomous,
  title={Autonomous driving at intersections: combining theoretical analysis with practical considerations},
  author={de La Fortelle, Arnaud and Qian, Xiangjun},
  booktitle={Its world congress 2015},
  year={2015}
}

@article{lighthill1955kinematic,
  title={On kinematic waves II. A theory of traffic flow on long crowded roads},
  author={Lighthill, Michael James and Whitham, Gerald Beresford},
  journal={Proceedings of the royal society of london. series a. mathematical and physical sciences},
  volume={229},
  number={1178},
  pages={317--345},
  year={1955},
  publisher={The Royal Society London}
}

@article{richards1956shock,
  title={Shock waves on the highway},
  author={Richards, Paul I},
  journal={Operations research},
  volume={4},
  number={1},
  pages={42--51},
  year={1956},
  publisher={INFORMS}
}

@article{nie2005comparative,
  title={A comparative study of some macroscopic link models used in dynamic traffic assignment},
  author={Nie, Xiaojian and Zhang, H Michael},
  journal={Networks and Spatial Economics},
  volume={5},
  pages={89--115},
  year={2005},
  publisher={Springer}
}

@inproceedings{lebacque2005first,
  title={First-order macroscopic traffic flow models: Intersection modeling, network modeling},
  author={Lebacque, Jean-Patrick},
  booktitle={Transportation and Traffic Theory. Flow, Dynamics and Human Interaction. 16th International Symposium on Transportation and Traffic TheoryUniversity of Maryland, College Park},
  year={2005}
}

@article{jin2015continuous,
  title={Continuous formulations and analytical properties of the link transmission model},
  author={Jin, Wen-Long},
  journal={Transportation Research Part B: Methodological},
  volume={74},
  pages={88--103},
  year={2015},
  publisher={Elsevier}
}

@article{han2016continuous,
  title={Continuous-time link-based kinematic wave model: formulation, solution existence, and well-posedness},
  author={Han, Ke and Piccoli, Benedetto and Szeto, WY},
  journal={Transportmetrica B: Transport Dynamics},
  volume={4},
  number={3},
  pages={187--222},
  year={2016},
  publisher={Taylor \& Francis}
}

@article{ban2012continuous,
  title={Continuous-time point-queue models in dynamic network loading},
  author={Ban, Xuegang Jeff and Pang, Jong-Shi and Liu, Henry X and Ma, Rui},
  journal={Transportation Research Part B: Methodological},
  volume={46},
  number={3},
  pages={360--380},
  year={2012},
  publisher={Elsevier}
}

@article{ben2012dynamic,
  title={A dynamic traffic assignment model for highly congested urban networks},
  author={Ben-Akiva, Moshe E and Gao, Song and Wei, Zheng and Wen, Yang},
  journal={Transportation research part C: emerging technologies},
  volume={24},
  pages={62--82},
  year={2012},
  publisher={Elsevier}
}

@article{sumalee2011stochastic,
  title={Stochastic cell transmission model (SCTM): A stochastic dynamic traffic model for traffic state surveillance and assignment},
  author={Sumalee, Agachai and Zhong, RX and Pan, TL and Szeto, WY},
  journal={Transportation Research Part B: Methodological},
  volume={45},
  number={3},
  pages={507--533},
  year={2011},
  publisher={Elsevier}
}

@article{chiu2011dynamic,
  title={Dynamic traffic assignment: A primer (transportation research circular e-c153)},
  author={Chiu, Yi-Chang and Bottom, Jon and Mahut, Michael and Paz, Alexander and Balakrishna, Ramachandran and Waller, Steven and Hicks, Jim},
  year={2011},
  publisher={Transportation Research Board}
}

@article{bellei2005within,
  title={A within-day dynamic traffic assignment model for urban road networks},
  author={Bellei, Giuseppe and Gentile, Guido and Papola, Natale},
  journal={Transportation Research Part B: Methodological},
  volume={39},
  number={1},
  pages={1--29},
  year={2005},
  publisher={Elsevier}
}

@article{janson1991dynamic,
  title={Dynamic traffic assignment for urban road networks},
  author={Janson, Bruce N},
  journal={Transportation Research Part B: Methodological},
  volume={25},
  number={2-3},
  pages={143--161},
  year={1991},
  publisher={Elsevier}
}

@article{chan2021quasi,
  title={Quasi-dynamic traffic assignment using high performance computing},
  author={Chan, Cy and Kuncheria, Anu and Zhao, Bingyu and Cabannes, Theophile and Keimer, Alexander and Wang, Bin and Bayen, Alexandre and Macfarlane, Jane},
  journal={arXiv preprint arXiv:2104.12911},
  year={2021}
}

@article{guo2012autonomous,
  title={Autonomous platoon control allowing range-limited sensors},
  author={Guo, Ge and Yue, Wei},
  journal={IEEE Transactions on vehicular technology},
  volume={61},
  number={7},
  pages={2901--2912},
  year={2012},
  publisher={IEEE}
}

@article{zhu2018distributed,
  title={Distributed adaptive longitudinal control for uncertain third-order vehicle platoon in a networked environment},
  author={Zhu, Yang and Zhu, Feng},
  journal={IEEE Transactions on Vehicular Technology},
  volume={67},
  number={10},
  pages={9183--9197},
  year={2018},
  publisher={IEEE}
}

@article{becker2022driver,
  title={Driver-initiated take-overs during critical braking maneuvers in automated driving--the role of time headway, traction usage, and trust in automation},
  author={Becker, Sandra and Brandenburg, Stefan and Th{\"u}ring, Manfred},
  journal={Accident Analysis \& Prevention},
  volume={174},
  pages={106725},
  year={2022},
  publisher={Elsevier}
}

@article{jabari2016node,
  title={Node modeling for congested urban road networks},
  author={Jabari, Saif Eddin},
  journal={Transportation Research Part B: Methodological},
  volume={91},
  pages={229--249},
  year={2016},
  publisher={Elsevier}
}

@article{li2023optimal,
  title={Optimal intersection design and signal setting in a transportation network with mixed HVs and CAVs},
  author={Li, Tongfei and Cao, Yaning and Xu, Min and Sun, Huijun},
  journal={Transportation Research Part E: Logistics and Transportation Review},
  volume={175},
  pages={103173},
  year={2023},
  publisher={Elsevier}
}

@article{zhu2015linear,
  title={A linear programming formulation for autonomous intersection control within a dynamic traffic assignment and connected vehicle environment},
  author={Zhu, Feng and Ukkusuri, Satish V},
  journal={Transportation Research Part C: Emerging Technologies},
  volume={55},
  pages={363--378},
  year={2015},
  publisher={Elsevier}
}

@article{yu2018optimal,
  title={Optimal traffic signal control under dynamic user equilibrium and link constraints in a general network},
  author={Yu, Hao and Ma, Rui and Zhang, H Michael},
  journal={Transportation research part B: methodological},
  volume={110},
  pages={302--325},
  year={2018},
  publisher={Elsevier}
}

@article{ma2018estimating,
  title={Estimating multi-year 24/7 origin-destination demand using high-granular multi-source traffic data},
  author={Ma, Wei and Qian, Zhen Sean},
  journal={Transportation Research Part C: Emerging Technologies},
  volume={96},
  pages={96--121},
  year={2018},
  publisher={Elsevier}
}

@article{zhang2021short,
  title={Short-term origin-destination demand prediction in urban rail transit systems: A channel-wise attentive split-convolutional neural network method},
  author={Zhang, Jinlei and Che, Hongshu and Chen, Feng and Ma, Wei and He, Zhengbing},
  journal={Transportation Research Part C: Emerging Technologies},
  volume={124},
  pages={102928},
  year={2021},
  publisher={Elsevier}
}

@article{ke2021predicting,
  title={Predicting origin-destination ride-sourcing demand with a spatio-temporal encoder-decoder residual multi-graph convolutional network},
  author={Ke, Jintao and Qin, Xiaoran and Yang, Hai and Zheng, Zhengfei and Zhu, Zheng and Ye, Jieping},
  journal={Transportation Research Part C: Emerging Technologies},
  volume={122},
  pages={102858},
  year={2021},
  publisher={Elsevier}
}

@article{bhattacharjee2001modeling,
  title={Modeling the effects of traveler information on freeway origin--destination demand prediction},
  author={Bhattacharjee, Deb and Sinha, Kumares C and Krogmeier, James V},
  journal={Transportation Research Part C: Emerging Technologies},
  volume={9},
  number={6},
  pages={381--398},
  year={2001},
  publisher={Elsevier}
}

@article{xiong2020dynamic,
  title={Dynamic origin--destination matrix prediction with line graph neural networks and kalman filter},
  author={Xiong, Xi and Ozbay, Kaan and Jin, Li and Feng, Chen},
  journal={Transportation Research Record},
  volume={2674},
  number={8},
  pages={491--503},
  year={2020},
  publisher={SAGE Publications Sage CA: Los Angeles, CA}
}

\clearpage

\appendix

  \section{Notations}
  \label{app:notations}
  
  \begin{table}[H]
    \centering
    \resizebox{1.02\textwidth}{!}{
    \begin{tabular}{lll}
        \toprule
        Type      & Notation   & Description \\
        \midrule
        Sets      & \textit{V} & Set of nodes \\
                  & \textit{R} & Set of origins \\
                  & $\widetilde{R}$ & Set of dummy origins \\
                  & \textit{S} & Set of destinations\\
                  & $\widetilde{S}$ & Set of dummy destinations \\
                  & \textit{E} & Set of links \\
                  & $L_R$ & Set of origin connectors \\
                  & $L_S$ & Set of destination connectors \\
        Parameters & $\textit{v}_{i,j}^f$ & Free-flow speed of link $(i,j)$ \\
                  & $L_{i,j}$ & Length of link $(i,j)$ \\
                  & \textit{L} & Length of an AV \\
                  & $\bar{Q}_{i,j}^{\mathcal{U}}$ & Upstream queue capacity of link $(i,j)$ \\
                  & $\bar{Q}_{i,j}^{\mathcal{D}}$ & Downstream queue capacity of link $(i,j)$ \\
                  & $\bar{C}_{i,j}^{\mu}$ & Inflow capacity of link $(i,j)$ \\ 
                  & $\bar{C}_{i,j}^{\nu}$ & Outflow capacity of link $(i,j)$ \\
                  & $\textit{T}_\textit{1}$ & Time horizon of demand \\
                  & \textit{T} & Total time horizon \\
                  & $\textit{h}_{i,j;t}^{min}$ & Lower bound of headway in link $(i,j)$ at time $t$ \\
                  & $\textit{h}_{i,j;t}^{max}$ & Upper bound of headway in link $(i,j)$ at time $t$ \\
        Variables & $\textit{u}_{i,j}(t)$ & Inflow of the flow area of link $(i,j)$ at time $t$ \\
                  & $\textit{u}_{i,j}^{s^{\prime}}(t)$ & Inflow of the flow area of link $(i,j)$ with destination $s^{\prime}$ at time $t$ \\
                  & $\textit{v}_{i,j}(t)$ & Outflow of the buffer area of link $(i,j)$ at time $t$ \\
                  & $\textit{v}_{i,j}^{s^{\prime}}(t)$ & Outflow of the buffer area of link $(i,j)$ with destination $s^{\prime}$ at time $t$ \\
                  & $\textit{f}_{i,j}(t)$ & Inflow at the boundary of the flow area and  buffer area of link $(i,j)$ at time $t$ \\
                  & $\textit{f}_{i,j}^{~ s^{\prime}}(t)$ & Inflow at the boundary of the flow area and  buffer area of link $(i,j)$ with destination $s^{\prime}$ at time $t$ \\
                  & $\rho_{i,j}(t)$ & Density of the flow area of link $(i,j)$ at time $t$ \\
                  & $\rho_{i,j}^{s^{\prime}}(t)$ & Density of the flow area of link $(i,j)$ with destination $s^{\prime}$ at time $t$ \\
                  & $\textit{q}_{i,j}^{\mathcal{D}}(t)$ & Downstream queue of link $(i,j)$ at time $t$ \\
                  & $\textit{q}_{i,j}^{\mathcal{D},s^{\prime}}(t)$ & Downstream queue of link $(i,j)$ with destination $s^{\prime}$ at time $t$ \\
                  & $\textit{q}_{i,j}^{\mathcal{U}}(t)$ & Upstream queue of link $(i,j)$ at time $t$ \\
                  & $\textit{q}_{i,j}^{\mathcal{U},s^{\prime}}(t)$ & Upstream queue of link $(i,j)$ with destination $s^{\prime}$ at time $t$ \\
                  & $\textit{h}_{i,j}(t)$ & Headway of automated vehicles in link $(i,j)$ at time $t$ \\
                  & $\tau_{i,j}^{w}(t)$ & Wave travel time in congested states on link $(i,j)$ at time $t$ \\
                  & $\textit{d}_{r^{\prime},s^{\prime}}(t)$ & Travel demand of O-D pair  $(r^{\prime},s^{\prime})$ at time $t$ \\
                  & $\textit{w}_{r^{\prime},r}^{~s^{\prime}}(t)$ & Number of AVs waiting at origins in the origin connectors  $(r^{\prime},r)$ with destination $s^{\prime}$ at time $t$ \\
                  & $\textit{U}_{i,j}^{s^{\prime}}(t)$ & Cumulative inflow of the flow area of link $(i,j)$ with destination $s^{\prime}$ at time $t$ \\
                  & $\textit{F}_{i,j}^{s^{\prime}}(t)$ & Cumulative inflow at the boundary of the flow area and buffer area of link $(i,j)$ with destination $s^{\prime}$ at time $t$ \\
                  & $\textit{V}_{i,j}^{s^{\prime}}(t)$ & Cumulative outflow of the buffer area of link $(i,j)$ with destination $s^{\prime}$ at time $t$ \\
        \bottomrule
    \end{tabular}
    }
    \caption{Notation table.}
    \label{notation_con}
\end{table}

  \section{Proof of Proposition~\ref{prop:HDQ&DQ}}
  \label{app:generality}
Proposition \ref{prop:HDQ&DQ} indicates that the headway-dependent double queue (HDQ) model is a generalized form of the double queue (DQ) model.
The formulation of the downstream queue $\textit{q}_{i,j}^{d,s^{\prime}}(t)$ and upstream queue $\textit{q}_{i,j}^{u,s^{\prime}}(t)$ in DQ for $ (i,j)\in E,~s^{\prime}\in \widetilde{S}$ and $t\in[0,T]$ are shown as follows:
\begin{eqnarray}
\textit{q}_{i,j}^{d,s^{\prime}}(t) &=& \int_{0}^{\textit{t}-\tau_{i,j}^{f}} \textit{u}_{i,j}^{s^{\prime}}(\hat{t}) ~d\hat{t} - \int_{0}^{\textit{t}} \textit{v}_{i,j}^{s^{\prime}}(\hat{t}) ~d\hat{t} + \textit{q}_{i,j}^{d,s^{\prime}}(0) ~ = ~ \textit{U}_{i,j}^{s^{\prime}}(t-\tau_{i,j}^{f})-\textit{V}_{i,j}^{s^{\prime}}(t),
\label{eq:down_DQ} 
\nonumber\\
\textit{q}_{i,j}^{u,s^{\prime}}(t) &=& \int_{0}^{\textit{t}} \textit{u}_{i,j}^{s^{\prime}}(\hat{t}) ~d\hat{t} - \int_{0}^{\textit{t}-\tau_{i,j}^{w}} \textit{v}_{i,j}^{s^{\prime}}(\hat{t}) ~d\hat{t} + \textit{q}_{i,j}^{u,s^{\prime}}(0)  ~ = ~
\textit{U}_{i,j}^{s^{\prime}}(t) - \textit{V}_{i,j}^{s^{\prime}}(t-\tau_{i,j}^{w}),
\label{eq:up_DQ}
\nonumber
\end{eqnarray}
 where $\textit{q}_{i,j}^{d,s^{\prime}}(t)$ and $\textit{q}_{i,j}^{u,s^{\prime}}(t)$ represent the downstream and upstream queue in link $(i,j)$ with destination $s^{\prime}$ at time $t$ in DQ, respectively, and the two queues are interpreted as follows:

\begin{itemize}
    \item Downstream queue $\textit{q}_{i,j}^{d,s^{\prime}}(t)$ measures the number of vehicles with destination $s^{\prime}$ at the end of link $(i,j)$ at time $t$. DQ assumes the inflow $u_{i,j}^{s^{\prime}}$ would travel from the entrance to the exit of link $(i,j)$ at the free-flow speed, and the outflow of link $(i,j)$ at time $t$ is determined by the outflow capacity link $(i,j)$ and upstream queue in next link.
    \item Upstream queue $\textit{q}_{i,j}^{u,s^{\prime}}(t)$ measures the number of vehicles at the entrance of link $(i,j)$ at time $t$. DQ assumes the space of outflow $v_{i,j}^{s^{\prime}}$ would travel from the exit to the entrance of link $(i,j)$ at the wave travel speed in congested states. 
\end{itemize}

However, the assumptions of DQ model may have the following limitations:

\begin{itemize}
    \item Congestion may occur in the flow area of a link so that AVs may not travel from the entrance to the exit of a link at the free-flow speed. The influence of headway on the flow of a link is also not considered in DQ.
    \item Space may not propagate to the entrance of a link in time and retain in the downstream. Take an extreme example, if the front AV leaves the link but other AVs behind are still, the space left by the front AV would still retain in the downstream and it is not accurate to use the outflow to calculate the propagation of space in wave travel speed. Besides, the influence of headway on the wave travel speed on a link is also not considered in DQ. 
\end{itemize}

The HDQ addresses the above limitations by exploring the effects of headway. Equations \ref{eq:dq_con} and \ref{eq:uq_con} formulate the downstream and upstream queue in HDQ. The relationship between the HDQ and DQ is presented in Equation \ref{eq:down_DQ_HDQ} and \ref{eq:up_DQ_HDQ}.

  \begin{eqnarray}
\textit{q}_{i,j}^{d,s^{\prime}}(t) &=&  \textit{U}_{i,j}^{s^{\prime}}(t-\tau_{i,j}^{f})-\textit{V}_{i,j}^{s^{\prime}}(t) \nonumber \\
&=& \textit{U}_{i,j}^{s^{\prime}}(t-\tau_{i,j}^{f})-\textit{F}_{i,j}^{s^{\prime}}(t)+\textit{F}_{i,j}^{s^{\prime}}(t)
-\textit{V}_{i,j}(t) \nonumber \\
&=& \textit{q}_{i,j}^{\mathcal{D}}(t)+\textit{U}_{i,j}^{s^{\prime}}(t-\tau_{i,j}^{f})-\textit{F}_{i,j}^{s^{\prime}}(t) \label{eq:down_DQ_HDQ} \\
\textit{q}_{i,j}^{u,s^{\prime}}(t) &=& 
\textit{U}_{i,j}^{s^{\prime}}(t) - \textit{V}_{i,j}^{s^{\prime}}(t-\tau_{i,j}^{w} \nonumber \\
&=& \textit{U}_{i,j}^{s^{\prime}}(t) - \textit{F}_{i,j}^{s^{\prime}}(t-\tau_{i,j}^{w} + \textit{F}_{i,j}^{s^{\prime}}(t-\tau_{i,j}^{w}) - \textit{V}_{i,j}^{s^{\prime}}(t-\tau_{i,j}^{w}) \nonumber \\
&=& \textit{q}_{i,j}^{\mathcal{U},s^{\prime}}(t)+\textit{q}_{i,j}^{\mathcal{D},s^{\prime}}(t-\tau_{i,j}^{w})
\label{eq:up_DQ_HDQ}
\end{eqnarray}

In Equation \ref{eq:down_DQ_HDQ}, $\textit{U}_{i,j}^{s^{\prime}}(t-\tau_{i,j}^{f})-\textit{F}_{i,j}^{s^{\prime}}(t)$ represents the congestion in the flow area in HDQ. However, under the assumption of DQ, we have $\textit{u}_{i,j}^{s^{\prime}}(t-\tau_{i,j}^{f}) =\textit{f}_{i,j}^{s^{\prime}}(t)$, meaning that vehicles would travel from the entrance of link to the boundary of the flow area and buffer area at free-flow speed. Therefore, the downstream queue of HDQ is equivalent to the DQ as follows:

\begin{equation}
\textit{q}_{i,j}^{d,s^{\prime}}(t) = \textit{q}_{i,j}^{\mathcal{D},s^{\prime}}(t)+\textit{U}_{i,j}^{s^{\prime}}(t-\tau_{i,j}^{f})-\textit{F}_{i,j}^{s^{\prime}}(t) = \textit{q}_{i,j}^{\mathcal{D},s^{\prime}}(t).
    \nonumber
\end{equation}

In Equation \ref{eq:up_DQ_HDQ}, the upstream queue in DQ actually consists of the queue in both the flow area and buffer area in HDQ. The upstream queue in HDQ actually considers the flow entering the buffer area to fill the leaving space. If we ignore it, we would derive the upstream queue by outflow $v$ as DQ.

Different from DQ, the HDQ specifies the traffic state in a link and incorporates the effect of headway. To calculate the number of AVs at the entrance and end of a link, we just focus on the flow area and buffer area, respectively. Compared with the DQ, the HDQ provides a more precise way to calculate the downstream and upstream queue if we are clear about the state of the boundary of the flow area and buffer area compared with the DQ. If under the assumption of DQ, the effect of headway is ignored and HDQ would be reduced to DQ.

  \section{Proof of Proposition~\ref{prop:property}}
  \label{app:HDQ}
  This appendix proves the properties of the headway-dependent double queue (HDQ) model, and we mainly show the constraints on the flow rates of the HDQ.
  
  With Eqs.(4.29) in \cite{LTM} and the formulation of the downstream queue in HDQ, Equation \ref{app:eq_send_HDQ} presents the constraint on the sending flow $\textit{S}_{i,j}(t)$ of HDQ.
  
  \begin{equation}
   \textit{S}_{i,j}(t)\leq \textit{V}_{i,j}(t+\Delta_{t})-\textit{V}_{i,j}(t) \leq  \textit{F}_{i,j}(t+\Delta_{t})-\textit{V}_{i,j}(t)
  \label{app:eq_send_HDQ}
  \end{equation}
  where $\textit{V}_{i,j}(t)$ and $\textit{F}_{i,j}(t)$ are the cumulative outflow of the buffer area and cumulative inflow at the boundary of the flow area and buffer area of link $(i,j)$  at time $t$, respectively.
  
  With Eqs.(4.33) in \cite{LTM} and the formulation of the upstream queue in HDQ, the constraint on the receiving flow $\textit{R}_{i,j}(t)$ of the HDQ is stated in Equation \ref{app:eq_receive_HDQ}.
  
  \begin{equation}
  \textit{R}_{i,j}(t) \leq \textit{U}_{i,j}(t+\Delta_{t})-\textit{U}_{i,j}(t) \leq  \bar{Q}_{i,j}^{u}+\textit{F}_{i,j}(t+\Delta_{t}-\tau_{i,j}^{w}(t))-\textit{U}_{i,j}(t)
  \label{app:eq_receive_HDQ}
  \end{equation}

  where $\textit{U}_{i,j}(t)$ is the cumulative inflow of the flow area of link $(i,j)$  at time $t$.
  
  When the time step $\Delta_{t}$ approaches zero, the constraint of outflow $\textit{v}$ is shown in \ref{app:v_con} and the constraint of inflow $\textit{u}$ is shown in \ref{app:u_con}.
  
  \begin{eqnarray}
\textit{v}_{i,j}(t) ~\leq~ \lim_{\Delta_{t} \to 0} \frac{\textit{S}_{i,j}(t)} {\Delta_{t}} 
&\leq& \lim_{\Delta_{t} \to 0} \min \left\{ \frac{\textit{F}_{i,j}(t+\Delta_{t})-\textit{V}_{i,j}(t)}{\Delta_{t}},~ \bar{C}_{i,j}^{\nu} \right\} \nonumber \\ 
&=& \lim_{\Delta_{t} \to 0} \min \left\{ \frac{\textit{F}_{i,j}(t+\Delta_{t})-\textit{F}_{i,j}(t)}{\Delta_{t}}+  \frac{\textit{F}_{i,j}(t)-\textit{V}_{i,j}(t)}{\Delta_{t}}, ~ \bar{C}_{i,j}^{\nu} \right\} \nonumber \\
&=& \lim_{\Delta_{t} \to 0} \min \left\{ \frac{\textit{F}_{i,j}(t+\Delta_{t})-\textit{F}_{i,j}(t)}{\Delta_{t}}+  \frac{\textit{q}_{i,j}^{\mathcal{D}}(t)}{\Delta_{t}}, ~ \bar{C}_{i,j}^{\nu} \right\} \nonumber \\
&=& \left\{ \begin{array}{lcl}
\bar{C}_{i,j}^{\nu},  &  & \textit{q}_{i,j}^{\mathcal{D}}(t)>0 \\
\min \left\{ \textit{f}_{i,j}(t),~ \bar{C}_{i,j}^{\nu} \right\}, &  & \textit{q}_{i,j}^{\mathcal{D}}(t)=0
\end{array} \right. \label{app:v_con}\\
\textit{u}_{i,j}(t) ~\leq~ \lim_{\Delta_{t} \to 0} \frac{\textit{R}_{i,j}(t)} {\Delta_{t}} 
&\leq& \lim_{\Delta_{t} \to 0} \min \left\{ \frac{\bar{Q}_{i,j}^{u}+\textit{F}_{i,j}(t+\Delta_{t}-\tau_{i,j}^{w}(t))-\textit{V}_{i,j}(t)}{\Delta_{t}},~ \bar{C}_{i,j}^{\mu} \right\} \nonumber \\
&=& \lim_{\Delta_{t} \to 0} \min \left\{ \frac{\bar{Q}_{i,j}^{u}-\left(\textit{V}_{i,j}(t)-\textit{F}_{i,j}(t-\tau_{i,j}^{w}(t))\right)}{\Delta_{t}}+  \frac{\textit{F}_{i,j}(t+\Delta_{t}-\tau_{i,j}^{w}(t))-\textit{F}_{i,j}(t-\tau_{i,j}^{w}(t))}{\Delta_{t}}, ~ \bar{C}_{i,j}^{\mu} \right\} \nonumber \\
&=& \lim_{\Delta_{t} \to 0} \min \left\{ \frac{\bar{Q}_{i,j}^{u}-\textit{q}_{i,j}^{\mathcal{U}}(t)}{\Delta_{t}}+  \frac{\textit{F}_{i,j}(t+\Delta_{t}-\tau_{i,j}^{w}(t))-\textit{F}_{i,j}(t-\tau_{i,j}^{w}(t))}{\Delta_{t}}, ~ \bar{C}_{i,j}^{\mu} \right\} \nonumber \\
&=& \left\{ \begin{array}{lcl}
\bar{C}_{i,j}^{\mu},  &  & \textit{q}_{i,j}^{\mathcal{U}}(t)<\bar{Q}_{i,j}^{u} \\
\min \left\{ \textit{f}_{i,j}(t-\tau_{i,j}^{w}(t)),~ \bar{C}_{i,j}^{\mu} \right\}, &  & \textit{q}_{i,j}^{\mathcal{U}}(t)=\bar{Q}_{i,j}^{u}
\end{array} \right. \label{app:u_con}
\end{eqnarray}
  
  \section{Proof of Proposition~\ref{prop:feasibility}}
  \label{app:feas}
  In general, we could construct a feasible solution to Equation~\ref{eq:obj_dis} using the following three steps:
  
\begin{itemize} 
\item Let all the demands of each O-D pair arrive at the downstream queue in the origin connectors while not releasing them. We define the time period for such a process as  $\textit{N}_\textit{1}$. We don't consider the capacity of the number of AVs waiting at origins here. 

\item For each O-D pair $(r^{\prime},s^{\prime})$, we divide the total demand $D^{r^{\prime},s^{\prime}}$ into many equal batches $B^{r^{\prime},s^{\prime}}$, and $D^{r^{\prime},s^{\prime}}=n^{r^{\prime},s^{\prime}} B^{r^{\prime},s^{\prime}}$. We only use one path of each O-D pair $(r^{\prime},s^{\prime})$, which is denoted as $P^{r^{\prime},s^{\prime}}$. We only release one batch of an O-D pair only once. After all the AVs in this batch arrive at the destination, we will release another batch of this O-D pair. Repeat this process until finishing all batches of an O-D pair and continue another O-D pair.


\item In each batch, only when all the AVs in this batch arrive at the downstream queue in each link, we will release them into the next link in the selected path. The time of passing link $(i,j)$ for batch $B^{r^{\prime},s^{\prime}}$ is denoted by $N_{i,j}^{r^{\prime},s^{\prime}}$. 
\end{itemize}

Therefore, the total travel time consumed by the above three steps is $\textit{N}_\textit{1}+\sum\limits_{(r^{\prime},s^{\prime})}n^{r^{\prime},s^{\prime}}\sum\limits_{(i,j)\in P^{r^{\prime},s^{\prime}}}N_{i,j}^{r^{\prime},s^{\prime}}$. If we choose the time horizon $N$ that is larger than $\textit{N}_\textit{1}+\sum\limits_{(r^{\prime},s^{\prime})}n^{r^{\prime},s^{\prime}}\sum\limits_{(i,j)\in P^{r^{\prime},s^{\prime}}}N_{i,j}^{r^{\prime},s^{\prime}}$, we could construct a solution so that the problem is feasible. To this end, we present the details of how to construct the flow as follows.


In order to satisfy the queue capacity constraint in  Equation \ref{eq:ineq_1_dis}, the size of the batch should be constrained as follows:

\begin{equation}
    B^{r^{\prime},s^{\prime}}<\mathop{min}\limits_{(i,j)\in P^{r^{\prime},s^{\prime}}}\left\{\bar{Q}_{i,j}^{\mathcal{U}},\bar{Q}_{i,j}^{\mathcal{D}}\right\}.
    \nonumber
\end{equation}

To keep the AVs in each batch following the free-flow region of HFD in the flow area in each link, the size of the batch should satisfy the following equation.

\begin{equation}
    B^{r^{\prime},s^{\prime}}<\mathop{min}\limits_{(i,j)\in P^{r^{\prime},s^{\prime}},~k\in[1,N]} \left\{\frac{L_{i.j}}{\textit{h}_{i,j;k}(t)\textit{v}_{i,j}^f+\emph{L}}
    \right\}
    \nonumber
\end{equation}

Therefore, the size of batch $B^{r^{\prime},s^{\prime}}$ are restricted as follows.

\begin{eqnarray}
    \sum_{k=1}^{N_{1}}d_{r^{\prime},s^{\prime}}(k) &=& D^{r^{\prime},s^{\prime}}=n^{r^{\prime},s^{\prime}} B^{r^{\prime},s^{\prime}} \nonumber \\
    B^{r^{\prime},s^{\prime}}&<&\mathop{min}\limits_{(i,j)\in P^{r^{\prime},s^{\prime}},~k\in[1,N]} \left\{\bar{Q}_{i,j}^{\mathcal{U}},\bar{Q}_{i,j}^{\mathcal{D}},\frac{L_{i.j}}{\textit{h}_{i,j;k}(t)\textit{v}_{i,j}^f+\emph{L}}
    \right\}
    \nonumber
\end{eqnarray}

After determining the size of each batch, the next step is to determine the flow rate of each batch. Suppose we release the batch at a constant rate, so $B^{r^{\prime},s^{\prime}}=N_{B}^{r^{\prime},s^{\prime}}~F_{B}^{r^{\prime},s^{\prime}}$, where $N_{B}^{r^{\prime},s^{\prime}}$ is the time of releasing all AVs in this batch and $F_{B}^{r^{\prime},s^{\prime}}$ is the flow rate of releasing AVs. In order to satisfy the inflow and outflow capacity constraints in Equation \ref{eq:ineq_2_dis}, $F_{B}^{r^{\prime},s^{\prime}}$ should satisfy the following equation.

\begin{equation}
     F_{B}^{r^{\prime},s^{\prime}}<\mathop{min}\limits_{(i,j)\in P^{r^{\prime},s^{\prime}}}\left\{\bar{C}_{i,j}^{\mu},\bar{C}_{i,j}^{\nu}\right\}. 
    \nonumber
\end{equation}




Then we calculate $N_{i,j}^{r^{\prime},s^{\prime}}$, which represents the time of passing link $(i,j)$ for batch $B^{r^{\prime},s^{\prime}}$ and $N_{i,j}^{r^{\prime},s^{\prime}}$ is comprised of two time periods.

The first time period refers to the case that we release the batch from the prior link to link $(i,j)$. In this time period, we have $u_{i,j}(k)=F_{B}^{r^{\prime},s^{\prime}}$ for $k \in [N_{B,i,j}^{r^{\prime},s^{\prime}},N_{B,i,j}^{r^{\prime},s^{\prime}}+N_{B}^{r^{\prime},s^{\prime}}]$, where $N_{B,i,j}^{r^{\prime},s^{\prime}}$ is the time ready to release the batch $B$ with O-D $(r^{\prime},s^{\prime})$ to link $(i,j)$. In this time period, the flow in the flow area follows the free-flow region in HFD, so we have the following equations.

\begin{eqnarray}
    f_{i,j}(k) &=& \textit{v}_{i,j}^{f}~ \rho_{i,j}(k) 
    \nonumber \\
    \rho_{i,j}^{s^{\prime}}(k)&=&\rho_{i,j}^{s^{\prime}}(k-1)+\Delta_{t}\frac{F_{B}^{r^{\prime},s^{\prime}}-\textit{f}_{i,j}^{~s^{\prime}}(k)}{L_{i,j}} 
    \nonumber
\end{eqnarray}

Combining the above two equations, we have

\begin{eqnarray}
\rho_{i,j}(k) &=& \frac{L_{i,j}~\rho_{i,j}(k-1)+\Delta_{t}~F_{B}^{r^{\prime},s^{\prime}}}{L_{i,j}+\Delta_{t}~\textit{v}_{i,j}^{f}} 
\nonumber \\
&=& a_{i,j}~\rho_{i,j}(k-1)+b_{i,j}^{B,r^{\prime},s^{\prime}}
\nonumber
\end{eqnarray}

where $a_{i,j}=\frac{L_{i,j}}{L_{i,j}+\Delta_{t}~\textit{v}_{i,j}^{f}}$
and $b_{i,j}^{B,r^{\prime},s^{\prime}}=\frac{\Delta_{t}~F_{B}^{r^{\prime},s^{\prime}}}{L_{i,j}+\Delta_{t}~\textit{v}_{i,j}^{f}}$.

Therefore, we have the formulation of density at time $N_{B,i,j}^{r^{\prime},s^{\prime}}+N_{B}^{r^{\prime},s^{\prime}}$ as follows. 

\begin{equation}
    \rho_{i,j}(N_{B,i,j}^{r^{\prime},s^{\prime}}+N_{B}^{r^{\prime},s^{\prime}})
    = b_{i,j}^{B,r^{\prime},s^{\prime}}(a_{i,j}^{(N_{B}^{r^{\prime},s^{\prime}})}+a_{i,j}^{(N_{B}^{r^{\prime},s^{\prime}}-1)}+\cdots+1) = c_{i,j}^{B,r^{\prime},s^{\prime}}
    \nonumber
\end{equation}

In the second period, all AVs in this batch have already been released and we calculate the time when all AVs in this batch enter the buffer area to satisfy the end constraint. So we have $u_{i,j}(k)=0$ for $k \in [N_{B,i,j}^{r^{\prime},s^{\prime}}+N_{B}^{r^{\prime},s^{\prime}},N_{B,i,j}^{r^{\prime},s^{\prime}}+N_{i,j}^{r^{\prime},s^{\prime}}]$. Then we have following equations in this time period.

\begin{eqnarray}
    f_{i,j}(k) &=& \textit{v}_{i,j}^{f}~ \rho_{i,j}(k) 
    \nonumber \\
    \rho_{i,j}^{s^{\prime}}(k)&=&\rho_{i,j}^{s^{\prime}}(k-1)-\Delta_{t}\frac{\textit{f}_{i,j}^{~s^{\prime}}(k)}{L_{i,j}} 
    \nonumber
\end{eqnarray}

Combining the above two equations, we have

\begin{eqnarray}
\rho_{i,j}(k) &=& \frac{L_{i,j}~\rho_{i,j}(k-1)}{L_{i,j}+\Delta_{t}~\textit{v}_{i,j}^{f}} 
\nonumber \\
&=& a_{i,j}~\rho_{i,j}(k-1)
\nonumber
\end{eqnarray}

To satisfy the end constraint in Equation \ref{ineq:end_dis}, the density at time $N_{B,i,j}^{r^{\prime},s^{\prime}}+N_{i,j}^{r^{\prime},s^{\prime}}$ should satisfy the following equation.

\begin{equation}
    \rho_{i,j}(N_{B,i,j}^{r^{\prime},s^{\prime}}+N_{i,j}^{r^{\prime},s^{\prime}})=c_{i,j}~a_{i,j}^{(N_{i,j}^{r^{\prime},s^{\prime}}-N_{B}^{r^{\prime},s^{\prime}})}<\frac{1}{L_{i,j}}
    \nonumber
\end{equation}

So we have

\begin{equation}
    N_{i,j}^{r^{\prime},s^{\prime}}>N_{B}^{r^{\prime},s^{\prime}}-\frac{\ln{c_{i,j}~L_{i,j}}}{\ln{a_{i,j}}}
    \nonumber
\end{equation}

Then we set $N_{i,j}^{r^{\prime},s^{\prime}}=\max\{N_{B}^{r^{\prime},s^{\prime}},\lceil N_{B}^{r^{\prime},s^{\prime}}-\frac{\ln{c_{i,j}~L_{i,j}}}{\ln{a_{i,j}}} \rceil\}$. It is obvious that other constraints are satisfied by constructing the solution based on the above procedures. Therefore, if we choose the time horizon $N$ so that $N$ is larger than $\textit{N}_\textit{1}+\sum\limits_{(r^{\prime},s^{\prime})}n^{r^{\prime},s^{\prime}}\sum\limits_{(i,j)}N_{i,j}^{r^{\prime},s^{\prime}}\delta_{i,j}^{P^{r^{\prime},s^{\prime}}}$, where $\delta_{i,j}^{P^{r^{\prime},s^{\prime}}}$ indicates whether link $(i,j)$ is in the path $P^{r^{\prime},s^{\prime}}$, then we could always find a feasible solution of proposed problem.

  \section{Sensitivity analysis of the headway control for SO-DTA}
  \label{app:Sensitivity_Analysis}
  
  This section proposes a sensitivity-based algorithm to solve the SO-DTA problem. We first transform the discretized version of the headway-dependent SO-DTA formulation into linear programming (LP)  given an exogenous headway variable $\textbf{h}$, as shown in \ref{app:MILP}, then the sensitivity-based algorithm is presented. 
  
 \subsection{LP formulation}
 \label{app:MILP}

 This section transforms the discretized version of SO-DTA problem formulation from Equations \ref{eq:obj_dis} to \ref{eq:FD_dis} and Equations \ref{eq:shockwave_dis1} to \ref{ineq:headway_dis} into LP. 
 
 The objective of section \ref{sec_system_optimal_headway_formulation} is to find the optimal headway $\left\{\textit{h}_{i,j}(k)\right\}$. In the proposed programming, the involved variables include $\left\{\rho_{i,j}(k)\right\}$, 
$\left\{\rho_{i,j}^{s^{\prime}}(k)\right\}$, 
$\left\{\textit{q}_{i,j}^{\mathcal{D}}(k)\right\}$,
$\left\{\textit{q}_{i,j}^{\mathcal{D},s^{\prime}}(k)\right\}$, $\left\{\textit{q}_{i,j}^{\mathcal{U}}(k)\right\}$, $\left\{\textit{q}_{i,j}^{\mathcal{U},s^{\prime}}(k)\right\}$,
$\left\{\textit{u}_{i,j}(k)\right\}$,
$\left\{\textit{u}_{i,j}^{s^{\prime}}(k)\right\}$,
$\left\{\textit{f}_{i,j}(k)\right\}$,
$\left\{\textit{f}_{i,j}^{~s^{\prime}}(k)\right\}$,
$\left\{\textit{v}_{i,j}(k)\right\}$,
$\left\{\textit{v}_{i,j}^{s^{\prime}}(k)\right\}$,
$\left\{\textit{n}_{i,j}(k)\right\}$,
$\left\{\textit{h}_{i,j}(k)\right\}$ and
$\left\{w_{i,j}(k)\right\}$.
However, the proposed problem is hard to solve due to the nonlinearity in Equations \ref{eq:FD_dis} and Equations \ref{eq:shockwave_dis1} to  \ref{eq:shockwave_dis2}.

We note that if the headway $\left\{\textit{h}_{i,j}(k)\right\}$ is fixed, then the variables are  $\left\{\rho_{i,j}(k)\right\}$, 
$\left\{\rho_{i,j}^{s^{\prime}}(k)\right\}$, 
$\left\{\textit{q}_{i,j}^{\mathcal{D}}(k)\right\}$, $\left\{\textit{q}_{i,j}^{\mathcal{D},s^{\prime}}(k)\right\}$, $\left\{\textit{q}_{i,j}^{\mathcal{U}}(k)\right\}$, $\left\{\textit{q}_{i,j}^{\mathcal{U},s^{\prime}}(k)\right\}$,
$\left\{\textit{u}_{i,j}(k)\right\}$,
$\left\{\textit{u}_{i,j}^{s^{\prime}}(k)\right\}$,
$\left\{\textit{f}_{i,j}(k)\right\}$,
$\left\{\textit{f}_{i,j}^{~s^{\prime}}(k)\right\}$,
$\left\{\textit{v}_{i,j}(k)\right\}$,
$\left\{\textit{v}_{i,j}^{s^{\prime}}(k)\right\}$,
$\left\{\textit{n}_{i,j}(k)\right\}$ and
$\left\{w_{i,j}(k)\right\}$. Hence the proposed formulation becomes a  linear programming program (LP).

For simplicity in notations, we use a generalized form of mixed integer linear programming to represent the formulation, as shown in Equation~\ref{eq:MILP_obj}.


\begin{eqnarray}
     &\min \limits_{\mathbf{x}} ~~  TTT = \textbf{P}^{T}\textbf{x} 
     \label{eq:MILP_obj}\\
    & \, s.t.  \quad \textbf{A}(\textbf{h})\textbf{x}  =  \textbf{B}(\textbf{h}) \nonumber \\
    & \quad \quad ~~ \textbf{C}(\textbf{h})\textbf{x} \leq \textbf{D}(\textbf{h}) \nonumber
\end{eqnarray}

where we have:

\begin{itemize} 
\item $\textbf{x}=\left\{\rho_{i,j}(k),~
\rho_{i,j}^{s^{\prime}}(k),~
\textit{q}_{i,j}^{\mathcal{D}}(k),~ \textit{q}_{i,j}^{\mathcal{D},s^{\prime}}(k),~
\textit{q}_{i,j}^{\mathcal{U}}(k),~ \textit{q}_{i,j}^{\mathcal{U},s^{\prime}}(k),~
\textit{u}_{i,j}(k),~
\textit{u}_{i,j}^{s^{\prime}}(k),~
\textit{f}_{i,j}(k),~
\textit{f}_{i,j}^{~s^{\prime}}(k),~
\textit{v}_{i,j}(k),~
\textit{v}_{i,j}^{s^{\prime}}(k),~
\textit{n}_{i,j}(k),~
w_{i,j}(k) \right\}$ for $(\textit{i}.\textit{j})\in \textit{E}, ~ 1\leq{k}\leq\textit{N};
\textit{s}^{\prime} \in \widetilde{S}$, and the dimension of variables $\textbf{x}$ is $n$.

\item $\textbf{h}=\left\{\textit{h}_{\emph{i},~\emph{j}}(k) \mid (\textit{i},\textit{j})\in E\setminus(L_R\cup L_S);~1\leq{k}\leq\textit{N}\right\} \in\Re^{m}$, and the dimension of parameters $\textbf{h}$ is $m$.
\item $\textbf{A}(\textbf{h})=
\left[
\begin{matrix}
a_{1}(\textbf{h})  \\
\vdots \\
a_{l_{1}}(\textbf{h})\\
\end{matrix}
\right]=
\left[
\begin{matrix}
a_{11}(\textbf{h})  & \cdots & a_{1n}(\textbf{h}) \\
\vdots  & \ddots & \vdots \\
a_{l_{1}1}(\textbf{h})  & \cdots & a_{l_{1}n}(\textbf{h})\\
\end{matrix}
\right] \in\Re^{l_{1}\times n}; \quad
\textbf{B}(\textbf{h})=
\left[
\begin{matrix}
b_{1}(\textbf{h})  \\
\vdots \\
b_{l_{1}}(\textbf{h})\\
\end{matrix}
\right] \in\Re^{l_{1}}
$. $l_{1}$ represents the number of equality constraints. 
\item $\textbf{C}(\textbf{h})=
\left[
\begin{matrix}
c_{1}(\textbf{h})  \\
\vdots \\
c_{l_{2}}(\textbf{h})\\
\end{matrix}
\right]=
\left[
\begin{matrix}
c_{11}(\textbf{h})  & \cdots & a_{1n}(\textbf{h}) \\
\vdots  & \ddots & \vdots \\
c_{l_{2}1}(\textbf{h})  & \cdots & c_{l_{2}n}(\textbf{h})\\
\end{matrix}
\right] \in\Re^{l_{2}\times n}; \quad
\textbf{D}(\textbf{h})=
\left[
\begin{matrix}
d_{1}(\textbf{h})  \\
\vdots \\
d_{l_{2}}(\textbf{h})\\
\end{matrix}
\right] \in\Re^{l_{2}}
$. $l_{2}$ represents the number of inequality constraints. 
\end{itemize}

\subsection{Sensitivity-based optimal headway control solution algorithm}
\label{app:SA}
  We then present the sensitivity analysis for the developed LP,
  and the focus is to derive the gradient of TTT with respect to the headway.
The Lagrange multiplier of Equation~\ref{eq:MILP_obj} is shown in Equation \ref{eq:KKT_obj} and the corresponding Karush–Kuhn–Tucker (KKT) conditions are presented from Equations \ref{eq:KKT_1} to \ref{eq:KKT_5}.

\begin{equation}
\min_\mathbf{x} \quad\textit{\textbf{L}}=\textbf{P}^{T}\textbf{x}+\bm{\lambda}^{T}\left(\textbf{A}(\textbf{h})\textbf{x}-\textbf{B}(\textbf{h})\right)+\bm{\mu}^{T}\left(\textbf{C}(\textbf{h})\textbf{x}-\textbf{D}(\textbf{h})\right)
\label{eq:KKT_obj}
\end{equation}

\begin{equation}
\textbf{P}^{T}+{\bm{\lambda^{*}}}^{T}\textbf{A}(\textbf{h})+
{\bm{\mu^{*}}}^{T}\textbf{C}(\textbf{h})=\bm{0}
\label{eq:KKT_1}
\end{equation}

\begin{equation}
\textbf{A}(\textbf{h})\textbf{x}^{*}=\textbf{B}(\textbf{h})
\label{eq:KKT_2}
\end{equation}

\begin{equation}
\textbf{C}(\textbf{h})\textbf{x}^{*}\leq\textbf{D}(\textbf{h})
\label{eq:KKT_3}
\end{equation}

\begin{equation}
{\bm{\mu^{*}}_{i}}\left(\textbf{C}(\textbf{h})\textbf{x}^{*}-\textbf{D}(\textbf{h})\right)_{i}=\bm{0}
\label{eq:KKT_4}
\end{equation}

\begin{equation}
\bm{\mu^{*}}\geq\bm{0}
\label{eq:KKT_5}
\end{equation}

where,

\begin{itemize} 
\item $\bm{\lambda^{*}}\in\Re^{l_{1}}$ and $\bm{\mu^{*}}\in\Re^{l_{2}}$. 
\item  $\textit{K}=\left\{i \mid \mu_{\textit{i}}^{*}>0 ; 1\leq\textit{i}\leq{l_{2}}\right\}$ is the set of index of all positive 
$\mu_{\textit{i}}^{*}$. And the size of $\textit{K}$ is $\textit{k}$. If i $\in \textit{K}$, then $\mu_{\textit{i}}^{*}>0$.
\end{itemize} 

Then we do the differential to the Equation \ref{eq:MILP_obj}  as $\mathbf{z}=\textbf{P}^{T}\textbf{X}$  and Equations \ref{eq:KKT_1} to \ref{eq:KKT_4}.

\begin{equation}
\textbf{P}^{T}d\textbf{x}-d\textbf{z}=0
\label{eq:KKT_dif_1}
\end{equation}

\begin{equation}
\left(\sum_{i=1}^{l_{1}}\lambda_{\textit{i}}^{*}\nabla_{\textbf{h}}a_{i}(\textbf{h})+
\sum_{j=1}^{l_{2}}\mu_{\textit{j}}^{*}\nabla_{\textbf{h}}c_{j}(\textbf{h})\right)d\bm{\textbf{h}}+\left(\textbf{A}(\textbf{h})\right)^{T}d\bm{\lambda}+\left(\textbf{C}(\textbf{h})\right)^{T}d\bm{\mu}=\bm{0}
\label{eq:KKT_dif_2}
\end{equation}

\begin{equation}
\textbf{A}(\textbf{h})d\textbf{x}+\left[\nabla_{\textbf{h}}\left(\textbf{A}(\textbf{h})\textbf{x}^{*}\right)-\nabla_{\textbf{h}}\textbf{B}(\textbf{h})\right]d\bm{\textbf{h}}=\bm{0}
\label{eq:KKT_dif_3}
\end{equation}

\begin{equation}
\textbf{C}_{\textbf{k}}(\textbf{h})d\textbf{x}+\left[\nabla_{\textbf{h}}\left(\textbf{C}_{\textbf{k}}(\textbf{h})\textbf{x}^{*}\right)-\nabla_{\textbf{h}}\textbf{D}_{\textbf{k}}(\textbf{h})\right]d\bm{\textbf{h}}=\bm{0}
\label{eq:KKT_dif_4}
\end{equation}

where,

\begin{itemize} 
\item $\nabla_{\textbf{h}}a_{i}(\textbf{h})=\nabla_{\textbf{h}}\left[
\begin{matrix}
a_{i1}(\textbf{h}) & \cdots &  a_{in}(\textbf{h})\\
\end{matrix}
\right]=\left[
\begin{matrix}
\frac{\partial a_{i1}(\textbf{h})}{\partial \textbf{h}_{1}} & \cdots &  \frac{\partial a_{i1}(\textbf{h})}{\partial \textbf{h}_{m}}\\
\vdots  & \ddots & \vdots \\
\frac{\partial a_{in}(\textbf{h})}{\partial \textbf{h}_{1}} & \cdots &  \frac{\partial a_{in}(\textbf{h})}{\partial \textbf{h}_{m}}\\
\end{matrix}\right] \in\Re^{n\times m}$
\item $\nabla_{\textbf{h}}c_{j}(\textbf{h})=\nabla_{\textbf{h}}\left[
\begin{matrix}
c_{j1}(\textbf{h}) & \cdots &  c_{jn}(\textbf{h})\\
\end{matrix}
\right]=\left[
\begin{matrix}
\frac{\partial c_{j1}(\textbf{h})}{\partial \textbf{h}_{1}} & \cdots &  \frac{\partial c_{j1}(\textbf{h})}{\partial \textbf{h}_{m}}\\
\vdots  & \ddots & \vdots \\
\frac{\partial c_{jn}(\textbf{h})}{\partial \textbf{h}_{1}} & \cdots &  \frac{\partial c_{jn}(\textbf{h})}{\partial \textbf{h}_{m}}\\
\end{matrix}\right] \in\Re^{n\times m}$
\item $\textbf{A}^{x}=\textbf{A}(\textbf{h})\textbf{x}^{*}=\left[
\begin{matrix}
a^{x}_{1}(\textit{x}^{*},\textbf{h}) \\
\vdots \\
a^{x}_{l_{1}}(\textit{x}^{*},\textbf{h})\\
\end{matrix}\right]\in\Re^{l_{1}}$;
$\nabla_{\textbf{h}}\textbf{A}^{x}=\left[
\begin{matrix}
\frac{\partial a^{x}_{1}(\textit{x}^{*},\textbf{h})}{\partial \textbf{h}_{1}} & \cdots &  \frac{\partial a^{x}_{1}(\textit{x}^{*},\textbf{h})}{\partial \textbf{h}_{m}}\\
\vdots  & \ddots & \vdots \\
\frac{\partial a^{x}_{l_{1}}(\textit{x}^{*},\textbf{h})}{\partial \textbf{h}_{1}} & \cdots &  \frac{\partial a^{x}_{l_{1}}(\textit{x}^{*},\textbf{h})}{\partial \textbf{h}_{m}}\\
\end{matrix}\right]\in\Re^{l_{1}\times m}$
\item $\textbf{C}_{\textbf{k}}(\textbf{h})=\left[
\begin{matrix}
\vdots  \\
c_{j}(\textbf{h})  \\
\vdots\\
\end{matrix}
\right]_{j \in \textit{K}} \in\Re^{k\times n}$; ~
$\textbf{C}_{\textbf{k}}^{x}=\textbf{C}_{\textbf{k}}(\textbf{h})\textbf{x}^{*}= \left[ 
\begin{matrix}
\vdots  \\
c_{j}^{x}(\textit{x}^{*},\textbf{h})  \\
\vdots\\
\end{matrix}
\right]_{j \in \textit{K}} \in\Re^{k}$; ~
$\textbf{D}_{\textbf{k}}(\textbf{h})=\left[
\begin{matrix}
\vdots  \\
d_{j}(\textbf{h})  \\
\vdots\\
\end{matrix}
\right]_{j \in \textit{K}} \in\Re^{k}$
\item $\nabla_{\textbf{h}}\textbf{C}_{\textbf{k}}^{x}=\left[
\begin{matrix}
\vdots & \cdots & \vdots \\
\frac{\partial c_{j}^{x}(\textit{x}^{*},\textbf{h}))}{\partial \textbf{h}_{1}} & \cdots & \frac{\partial c_{j}^{x}(\textit{x}^{*},\textbf{h}))}{\partial \textbf{h}_{m}} \\
\vdots & \cdots & \vdots \\
\end{matrix}
\right]_{j \in \textit{K}} \in\Re^{k\times m}$
\end{itemize}

We cast Equations \ref{eq:KKT_dif_1} to \ref{eq:KKT_dif_4} into the form of block matrices.

\begin{equation}
\left[
\begin{matrix}
\textbf{P}^{T} & \bm{0} & \bm{0} & -1 \\ 
\bm{0} & \left(\textbf{A}(\textbf{h})\right)^{T} & \left(\textbf{C}(\textbf{h})\right)^{T} & \bm{0} \\ 
\textbf{A}(\textbf{h}) & \bm{0} & \bm{0} & \bm{0} \\
\textbf{C}_{\textbf{k}}(\textbf{h}) & \bm{0} & \bm{0} & \bm{0} \\
\end{matrix}
\right] ~
\left[
\begin{matrix}
d\textbf{x} \\ 
d\bm{\lambda} \\ 
d\bm{\mu} \\
d\textbf{z} \\
\end{matrix}
\right]\quad = \quad -
\left[
\begin{matrix}
\bm{0} \\ 
\sum_{i=1}^{l_{1}}\lambda_{\textit{i}}^{*}\nabla_{\textbf{h}}a_{i}(\textbf{h})+
\sum_{j=1}^{l_{2}}\mu_{\textit{j}}^{*}\nabla_{\textbf{h}}c_{j}(\textbf{h}) \\ 
\nabla_{\textbf{h}}\left(\textbf{A}(\textbf{h})\textbf{x}^{*}\right)-\nabla_{\textbf{h}}\textbf{B}(\textbf{h}) \\
\nabla_{\textbf{h}}\left(\textbf{C}_{\textbf{k}}(\textbf{h})\textbf{x}^{*}\right)-\nabla_{\textbf{h}}\textbf{D}_{\textbf{k}}(\textbf{h}) \\
\end{matrix}
\right] ~ d\bm{\textbf{h}}
\label{matrix:1}
\end{equation}

We further define 
\begin{equation} \textbf{Q}=
\left[
\begin{matrix}
\textbf{P}^{T} & \bm{0} & \bm{0} & -1 \\ 
\bm{0} & \left(\textbf{A}(\textbf{h})\right)^{T} & \left(\textbf{C}(\textbf{h})\right)^{T} & \bm{0} \\ 
\textbf{A}(\textbf{h}) & \bm{0} & \bm{0} & \bm{0} \\
\textbf{C}_{\textbf{k}}(\textbf{h}) & \bm{0} & \bm{0} & \bm{0} \\
\end{matrix}
\right].
\label{matrix:2}
\end{equation}

The dimension of matrix $\textbf{Q}$ is $\left(n+l_{1}+k+1\right)\times\left(n+l_{1}+l_{2}+1\right)$. If there exists $\mu_{\textit{i}}^{*}=0$, then the matrix $\textbf{Q}$ is not a square matrix and hence not invertible. So we calculate the generalized inverse matrix of $\textbf{Q}$ and denote it as $\textbf{Q}^{-1}$ \citep{penrose1955generalized}.

\begin{equation}
\left[
\begin{matrix}
d\textbf{x} \\ 
d\bm{\lambda} \\ 
d\bm{\mu} \\
d\textbf{z} \\
\end{matrix}
\right]\quad = \quad - \textbf{Q}^{-1}
\left[
\begin{matrix}
\bm{0} \\ 
\sum_{i=1}^{l_{1}}\lambda_{\textit{i}}^{*}\nabla_{\textbf{h}}a_{i}(\textbf{h})+
\sum_{j=1}^{l_{2}}\mu_{\textit{j}}^{*}\nabla_{\textbf{h}}c_{j}(\textbf{h}) \\ 
\nabla_{\textbf{h}}\left(\textbf{A}(\textbf{h})\textbf{x}^{*}\right)-\nabla_{\textbf{h}}\textbf{B}(\textbf{h}) \\
\nabla_{\textbf{h}}\left(\textbf{C}_{\textbf{k}}(\textbf{h})\textbf{x}^{*}\right)-\nabla_{\textbf{h}}\textbf{D}_{\textbf{k}}(\textbf{h}) \\
\end{matrix}
\right] ~ d\bm{\textbf{h}} 
\label{matrix:3}
\end{equation}

Then we replace $d\bm{\textbf{h}}$ with the identity matrix $\textit{I}_{m}$ and derive the gradient $\frac{\partial \textbf{z}}{\partial \textbf{h}}$.

\begin{equation}
\left[
\begin{matrix}
\frac{\partial \textbf{x}}{\partial \textbf{h}}\vspace{1ex} \\ 
\frac{\partial \bm{\lambda}}{\partial \textbf{h}} \vspace{1ex} \\ 
\frac{\partial \bm{\mu}}{\partial \textbf{h}} \vspace{1ex}\\
\frac{\partial \textbf{z}}{\partial \textbf{h}} \\
\end{matrix}
\right] \quad = \quad - \textbf{Q}^{-1}
\left[
\begin{matrix}
\bm{0} \\ 
\sum_{i=1}^{l_{1}}\lambda_{\textit{i}}^{*}\nabla_{\textbf{h}}a_{i}(\textbf{h})+
\sum_{j=1}^{l_{2}}\mu_{\textit{j}}^{*}\nabla_{\textbf{h}}c_{j}(\textbf{h}) \\ 
\nabla_{\textbf{h}}\left(\textbf{A}(\textbf{h})\textbf{x}^{*}\right)-\nabla_{\textbf{h}}\textbf{B}(\textbf{h}) \\
\nabla_{\textbf{h}}\left(\textbf{C}_{\textbf{k}}(\textbf{h})\textbf{x}^{*}\right)-\nabla_{\textbf{h}}\textbf{D}_{\textbf{k}}(\textbf{h}) \\
\end{matrix}
\right] 
\label{matrix:4}
\end{equation}

Therefore, a sensitivity-based optimal headway control algorithm is proposed based on exploring the gradient descent of TTT, as shown in Algorithm \ref{algo:sen}.

\begin{algorithm}[H]
\caption{Sensitivity-based Optimal Headway Solving Algorithm}
{\textbf{Input:}  learning rate $\eta$}.

{\textbf{Output:} time-dependent and link-specific headway $\left\{\textit{h}_{\emph{i},\,\emph{j}}(k) \right\}$}.

Initialize the headway $\textbf{h}$.
\begin{algorithmic}[1]
\FOR{$(\textit{i}=0; \textit{i}<\textit{I}; ++\textit{i})$}
\STATE Solve the Formulation~\ref{eq:MILP_obj} given the current headway setting $\textbf{h}$ and derive $\textbf{z}^{*}, {\lambda}^{*}, {\mu}^{*}$ and $\textbf{x}^{*}$.
\STATE Conduct the sensitivity analysis and derive $ \frac{\partial \textbf{z}^{*}}{\partial \textbf{h}} $.
\STATE Update the headway by $\textbf{h}=\textbf{h}-\eta\frac{\partial \textbf{z}^{*}}{\partial \textbf{h}} $.
\ENDFOR
\end{algorithmic}
\textbf{Return:} optimal headway $\textbf{h}^{*}$.
\label{algo:sen}
\end{algorithm}

  \section{Proof of Proposition \ref{prop:minimum}}
  \label{app:minimum}

We first present Lemma \ref{lemma:1} to show that a smaller headway $h_{i,j}(k)$ could generate a larger flow $f_{i,j}(k)$. 

\begin{lemma}
Given $u_{i,j}(k)$ and $\rho_{i,j}(k-1)$, if $h_{i,j}^{(1)}(k)\leq h_{i,j}^{(2)}(k)$, then $f_{i,j}^{(1)}(k)\geq f_{i,j}^{(2)}(k)$.
\label{lemma:1}
\end{lemma}


\begin{proof}
We solve  $f_{i,j}(k)$ and $\rho_{i,j}(k)$ under given $u_{i,j}(k)$ and $\rho_{i,j}(k-1)$ in both free-flow and congested states:

\begin{itemize}
\item Free-flow state:

\begin{eqnarray}
     \textit{f}_{i,j}(k) & =&  \textit{v}_{i,j}^{f}{~} \rho_{i,j}(k) \nonumber \\
 \rho_{i,j}(k) & =&  \rho_{i,j}(k-1) +\Delta_{t}\frac{\textit{u}_{i,j}(k)-\textit{f}_{i,j}(k)}{L_{i,j}} \nonumber
\end{eqnarray}


Then we obtain the following equations:

\begin{eqnarray}
 \textit{f}_{i,j}(k) & =&  \frac{L_{i,j}~\textit{v}_{i,j}^{f}~\rho_{i,j}(k-1)+\Delta_{t}~\textit{v}_{i,j}^{f}~u_{i,j}(k)}{L_{i,j}+\Delta_{t}~\textit{v}_{i,j}^{f}} \nonumber\\
 \rho_{i,j}(k) & =&  \frac{L_{i,j}~\rho_{i,j}(k-1)+\Delta_{t}~u_{i,j}(k)}{L_{i,j}+\Delta_{t}~\textit{v}_{i,j}^{f}} \nonumber
\end{eqnarray}


\item Congested state:

\begin{eqnarray}
 \textit{f}_{i,j}(k) & =&  \frac{1-\rho_{i,j}(k)~L}{h_{i,j}(k)} \nonumber \\
 \rho_{i,j}(k) & =&  \rho_{i,j}(k-1)+\Delta_{t}\frac{\textit{u}_{i,j}(k)-\textit{f}_{i,j}(k)}{L_{i,j}} \nonumber
\end{eqnarray}


Then the two quantities can be derived as:

\begin{eqnarray}
 \textit{f}_{i,j}(k) & =&  \frac{L_{i,j}-L~L_{i,j}~\rho_{i,j}(k-1)-L~\Delta_{t}~u_{i,j}(k)}{L_{i,j}~h_{i,j}(k)-L~\Delta_{t}}  \nonumber\\
 \rho_{i,j}(k) & =&  \frac{h_{i,j}(k)~L_{i,j}~\rho_{i,j}(k-1)+h_{i,j}(k)~\Delta_{t}~u_{i,j}(k)-\Delta_{t}}{L_{i,j}~h_{i,j}(k)-L~\Delta_{t}} \nonumber    
\end{eqnarray}


\end{itemize} 

Next, We will discuss different states of $\rho_{i,j}^{(2)}(k)$:

\begin{itemize}

\item If $\rho_{i,j}^{(2)}(k)$ is under free-flow state:
Because $h_{i,j}^{(1)}(k)\leq h_{i,j}^{(2)}(k)$,  as is shown in Figure \ref{fig:FD}, it is obvious that $f_{i,j}^{(1)}(k)=f_{i,j}^{(2)}(k)$.

\item If $\rho_{i,j}^{(2)}(k)$ is under congested state:
\begin{itemize}
    \item The density in the free-flow state  should be infeasible:
       \begin{equation}
            \rho_{i,j}^{(2)}(k)=\frac{L_{i,j}~\rho_{i,j}^{(2)}(k-1)+\Delta_{t}~u_{i,j}^{(2)}(k)}{L_{i,j}+\Delta_{t}~\textit{v}_{i,j}^{f}}\geq \frac{1}{h_{i,j}^{(2)}(k)~\textit{v}_{i,j}^{f}+L}
        \end{equation}
    \item The density in the congested state should be feasible:
    \begin{equation}
         \rho_{i,j}^{(2)}(k)=\frac{h_{i,j}^{(2)}(k)~L_{i,j}~\rho_{i,j}^{
(2)}(k-1)+h_{i,j}^{(2)}(k)~\Delta_{t}~u_{i,j}^{(2)}(k)-\Delta_{t}}{L_{i,j}~h_{i,j}^{(2)}(k)-L~\Delta_{t}} \geq \frac{1}{h_{i,j}^{(2)}(k)~\textit{v}_{i,j}^{f}+L}
    \end{equation}
\end{itemize}

For the free-flow state, if $L_{i,j}~h_{i,j}^{(2)}(k)-L~\Delta_{t}<0$, we have

\begin{equation}
\begin{aligned}
\rho_{i,j}^{(2)}(k)&=\frac{h_{i,j}^{(2)}(k)~L_{i,j}~\rho_{i,j}^{(2)}(k-1)+h_{i,j}^{(2)}(k)~\Delta_{t}~u_{i,j}^{(2)}(k)-\Delta_{t}}{L_{i,j}~h_{i,j}^{(2)}(k)-L~\Delta_{t}}\\
&\leq \frac{\frac{L_{i,j}+\Delta_{t}~\textit{v}_{i,j}^{f}}{h_{i,j}^{(2)}(k)~\textit{v}_{i,j}^{f}+L}~h_{i,j}^{(2)}(k)-\Delta_{t}}{L_{i,j}~h_{i,j}^{(2)}(k)-L~\Delta_{t}} = \frac{1}{h_{i,j}^{(2)}(k)~\textit{v}_{i,j}^{f}+L}
\nonumber
\end{aligned}
\end{equation}

It is contradicted with $\rho_{i,j}^{(2)}(k) \geq \frac{1}{h_{i,j}^{(2)}(k)~\textit{v}_{i,j}^{f}+L}$. Therefore, $L_{i,j}~h_{i,j}^{(2)}(k)-L~\Delta_{t}>0$. It is also similar to prove that if $\rho_{i,j}(k)$ is congested under the headway $h_{i,j}(k)$, then $L_{i,j}~h_{i,j}(k)-L~\Delta_{t}>0$. 

Next, we discuss the states of $\rho_{i,j}^{(1)}(k)$.

\begin{itemize}
\item[*] $\rho_{i,j}^{(1)}(k)$ is in congested state, then we have

\begin{equation}
\rho_{i,j}^{(1)}(k)=\frac{L_{i,j}~\rho_{i,j}^{(2)}(k-1)+\Delta_{t}~u_{i,j}^{(2)}(k)}{L_{i,j}+\Delta_{t}~\textit{v}_{i,j}^{f}}\geq \frac{1}{h_{i,j}^{(1)}(k)~\textit{v}_{i,j}^{f}+L}\geq
\frac{1}{h_{i,j}^{(2)}(k)~\textit{v}_{i,j}^{f}+L}
\nonumber
\end{equation}

The congested state of $\rho_{i,j}^{(1)}(k)$ can be calculated as follows:

\begin{equation}
\begin{aligned}
\rho_{i,j}^{(1)}(k)&=\frac{h_{i,j}^{(1)}(k)~L_{i,j}~\rho_{i,j}^{(2)}(k-1)+h_{i,j}^{(2)}(k)~\Delta_{t}~u_{i,j}^{(2)}(k)-\Delta_{t}}{L_{i,j}~h_{i,j}^{(1)}(k)-L~\Delta_{t}}\\
&\geq \frac{\frac{L_{i,j}+\Delta_{t}~\textit{v}_{i,j}^{f}}{h_{i,j}^{(1)}(k)~\textit{v}_{i,j}^{f}+L}~h_{i,j}^{(1)}(k)-\Delta_{t}}{L_{i,j}~h_{i,j}^{(1)}(k)-L~\Delta_{t}} = \frac{1}{h_{i,j}^{(1)}(k)~\textit{v}_{i,j}^{f}+L}
\nonumber
\end{aligned}
\end{equation}

The flow  $\textit{f}_{i,j}(k)=\frac{L_{i,j}-L~L_{i,j}~\rho_{i,j}(k-1)-L~\Delta_{t}~u_{i,j}(k)}{L_{i,j}~h_{i,j}(k)-L~\Delta_{t}}$ and we have proved that both numerator and denominator are positive, so if $h_{i,j}^{(1)}(k)\leq h_{i,j}^{(2)}(k)$, we have $f_{i,j}^{(1)}(k)\geq f_{i,j}^{(2)}(k)$.

\item[*] $\rho_{i,j}^{(1)}(k)$ is in the free-flow state, then the free-flow state of $\rho_{i,j}^{(1)}(k)$ should be feasible, as shown in the following equation.

\begin{equation}
\frac{1}{h_{i,j}^{(2)}(k)~\textit{v}_{i,j}^{f}+L} \leq \rho_{i,j}^{(1)}(k)=\frac{L_{i,j}~\rho_{i,j}^{(2)}(k-1)+\Delta_{t}~u_{i,j}^{(2)}(k)}{L_{i,j}+\Delta_{t}~\textit{v}_{i,j}^{f}} \leq \frac{1}{h_{i,j}^{(1)}(k)~\textit{v}_{i,j}^{f}+L}
\nonumber
\end{equation}

Therefore, for the flow $f_{i,j}^{(1)}(k)$, we have:

\begin{equation}
f_{i,j}^{(1)}(k)=\frac{L_{i,j}~\textit{v}_{i,j}^{f}~\rho_{i,j}^{(2)}(k-1)+\Delta_{t}~\textit{v}_{i,j}^{f}~u_{i,j}^{(2)}(k)}{L_{i,j}+\Delta_{t}~\textit{v}_{i,j}^{f}}\geq \frac{\textit{v}_{i,j}^{f}}{h_{i,j}^{(2)}(k)~\textit{v}_{i,j}^{f}+L} \geq f_{i,j}^{(2)}(k)
\nonumber
\end{equation}

\end{itemize} 
\end{itemize}
Combining the above proofs for the two states, we have $f_{i,j}^{(1)}(k)\geq f_{i,j}^{(2)}(k)$ hold.
\end{proof}

Then we develop Lemma \ref{lemma:2} to prove that the assumption of Lemma \ref{lemma:1} is reachable, which means that when $\textit{h}_{i,j}^{1}(k)\leq\textit{h}_{i,j}^{2}(k)$, we could always find the feasible $u_{i,j}(k)$ and $\rho_{i,j}(k-1)$ for both $\textit{h}_{i,j}^{1}(k)$ and $\textit{h}_{i,j}^{2}(k)$.

\begin{lemma}
For $\textbf{h}_{1}=\left\{\textit{h}_{i,j}^{(1)}(k) \right\}$ and $\textbf{h}_{2}=\left\{\textit{h}_{i,j}^{(2)}(k) \right\}$, if there exist only one $(i,j)$ and $k$ such that $\textit{h}_{i,j}^{(1)}(k) \leq \textit{h}_{i,j}^{(2)}(k)$ and $\textit{h}_{m,n}^{(1)}(l) = \textit{h}_{m,n}^{(2)}(l)$ for other links $(m,n)\neq(i,j)$ or time intervals $l\neq k$, suppose $u_{i,j}^{~ s^{\prime},(2)}(k) \in \textbf{x}_{2}, \forall s^{\prime}$, where $\textbf{x}_{2} \in \Omega_{\textbf{h}_{2}}$, then there always exists  $\textbf{x}_{1} \in \Omega_{\textbf{h}_{1}}$, such that $u_{i,j}^{~ s^{\prime},(1)}(k) \in \textbf{x}_{1}$ and $ u_{i,j}^{~ s^{\prime},(1)}(k)=u_{i,j}^{~ s^{\prime},(2)}(k)$.
\label{lemma:2}
\end{lemma}

\begin{proof}
In addition to the flow $f_{i,j}(k)$, headway affects the wave travel time in congested states $n_{i,j}^{w}(k)$ as follows:

\begin{eqnarray}
n_{i,j}^{w}(k) \Delta_{t} \frac{\textit{L}}{\textit{h}_{i,j}(k)}  &\leq& L_{i,j}, 
\nonumber \\
(n_{i,j}^{w}(k)+1) \Delta_{t} \frac{\textit{L}}{\textit{h}_{i,j}(k)}  & \geq  \color{black}& L_{i,j}. 
\nonumber
\end{eqnarray}

To show the existence of a feasible $\textbf{x}_{1}$ such that $u_{i,j}^{~ s^{\prime},(1)}(k)=u_{i,j}^{~ s^{\prime},(2)}(k)$, we need to show that $u_{i,j}^{~ s^{\prime},(2)}(k)$ satisfies the upstream queue capacity under $\textit{h}_{i,j}^{(1)}(k)$ and other related variables in $\textbf{x}_{2}$ if we set $\textit{u}_{i,j}^{~ s^{\prime},(1)}(l)=\textit{u}_{i,j}^{~ s^{\prime},(2)}(l), \forall 1\leq l \leq k-1$. Therefore, we could also derive $\textit{f}_{i,j}^{~ s^{\prime},(1)}(l)=\textit{f}_{i,j}^{~ s^{\prime},(2)}(l), \forall 1\leq l \leq k-1$ by both the proof of Lemma \ref{lemma:1} and the assumption that $\textit{h}_{i,j}^{(1)}(l)=\textit{h}_{i,j}^{(2)}(l), \forall 1\leq l \leq k-1$.

If $h_{i,j}^{(1)}(k)\leq h_{i,j}^{(2)}(k)$, then $n_{i,j}^{w,(1)}(k)\leq n_{i,j}^{w,(2)}(k)$, which implies that if we use a smaller the headway in link $(i,j)$ at the $k_{th}$ time interval, the wave travel time in congested states would be smaller. By Equation \ref{eq:uq_dis}, we have $\textit{q}_{i,j}^{\mathcal{U},s^{\prime},(1)}(k) \leq \textit{q}_{i,j}^{\mathcal{U},s^{\prime},(2)}(k)$. In order to keep $\textit{u}_{i,j}^{~ s^{\prime},(2)}(k)$ unchanged under $h_{i,j}^{(1)}(k)$, we next prove $\textit{q}_{i,j}^{\mathcal{U},s^{\prime},(1)}(k)\geq0$.

\begin{eqnarray}
    \textit{q}_{i,j}^{\mathcal{U},s^{\prime},(1)}(k) &=& \sum_{l=0}^{\textit{k}}\Delta_{t} \textit{u}_{i,j}^{~ s^{\prime},(1)}(l) - \sum_{l=0}^{\textit{k}-n_{i,j}^{w,(1)}(k)}\Delta_{t} \textit{f}_{i,j}^{~ s^{\prime},(1)}(l) 
    \nonumber \\
    &=& \sum_{l=0}^{\textit{k}}\Delta_{t} \textit{u}_{i,j}^{~ s^{\prime},(2)}(l) - \sum_{l=0}^{\textit{k}-n_{i,j}^{w,(1)}(k)}\Delta_{t} \textit{f}_{i,j}^{~ s^{\prime},(1)}(l) 
    \nonumber \\
    &\geq&
    \sum_{l=0}^{\textit{k}}\Delta_{t} \textit{u}_{i,j}^{~ s^{\prime},(2)}(l) - \sum_{l=0}^{\textit{k}-1}\Delta_{t} \textit{f}_{i,j}^{~ s^{\prime},(1)}(l) 
    \nonumber \\
    &=&
    \sum_{l=0}^{\textit{k}}\Delta_{t} \textit{u}_{i,j}^{~ s^{\prime},(2)}(l) - \sum_{l=0}^{\textit{k}-1}\Delta_{t} \textit{f}_{i,j}^{~ s^{\prime},(2)}(t) 
    \nonumber \\
    &=& L_{i,j}\rho_{i,j}^{s^{\prime},(2)}(k-1)+\Delta_{t}\textit{u}_{i,j}^{~ s^{\prime},(2)}(k) \geq 0 
    \nonumber
\end{eqnarray}
\end{proof}




By Lemma \ref{lemma:1} and \ref{lemma:2}, for link $(j,s)$, where $s \in S$,  if $h_{j,s}^{(1)}(k)\leq h_{j,s}^{(2)}(k), h_{j,s}^{(1)}(l) = h_{j,s}^{(2)}(l), \forall l \neq k$  
and we set $u_{j,s}^{(1)}(l)=u_{j,s}^{(2)}(l)$ for all $1 \leq l \leq N$, we have $f_{j,s}^{(1)}(k)\geq f_{j,s}^{(2)}(k)$.
For density, we have the following equation hold:

\begin{equation}
\rho_{j,s}(k)=\rho_{j,s}(k-1)+\Delta_{t}\frac{\textit{u}_{j,s}(k)-\textit{f}_{j,s}(k)}{L_{j,s}}. 
\label{eq:den_app}
\end{equation}

If $\rho_{j,s}^{(1)}(k-1)=\rho_{j,s}^{(2)}(k-1)$, then we have $\rho_{i,j}^{(1)}(k) \leq \rho_{i,j}^{(2)}(k)$. If $\rho_{i,j}^{(2)}(k+1)$ is in the free-flow state, the following equations hold in the proof of Lemma \ref{lemma:1} for $\rho_{j,s}^{(2)}(k+1)$ and $f_{j,s}^{(2)}(k+1)$:

\begin{eqnarray}
 \textit{f}_{j,s}(k) & =&  \frac{L_{j,s}~\textit{v}_{j,s}^{f}~\rho_{j,s}(k-1)+\Delta_{t}~\textit{v}_{j,s}^{f}~u_{j,s}(k)}{L_{j,s}+\Delta_{t}~\textit{v}_{j,s}^{f}} \nonumber\\
 \rho_{j,s}(k) & =&  \frac{L_{j,s}~\rho_{j,s}(k-1)+\Delta_{t}~u_{j,s}(k)}{L_{j,s}+\Delta_{t}~\textit{v}_{j,s}^{f}} \nonumber 
\end{eqnarray}


So we have $f_{i,j}^{(1)}(k+1) \leq f_{i,j}^{(2)}(k+1)$ and $\rho_{i,j}^{(1)}(k+1) \leq \rho_{i,j}^{(2)}(k+1)$ if $\rho_{i,j}^{(2)}(k+1)$ hold  in the free-flow state. 
Then if $\rho_{i,j}^{(2)}(k+1)$ is in the congested state, we have the following equations in the proof of Lemma \ref{lemma:1} hold for $\rho_{j,s}^{(2)}(k+1)$ and $f_{j,s}^{(2)}(k+1)$:

\begin{eqnarray}
   \textit{f}_{j,s}(k) & =&  \frac{L_{j,s}-L~L_{j,s}~\rho_{j,s}(k-1)-L~\Delta_{t}~u_{j,s}(k)}{L_{j,s}~h_{j,s}(k)-L~\Delta_{t}} \nonumber\\
 \rho_{j,s}(k) & =&  \frac{h_{j,s}(k)~L_{j,s}~\rho_{j,s}(k-1)+h_{j,s}(k)~\Delta_{t}~u_{j,s}(k)-\Delta_{t}}{L_{j,s}~h_{j,s}(k)-L~\Delta_{t}}  \nonumber
\end{eqnarray}


Similarly,   $f_{i,j}^{(1)}(k+1) \geq f_{i,j}^{(2)}(k+1)$ and $\rho_{i,j}^{(1)}(k+1) \leq \rho_{i,j}^{(2)}(k+1)$ hold for  $\rho_{i,j}^{(2)}(k+1)$ in the congested state. 

Therefore, we have $\rho_{i,j}^{(1)}(k+1) \leq \rho_{i,j}^{(2)}(k+1)$. We conduct the same procedure for $k+2,\cdots,N$ and $\rho_{i,j}^{(1)}(l) \leq \rho_{i,j}^{(2)}(l)$ is derived for $l \geq k$ and $\rho_{i,j}^{(1)}(l) = \rho_{i,j}^{(2)}(l)$ for $l < k$. Besides, by $F_{j,s}(k)=L_{j,s}\rho_{j,s}(k)+U_{j,s}(k)$, so we have $F_{j,s}^{(1)}(l)=F_{j,s}^{(2)}(l)$ for $l < k$ and $F_{j,s}^{(1)}(l) \geq F_{j,s}^{(2)}(l)$ for $l \geq k$. 

Similarly, we could also get the $\rho_{j,s}^{s_{0}^{\prime},(1)}(l) \leq \rho_{j,s}^{s_{0}^{\prime},(2)}(l)$ for $l \geq k$ and $\rho_{i,j}^{s_{0}^{\prime},(1)}(l) = \rho_{i,j}^{s_{0}^{\prime},(2)}(l)$ for $l < k$ by the same procedure for any destination $s_{0}$. Besides,  for any destination $s_{0}$, we also have $F_{j,s}^{s_{0}^{\prime},(1)}(k)=F_{j,s}^{s_{0}^{\prime},(2)}(k)$ for $l < k$ and $F_{j,s}^{s_{0}^{\prime},(1)}(k) \geq F_{j,s}^{s_{0}^{\prime},(2)}(k)$ for $l \geq k$. 

Therefore, for destination $s_{0} \neq s$, we could set $v_{j,s}^{s_{0}^{\prime},(1)}(l)=v_{j,s}^{s_{0}^{\prime},(2)}(l),\forall 1 \leq l \leq N$. For destination $s$,  we  have $V_{j,s}^{s^{\prime},(1)}(k)=V_{j,s}^{s^{\prime},(2)}(k)$ for $l < k$ and $V_{j,s}^{s^{\prime},(1)}(k) \geq V_{j,s}^{s^{\prime},(2)}(k)$ for $l \geq k$, where $V_{j,s}^{s^{\prime}}(k)=\sum_{l=1}^{k}\Delta_{t}v_{j,s}^{s^{\prime}}(l)$ represents the cumulative outflow with destination 
$s^{\prime}$ of link $(j,s)$ in the $k$th time interval.
Then the total travel time in link $(j,s)$ can be written  as follows:

\begin{eqnarray}
    TTT_{j,s} &=& \sum_{k=1}^{N} \Delta_{t} \sum_{l=1}^{k} \Delta_{t} u_{j,s}(l)-v_{j,s}(l)
    \nonumber \\
    &=& \sum_{k=1}^{N}\sum_{s_{0}} \Delta_{t}\Delta_{t} k\,u_{j,s}^{s_{0}^{\prime}}(k) - \Delta_{t}\Delta_{t} k\,v_{j,s}^{s_{0}^{\prime}}(k) 
    \nonumber \\
    &=& \sum_{k=1}^{N}\sum_{s_{0} \neq s} \left(\Delta_{t}\Delta_{t} k\,u_{j,s}^{s_{0}^{\prime}}(k) - \Delta_{t}\Delta_{t} k\,v_{j,s}^{s_{0}^{\prime}}(k) \right) + \sum_{k=1}^{N} \left(\Delta_{t}\Delta_{t} k\,u_{j,s}^{s^{\prime}}(k) - \Delta_{t}\Delta_{t} k\,v_{j,s}^{s^{\prime}}(k)\right) 
    \nonumber \\
    &=& \sum_{k=1}^{N}\sum_{s_{0} \neq s} \left(\Delta_{t}\Delta_{t} k\,u_{j,s}^{s_{0}^{\prime}}(k) - \Delta_{t}\Delta_{t} k\,v_{j,s}^{s_{0}^{\prime}}(k) \right) + \sum_{k=1}^{N} \Delta_{t}\Delta_{t} k\,u_{j,s}^{s^{\prime}}(k) - \sum_{k=1}^{N} \Delta_{t} V_{j,s}^{s^{\prime}}(k) 
    \nonumber
\end{eqnarray}

Therefore, we have proved that $u_{j,s}^{s_{0}^{\prime},(1)}(l)=u_{j,s}^{s_{0}^{\prime},(2)}(l)$, $v_{j,s}^{s_{0}^{\prime},(1)}(l)=v_{j,s}^{s_{0}^{\prime},(2)}(l)$ for $s_{0}\neq s$ and $1 \leq l \leq N$. For destination $s$, we have $u_{j,s}^{s^{\prime},(1)}(l)=u_{j,s}^{s^{\prime},(2)}(k)$ and $V_{j,s}^{s^{\prime},(1)}(l) \geq V_{j,s}^{s^{\prime},(2)}(l)$ for $1 \leq l \leq N$. Then we have

\begin{equation}
    TTT_{j,s}^{(1)} \leq TTT_{j,s}^{(2)}.
    \nonumber
\end{equation}

The total travel time consists of the total travel time in link $(j,s)$, total travel time in other links and waiting time at origins. The inflow and outflow variables in other links and origins are not affected by the change of headway from $h_{i,j}^{(1)}(k)$ to $h_{i,j}^{(2)}(k)$, so $TTT_{\textbf{h}_{1}} \leq TTT_{\textbf{h}_{2}}$ if $h_{j,s}^{(1)}(k) \leq h_{j,s}^{(2)}(k)$.
For link $(i,j)$, the node ${j}$ acts like the destination node, and we want to push more flow with destination $s$ to leave link $(i,j)$ and enter link $(j,s)$. The process is similar to what we have done for link $(j,s)$ by only changing one headway in link $(i,j)$ and keeping flow variables in other links unchanged. By repeating the above procedure for all the links and time intervals, we can show that TTT could be minimized under the minimum headway.

\section{Proof of Proposition \ref{prop:nonuni}}
\label{app:nonuniq}

Suppose $\textbf{x}^{*}$ and $\textbf{h}^{*}$ are the optimal solution of the proposed SO-DTA problem. If $\exists~ (i,j) \in E\setminus(L_R\cup L_S) $ and $1 \leq k \leq N$, $s.t. ~ \textit{v}_{i,j}^{f}{~} \rho_{i,j}^{*}(k)< \frac{1-\rho_{i,j}^{*}(k)\emph{L}}{\textit{h}_{i,j}^{*}(k)}$,  then $\textbf{h}^{*}$ and $\textbf{x}^{*}$ should satisfy the Equations \ref{eq:shockwave_dis1}, \ref{eq:shockwave_dis2} and \ref{eq:app_eq} as follows.

\begin{equation}
    \rho_{i,j}^{*}(k)  \leq  \frac{1}{\textit{h}_{i,j}^{*}(k)\textit{v}_{i,j}^f+\textit{L}} \label{eq:app_eq}
\end{equation}

If $h_{i,j;k}^{min} \leq \textit{h}_{i,j}^{*}(k) < \textit{h}_{i,j;k}^{max}$, we could always find an extreme small constant $c$ such that $c\textit{h}_{i,j}^{*}(k)$ also satisfy the Equations \ref{eq:shockwave_dis1}, \ref{eq:shockwave_dis2} and \ref{eq:app_eq} such that we find another optimal headway. If $\textit{h}_{i,j}^{*}(k)=\textit{h}_{i,j;k}^{max}$, by Proposition \ref{prop:minimum}, minimum headway setting achieves SO-DTA, so we show that SO-DTA can have at least two optimal headway that achieves the same TTT.




\section{Proof of Proposition \ref{prop:unique}}
\label{app:unique}

Suppose $~\textbf{x}^{*}$ with $TTT(\textbf{x}^{*})=TTT^{*}$ is unique, if $\hat{\textbf{h}}$ and $\bar{\textbf{h}}$ are different optimal solutions of Formulation \ref{eq:maximin}, then there must exist $(i_{1},j_{1}), (i_{2},j_{2}) \in E $ and $ 1\leq k_{1}, k_{2} \leq N$ such that $\hat{h}_{i_{1},j_{1}}(k_{1}) < \bar{h}_{i_{1},j_{1}}(k_{1})$ and $\hat{h}_{i_{2},j_{2}}(k_{2}) > \bar{h}_{i_{2},j_{2}}(k_{2})$. By the uniqueness of $~\textbf{x}^{*}$, we have $\textit{v}_{i_{1},j_{1}}^{f}{~} \rho_{i_{1},j_{1}}^{*}(k_{1}) < \frac{1-\rho_{i_{1},j_{1}}^{*}(k_{1})\emph{L}}{\textit{h}_{i_{1},j_{1}}^{*}(k_{1})}$ and $\textit{v}_{i_{2},j_{2}}^{f}{~} \rho_{i_{2},j_{2}}^{*}(k_{2}) < \frac{1-\rho_{i_{2},j_{2}}^{*}(k_{2})\emph{L}}{\textit{h}_{i_{2},j_{2}}^{*}(k_{2})}$.

Therefore, if we construct a new headway $\textbf{h}$ as follows:

\begin{eqnarray}
    \textit{h}_{i,j}(k) = \begin{cases}
\hat{h}_{i,j}(k), & (i,j)\neq (i_{1},j_{1})~or~k\neq k_{1}\\
\bar{h}_{i,j}(k), & (i,j) = (i_{1},j_{1})~and~k = k_{1}
\end{cases}
\end{eqnarray}

We note that $\textbf{h}$ is the optimal headway of the discretized SO-DTA formulation (Equations \ref{eq:obj_dis} to \ref{ineq:headway_dis}) as $\textbf{x}^{*} \in \Omega_{\textbf{h}}$. However, we also have $||~ \textbf{h}~ ||_{1} > ||~ \hat{\textbf{h}}~ ||_{1}$, and this contradicts to the assumption that $\hat{\textbf{h}}$ is the solution of Formulation \ref{eq:maximin}, which completes the proof.


\section{Settings of the small network}
  \label{app:parameter}
  For simplicity, the inflow capacity is set equal to the outflow capacity and the downstream queue capacity is set equal to the upstream queue capacity in each link. 
  The detailed configurations of each link are presented in Table~\ref{table1}.
\begin{table}[H]
	\centering
     \resizebox{1.02\textwidth}{!}{
	\begin{tabular}{c c c c c}  
		\hline  
		& & & &\\[-6pt] 
		Link & Flow capacity (vehicles/min) & Queue capacity (vehicles) & Free flow speed (km/min) & Link length (km)\\  
		\hline
		& & & &\\[-6pt]  
		1 $\to$ 3 & 50 & 600 & 0.9 & 1.6 \\
        \hline
        & & & &\\[-6pt] 
        1 $\to$ 4 & 45 & 600 & 1.0 & 1.2 \\
        \hline
        & & & &\\[-6pt] 
        2 $\to$ 3 & 45 & 600 & 1.2 & 3.6 \\
		\hline
        & & & &\\[-6pt] 
        2 $\to$ 4 & 50 & 600 & 1.1 & 3.3 \\
		\hline
        & & & &\\[-6pt]
        3 $\to$ 5 & 60 & 600 & 1.0 & 4.0 \\
		\hline
        & & & &\\[-6pt]
        4 $\to$ 5 & 60 & 600 & 1.0 & 3.0 \\
		\hline
	\end{tabular}
 }
 \caption{Parameters settings in the small network.}  
	\label{table1} 
\end{table}

The settings of minimum headway and maximum headway for each link are illustrated in Table \ref{table:minimum_headway} and Table \ref{table:maximum_headway}, respectively.

\begin{table}[htbp]
	\centering
    \scalebox{0.9}{
     \resizebox{0.85\textwidth}{!}{
	\begin{tabular}{c c c c c c c}  
		\hline   
		\diagbox{Time Interval}{Link} & 1 $\to$ 3 & 1 $\to$ 4 & 2 $\to$ 3 & 2 $\to$ 4 & 3 $\to$ 5 & 4 $\to$ 5 \\  
		\hline
		1  & 0.20s & 0.20s & 0.20s & 0.20s & 0.20s & 0.20s \\
        \hline
        2 & 0.20s & 0.20s & 0.20s & 0.20s & 0.20s & 0.20s \\
        \hline
        3 & 0.20s & 0.20s & 0.20s & 0.20s & 0.20s & 0.20s \\
		\hline
        4 & 0.75s & 0.60s & 0.45s & 0.55s & 0.50s & 0.50s \\
		\hline
        5 & 0.75s & 0.60s & 0.45s & 0.55s & 0.50s & 0.50s \\
		\hline
        6 & 0.75s & 0.60s & 0.45s & 0.55s & 0.50s & 0.50s \\
        \hline
        7 & 0.75s & 0.60s & 0.45s & 0.55s & 0.50s & 0.50s \\
        \hline
        8 & 1.05s & 1.15s & 0.85s & 0.75s &	0.80s & 0.90s \\
		\hline
        9 & 1.05s & 1.15s & 0.85s & 0.75s &	0.80s & 0.90s \\
        \hline
        10 & 1.05s & 1.15s & 0.85s & 0.75s & 0.80s & 0.90s \\
        \hline
        11 & 1.05s & 1.15s & 0.85s & 0.75s & 0.80s & 0.90s \\
        \hline
        12 & 1.25s & 1.35s & 1.10s & 0.95s & 1.00s & 1.20s \\
        \hline
        13 & 1.25s & 1.35s & 1.10s & 0.95s & 1.00s & 1.20s \\
        \hline
        14 & 1.25s & 1.35s & 1.10s & 0.95s & 1.00s & 1.20s \\
        \hline
        15 & 1.25s & 1.35s & 1.10s & 0.95s & 1.00s & 1.20s \\
        \hline
        16 & 1.55s & 1.65s & 1.30s & 1.25s & 1.25s & 1.30s \\
        \hline
        17 & 1.55s & 1.65s & 1.30s & 1.25s & 1.25s & 1.30s \\
        \hline
        18 & 1.55s & 1.65s & 1.30s & 1.25s & 1.25s & 1.30s \\
        \hline
	\end{tabular}
 }}
 \caption{Settings of minimum headway.}  
	\label{table:minimum_headway} 
\end{table}

\begin{table}[htbp]
	\centering
    \scalebox{0.9}{
     \resizebox{0.85\textwidth}{!}{
	\begin{tabular}{c c c c c c c}  
		\hline   
		\diagbox{Time Interval}{Link} & 1 $\to$ 3 & 1 $\to$ 4 & 2 $\to$ 3 & 2 $\to$ 4 & 3 $\to$ 5 & 4 $\to$ 5 \\  
		\hline
		1 & 1.80s & 1.80s & 1.80s & 1.80s & 1.80s & 1.80s \\
        \hline
        2 & 1.80s & 1.80s & 1.80s & 1.80s & 1.80s & 1.80s \\
        \hline
        3 & 1.80s & 1.80s & 1.80s & 1.80s & 1.80s & 1.80s \\
		\hline
        4 & 2.20s & 2.10s & 2.45s & 2.20s & 2.15s & 2.20s \\
		\hline
        5 & 2.20s & 2.10s & 2.45s & 2.20s & 2.15s & 2.20s \\
		\hline
        6 & 2.20s & 2.10s & 2.45s & 2.20s & 2.15s & 2.20s \\
        \hline
        7 & 2.20s & 2.10s & 2.45s & 2.20s & 2.15s & 2.20s \\
        \hline
        8 & 2.55s & 2.45s & 2.90s & 2.75s &	2.40s & 2.65s \\
		\hline
        9 & 2.55s & 2.45s & 2.90s & 2.75s &	2.40s & 2.65s \\
        \hline
        10 & 2.55s & 2.45s & 2.90s & 2.75s & 2.40s & 2.65s \\
        \hline
        11 & 2.55s & 2.45s & 2.90s & 2.75s & 2.40s & 2.65s \\
        \hline
        12 & 3.05s & 2.95s & 3.50s & 3.15s & 2.90s & 3.05s \\
        \hline
        13 & 3.05s & 2.95s & 3.50s & 3.15s & 2.90s & 3.05s \\
        \hline
        14 & 3.05s & 2.95s & 3.50s & 3.15s & 2.90s & 3.05s \\
        \hline
        15 & 3.05s & 2.95s & 3.50s & 3.15s & 2.90s & 3.05s \\
        \hline
        16 & 3.55s & 3.30s & 3.85s & 3.65s & 3.25s & 3.35s \\
        \hline
        17 & 3.55s & 3.30s & 3.85s & 3.65s & 3.25s & 3.35s \\
        \hline
        18 & 3.55s & 3.30s & 3.85s & 3.65s & 3.25s & 3.35s \\
        \hline
	\end{tabular}
 }}
 \caption{Results of the maximum headway.}  
	\label{table:maximum_headway} 
\end{table}

\section{Settings of the Hong Kong network}
\label{app:hk}

We introduce the geographical positioning of each origin-destination (O-D) pair and detail their mean demand profiles within the first hour, as showcased in Figure \ref{fig:demand}. The color saturation reflects the magnitude of the O-D demand, with a darker hue denoting a greater number of trips between the corresponding O-D pairs. Details of the link parameters are presented in Table \ref{table:hk_app}.

\begin{figure}[H]
    \centering
    \includegraphics[width=0.96\linewidth]{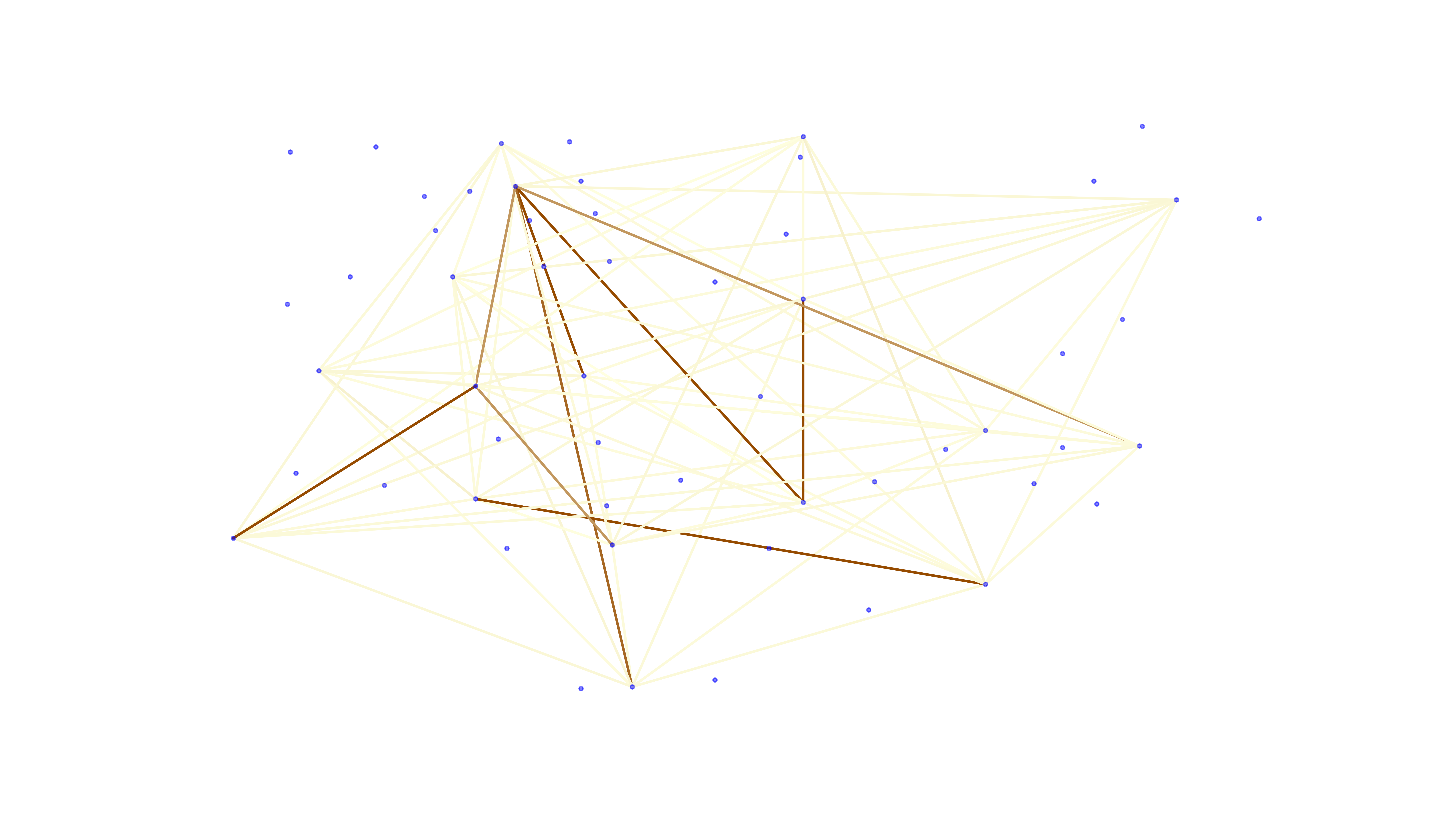}
    \caption{The hour O-D demand profile in the Hong Kong network. A darker color indicates more trips between the corresponding O-D pair.}
    \label{fig:demand}
\end{figure}

{\scriptsize 
{\begin{longtable}[H]{p{2.62cm} p{0.842cm} p{0.842cm} p{0.842cm} p{0.842cm} p{0.842cm} p{0.842cm} p{0.842cm} p{0.842cm} p{0.842cm} p{0.842cm} p{0.842cm}}
		\hline  
		\diagbox{Parameters}{Link} & 19 $\to$ 6 & 6 $\to$ 19 & 6 $\to$ 7 & 7 $\to$ 6 & 7 $\to$ 8 & 8 $\to$ 7 & 8 $\to$ 9 & 9 $\to$ 8 & 9 $\to$ 5 & 5 $\to$ 9 & 5 $\to$ 27  \\  
		\hline
		Flow capacity (veh/min) & 285 & 465 & 465 & 450 & 505 & 190  & 465 & 475 &195 & 465 & 505  \\
        \hline
        Queue capacity (veh) & 4200.0 & 1400.0 & 1300.0 & 2800.0 & 1583.0 & 1300.0 & 3600.0 &	1300.0 & 2600.0 & 3900.0 & 5200.0 \\
        \hline
        Free flow speed (km/min) & 0.232 & 0.232 & 0.37 & 0.463 & 0.231 &	0.232 &	0.463 &	0.37 & 0.463 & 0.232 & 0.232  \\
		\hline
        Link length (km) & 0.311 & 0.311 & 0.454 & 0.454 & 0.248 & 0.248 &	0.845 & 0.845 & 0.133 &	0.133 & 1.252 \\
        \hline
        \diagbox{Parameters}{Link} & 9 $\to$ 27 & 27 $\to$ 9 & 1 $\to$ 19 & 19 $\to$ 1 & 1 $\to$ 2 & 2 $\to$ 1 & 1 $\to$ 6 & 6 $\to$ 1 & 2 $\to$ 6 & 6 $\to$ 2 & 2 $\to$ 3  \\  
		\hline
		Flow capacity (veh/min) & 285 & 195 & 475 & 475 & 505 & 505 & 190 & 190 & 505 & 195 & 195  \\
        \hline
        Queue capacity (veh) & 3900.0&2600.0&	1400.0&	5600.0&	4200.0&	3166.0&	1583.0&	3166.0&	1583.0&	2800.0&	2800.0\\
        \hline
        Free flow speed (km/min) & 0.232&0.231&	0.370& 0.370&	0.232&	0.463&	0.231&	0.370 &	0.232&	0.463&	0.232  \\
		\hline
        Link length (km) & 1.228& 1.228& 0.563&	0.563&	0.168&	0.168&	0.366&	0.366&	0.446&	0.446&	0.168\\
		\hline
		\diagbox{Parameters}{Link} & 3 $\to$ 4 & 4 $\to$ 3 & 3 $\to$ 7 & 7 $\to$ 3 & 4 $\to$ 5 & 5 $\to$ 4 & 1 $\to$ 10 & 10 $\to$ 1 & 10 $\to$ 13 & 13 $\to$ 10 & 13 $\to$ 20   \\  
		\hline
		Flow capacity (veh/min) & 195 & 465 & 195 & 465 & 475 & 380 &	465	& 380	& 380	& 190	& 380	 \\
        \hline
        Queue capacity (veh) & 2800.0 & 1300.0 &  4200.0 &	3166.0 &	1583.0&	1583.0&	4200.0&	2800.0&	3166.0&	2600.0&	2800.0 \\
        \hline
        Free flow speed (km/min) & 0.232& 0.232& 0.463&	0.463&	0.463&	0.370&	0.463&	0.370&	0.232&	0.185&	0.463 \\
		\hline
        Link length (km) & 0.239&	0.239&	0.282&	0.282&	0.808&	0.808&	0.225&	0.225&	0.305&	0.305&	0.371 \\
		\hline
        \diagbox{Parameters}{Link} & 19 $\to$ 20 & 20 $\to$ 19 & 10 $\to$ 11 & 11 $\to$ 10 & 11 $\to$ 12 & 12 $\to$ 11 & 2 $\to$ 11 & 11 $\to$ 2 & 3 $\to$ 11 & 11 $\to$ 3 & 4 $\to$ 12   \\  
		\hline
		Flow capacity (veh/min) & 285&475&	505& 450& 380&	285&	450&	450&	190&	380&	475 \\
        \hline
        Queue capacity (veh) & 3166.0&	4200.0&	2600.0&	2250.0&	3166.0&	1583.0&	1583.0&	1583.0&	1583.0&	3166.0&	1583.0 \\
        \hline
        Free flow speed (km/min) & 0.232&	0.370&	0.370&	0.185&	0.463&	0.463&	0.231&	0.232&	0.463&	0.185&	0.370 \\
		\hline
        Link length (km) & 0.837& 0.837& 0.346&	0.346&	0.241&	0.241&	0.287&	0.287&	0.227&	0.227&	0.217 \\
		\hline
        \diagbox{Parameters}{Link} & 12 $\to$ 52 & 52 $\to$ 12 & 14 $\to$ 52 & 52 $\to$ 14 & 11 $\to$ 14 & 14 $\to$ 11 & 22 $\to$ 23 & 23 $\to$ 22 & 23 $\to$ 15 & 15 $\to$ 23 & 5 $\to$ 15   \\  
		\hline
		Flow capacity (veh/min) & 195&	465& 505& 195& 450&	465& 380&	285&	380&	380&	380 \\
        \hline
        Queue capacity (veh) & 1583.0&	4200.0&	1583.0&	2800.0&	4200.0&	1400.0&	2800.0&	4200.0&	1583.0&	1300.0&	3166.0 \\
        \hline
        Free flow speed (km/min) & 0.463& 0.232& 0.463&	0.463&	0.463&	0.463&	0.324&	0.232&	0.232&	0.370&	0.463 \\
		\hline
        Link length (km) & 0.314& 0.314& 0.239&	0.239&	0.303&	0.303&	0.338&	0.338&	0.403&	0.403&	0.501 \\
		\hline
        \diagbox{Parameters}{Link} & 18 $\to$ 21 & 21 $\to$ 18 & 16 $\to$ 17 & 17 $\to$ 16 & 16 $\to$ 42 & 42 $\to$ 16 & 31 $\to$ 32 & 32 $\to$ 31 & 31 $\to$ 25 & 25 $\to$ 31 & 20 $\to$ 41   \\  
		\hline
		Flow capacity (veh/min) & 190 & 380& 190&	475&	505&	475&	505&	190&	450&	450&	475 \\
        \hline
        Queue capacity (veh) & 4200.0&	3166.0&	1583.0&	1300.0&	1583.0&	6332.0&	6332.0&	1400.0&	3166.0&	1583.0&	1300.0 \\
        \hline
        Free flow speed (km/min) & 0.231&	0.370&	0.370&	0.232&	0.185&	0.463&	0.463&	0.232&	0.463&	0.463&	0.232 \\
		\hline
        Link length (km) & 0.361&	0.361&	0.397&	0.397&	0.575&	0.575&	0.382&	0.382&	0.341&	0.341&	0.288 \\
		\hline
        \diagbox{Parameters}{Link} & 41 $\to$ 42 & 42 $\to$ 41 & 42 $\to$ 43 & 43 $\to$ 42 & 43 $\to$ 45 & 45 $\to$ 43 & 43 $\to$ 44 & 44 $\to$ 43 & 43 $\to$ 41 & 41 $\to$ 43 & 31 $\to$ 36   \\  
		\hline
		Flow capacity (veh/min) & 475&	285&	195&	190&	505&	190&	380&	190&	285&	380&	190 \\
        \hline
        Queue capacity (veh) & 2800.0&	2800.0&	1400.0&	1583.0&	2800.0&	2600.0&	1583.0&	4200.0&	2800.0&	3166.0&	1300.0 \\
        \hline
        Free flow speed (km/min) & 0.463&	0.370&	0.463&	0.463&	0.324&	0.370&	0.463&	0.232&	0.370&	0.231&	0.231  \\
		\hline
        Link length (km) & 0.446&	0.446&	0.670&	0.670&	0.478&	0.478&	0.329&	0.329&	1.097&	1.097&	0.947 \\
		\hline
        \diagbox{Parameters}{Link} & 35 $\to$ 36 & 36 $\to$ 35 & 32 $\to$ 35 & 35 $\to$ 32 & 32 $\to$ 26 & 26 $\to$ 32 & 25 $\to$ 26 & 26 $\to$ 25 & 38 $\to$ 48 & 48 $\to$ 38 & 48 $\to$ 47   \\  
		\hline
		Flow capacity (veh/min) & 475&	380&	190&	285&	505&	380&	195&	190&	505&	380&	475 \\
        \hline
        Queue capacity (veh) & 1583.0&	1400.0&	4749.0&	1583.0&	3166.0&	1300.0&	2800.0&	1583.0&	3166.0&	1300.0&	3600.0 \\
        \hline
        Free flow speed (km/min) & 0.370&	0.463&	0.232&	0.463&	0.463&	0.463&	0.232&	0.370&	0.463&	0.232&	0.463 \\
		\hline
        Link length (km) & 0.186&	0.186&	0.922&	0.922&	0.256&	0.256&	0.476&	0.476&	0.323&	0.323&	0.323 \\
		\hline
        \diagbox{Parameters}{Link} & 27 $\to$ 38 & 38 $\to$ 27 & 15 $\to$ 27 & 27 $\to$ 15 & 13 $\to$ 14 & 14 $\to$ 13 & 13 $\to$ 16 & 16 $\to$ 13 & 39 $\to$ 40 & 40 $\to$ 39 & 14 $\to$ 17   \\  
		\hline
		Flow capacity (veh/min) & 450&	190& 450&	195&	465&	195&	190&	285&	195&	465&	190 \\
        \hline
        Queue capacity (veh) & 4200.0&	2600.0&	6332.0&	4749.0&	2800.0&	2800.0&	5600.0&	2600.0&	1583.0&	2800.0&	4200.0 \\
        \hline
        Free flow speed (km/min) & 0.231&	0.232&	0.231&	0.463&	0.324&	0.463&	0.232&	0.232&	0.463&	0.463&	0.232 \\
		\hline
        Link length (km) & 0.395&	0.395&	1.465&	1.465&	0.336&	0.336&	0.713&	0.713&	0.407&	0.407&	0.723 \\
		\hline
        \diagbox{Parameters}{Link} & 17 $\to$ 23 & 23 $\to$ 17 & 15 $\to$ 52 & 52 $\to$ 15 & 22 $\to$ 24 & 24 $\to$ 22 & 15 $\to$ 22 & 22 $\to$ 15 & 22 $\to$ 30 & 30 $\to$ 22 & 29 $\to$ 37   \\  
		\hline
		Flow capacity (veh/min) & 285&	450&	190&	285&	195&	190&	465&	190&	190&	465&	190 \\
        \hline
        Queue capacity (veh) & 2800.0&	5200.0&	2800.0&	4200.0&	2800.0&	3166.0&	1300.0&	2400.0&	4749.0&	2800.0&	1300.0 \\
        \hline
        Free flow speed (km/min) & 0.463&	0.463&	0.463&	0.463&	0.463&	0.463&	0.232&	0.463&	0.370&	0.463&	0.232 \\
		\hline
        Link length (km) & 0.772&	0.772&	0.663&	0.663&	0.650&	0.650&	0.425&	0.425&	1.078&	1.078&	0.290 \\
		\hline
        \diagbox{Parameters}{Link} & 24 $\to$ 29 & 29 $\to$ 24 & 18 $\to$ 51 & 51 $\to$ 18 & 37 $\to$ 51 & 51 $\to$ 37 & 34 $\to$ 46 & 46 $\to$ 34 & 34 $\to$ 35 & 35 $\to$ 34 & 33 $\to$ 34   \\  
		\hline
		Flow capacity (veh/min) & 380&	195&	450&	195&	450&	475&	380&	380&	380&	465&	505 \\
        \hline
        Queue capacity (veh) & 4200.0&	3166.0&	2800.0&	4200.0&	1583.0&	4200.0&	3900.0&	2600.0&	4749.0&	2800.0&	1400.0 \\
        \hline
        Free flow speed (km/min) & 0.463&	0.370&	0.463&	0.370&	0.232&	0.232&	0.463&	0.463&	0.185&	0.324&	0.463 \\
		\hline
        Link length (km) & 0.690&	0.690&	0.386&	0.386&	0.466&	0.466&	0.718&	0.718&	0.302&	0.302&	0.875 \\
		\hline
        \diagbox{Parameters}{Link} & 32 $\to$ 33 & 33 $\to$ 32 & 40 $\to$ 53 & 53 $\to$ 40 & 30 $\to$ 49 & 49 $\to$ 30 & 38 $\to$ 28 & 28 $\to$ 38 & 28 $\to$ 30 & 30 $\to$ 28 & 28 $\to$ 54   \\  
		\hline
		Flow capacity (veh/min) & 285&	285&	285&	285&	190&	380&	450&	195&	190&	475&	505 \\
        \hline
        Queue capacity (veh) & 2600.0&	3166.0&	2800.0&	4200.0&	4200.0&	3600.0&	2800.0&	2800.0&	4749.0&	1583.0&	1300.0 \\
        \hline
        Free flow speed (km/min) & 0.232&	0.463&	0.463&	0.232&	0.232&	0.232&	0.463&	0.463&	0.463&	0.463&	0.463 \\
		\hline
        Link length (km) & 0.567&	0.567&	0.658&	0.658&	0.299&	0.299&	1.124&	1.124&	0.571&	0.571&	0.310 \\
		\hline
        \diagbox{Parameters}{Link} & 54 $\to$ 47 & 47 $\to$ 54 & 28 $\to$ 48 & 48 $\to$ 28 & 39 $\to$ 54 & 54 $\to$ 39 & 46 $\to$ 53 & 53 $\to$ 46 & 46 $\to$ 29 & 29 $\to$ 46 & 29 $\to$ 53   \\  
		\hline
		Flow capacity (veh/min) & 285&	195&	505&	285&	380&	475&	195&	380&	465&	465&	380 \\
        \hline
        Queue capacity (veh) & 1583.0&	1300.0&	1583.0&	1583.0&	3166.0&	1583.0&	5600.0&	2800.0&	1300.0&	2600.0&	2800.0 \\
        \hline
        Free flow speed (km/min) & 0.324&	0.463&	0.463&	0.185&	0.232&	0.370&	0.232&	0.463&	0.463&	0.370&	0.232 \\
		\hline
        Link length (km) & 0.819&	0.819&	1.078&	1.078&	0.822&	0.822&	0.454&	0.454&	0.831&	0.831&	0.776 \\
		\hline
        \diagbox{Parameters}{Link} & 29 $\to$ 55 & 55 $\to$ 29 & 30 $\to$ 55 & 55 $\to$ 30 & 50 $\to$ 55 & 55 $\to$ 50 & 40 $\to$ 50 & 50 $\to$ 40 & 33 $\to$ 37 & 37 $\to$ 33 & 33 $\to$ 46   \\  
		\hline
		Flow capacity (veh/min) & 190&	195&	450&	190&	195&	380&	380&	450&	380&	465&	380 \\
        \hline
        Queue capacity (veh) & 1300.0&	2600.0&	1300.0&	5600.0&	2600.0&	4200.0&	1300.0&	4749.0&	4200.0&	1300.0&	1400.0 \\
        \hline
        Free flow speed (km/min) & 0.370&	0.231&	0.231&	0.463&	0.463&	0.463&	0.231&	0.231&	0.463&	0.370&	0.324 \\
		\hline
        Link length (km) & 0.333&	0.333&	0.189&	0.189&	0.389&	0.389&	0.263&	0.263&	0.324&	0.324&	0.538 \\
		\hline
        \diagbox{Parameters}{Link} & 21 $\to$ 25 & 25 $\to$ 21 & 18 $\to$ 26 & 26 $\to$ 18 & 25 $\to$ 44 & 44 $\to$ 25 & 42 $\to$ 44 & 44 $\to$ 42 & 31 $\to$ 45 & 45 $\to$ 31 & 49 $\to$ 50   \\  
		\hline
		Flow capacity (veh/min) & 190&	450&	475&	475&	285&	380&	195&	465&	475&	190&	285 \\
        \hline
        Queue capacity (veh) & 4200.0&	1583.0&	1583.0&	2800.0&	1583.0&	5600.0&	3166.0&	1583.0&	2800.0&	5200.0&	1583.0 \\
        \hline
        Free flow speed (km/min) & 0.232&	0.370&	0.463&	0.324&	0.463&	0.324&	0.370&	0.370&	0.463&	0.463&	0.232 \\
		\hline
        Link length (km) & 0.396&	0.396&	0.411&	0.411&	0.341&	0.341&	0.779&	0.779&	0.991&	0.991&	0.254 \\
		\hline
        \diagbox{Parameters}{Link} & 26 $\to$ 51 & 51 $\to$ 26 & 17 $\to$ 18 & 18 $\to$ 17 & 16 $\to$ 21 & 21 $\to$ 16 & 24 $\to$ 37 & 37 $\to$ 24 & 39 $\to$ 49 & 49 $\to$ 39 & 49 $\to$ 28   \\  
		\hline
		Flow capacity (veh/min) & 505&	285&	285&	450&	505&	475&	190&	380&	505&	380&	505 \\
        \hline
        Queue capacity (veh) & 1300.0&	4200.0&	5200.0&	4200.0&	2600.0&	4200.0&	1583.0&	1583.0&	2800.0&	2800.0&	1300.0 \\
        \hline
        Free flow speed (km/min) & 0.463&	0.231&	0.463&	0.370&	0.370&	0.370&	0.370&	0.370&	0.370&	0.231&	0.232 \\
		\hline
        Link length (km) & 0.315&	0.315&	0.435&	0.435&	0.353&	0.353&	0.704&	0.704&	0.278&	0.278&	0.609 \\
		\hline
        \diagbox{Parameters}{Link} & 27 $\to$ 5 & 3 $\to$ 2 & 20 $\to$ 13 & 12 $\to$ 4 & 15 $\to$ 5 & 41 $\to$ 20 & 36 $\to$ 31 & 47 $\to$ 48 & 17 $\to$ 14 
        & 37 $\to$ 29 & 34 $\to$ 33  \\  
		\hline
		Flow capacity (veh/min) & 475&	195 &	285 &	380 &	450 & 285 & 475	& 380 & 475
        & 475 & 285\\
        \hline
        Queue capacity (veh) & 1583.0 &	3166.0&	2600.0&	1583.0&	2600.0&	3166.0 & 1583.0
        & 3166.0 & 2600.0 & 1583.0 & 2800.0\\
        \hline
        Free flow speed (km/min) & 0.370 &	0.370 &	0.232 &	0.370&	0.232&	0.370& 0.232 
        & 0.463 & 0.232 & 0.370 & 0.463\\
		\hline
        Link length (km) & 1.252& 0.168 & 0.371 & 0.217& 0.501 & 0.288 & 0.947 & 0.323 & 0.723  & 0.290 & 0.875 \\
		\hline  
       \diagbox{Parameters}{Link} & 54 $\to$ 28 & 53 $\to$ 29 & 46 $\to$ 33 & 50 $\to$ 49
       & 28 $\to$ 49\\  
		\hline
		Flow capacity (veh/min) & 505 & 380 & 195 & 505 & 195 & & & & & & \\
        \hline
        Queue capacity (veh) & 2800.0 & 1400.0 & 2800.0 & 5600.0 & 1583.0 & & & & & &\\
        \hline
        Free flow speed (km/min) & 0.232 & 0.232 & 0.370 & 0.463 & 0.185 & & & & & &\\
		\hline
        Link length (km) & 0.310 & 0.776 & 0.538 & 0.254 & 0.609 & & & & & &\\
		\hline
   \caption{Link parameters in the Hong Kong network.}  
	\label{table:hk_app} 
\end{longtable}}

\section{Online algorithm}
\label{app:online}

\normalsize In this section, we propose an online algorithm to solve the SO-DTA problem in real-time. 
Suppose $\textbf{x}_{i,j}$ and $\textbf{h}_{i,j}$ denote the flow dynamics and headway variables from the $i$th to $j$th time intervals, respectively. Then we have $\textbf{x}_{i,j}=\mathop{\cup}\limits_{k=i}^{j}\textbf{x(k)}$ and $\textbf{h}_{i,j}=\mathop{\cup}\limits_{k=i}^{j}\textbf{h(k)}$, where $\textbf{x(k)}$ and $\textbf{h(k)}$ are the flow dynamics and headway variables in the $k$th time interval. Besides, $TTT_{k}$ denotes the total travel time in the $k$th time interval. 

In a real-time headway control framework, the future O-D demand is unclear at the current time. In order to obtain an effective real-time headway, it is necessary to estimate the future O-D demand. Besides, many studies are focusing on the estimation of O-D demand \citep{zhang2021short,ke2021predicting,bhattacharjee2001modeling,xiong2020dynamic}. And we use  $\textbf{d}_{i,j}$ and $\widetilde{\textbf{d}}_{i,j}$ to represent the real and estimated demand from the $i$th to $j$th time intervals, respectively.

Then the structure of the online algorithm is illustrated in the Algorithm \ref{algo:online}. $t_{0}$ and $t_{n}$ are the beginning and end of the time slices of headway control in the real world, such as 24 hours or 7 days. We divide the total time of controlling into many time slices. In the current time slice, we would predict the demand in the next time slice. Then we obtain the real-time headway control of all time intervals in the current time slice. 

\begin{breakablealgorithm}
\caption{Online version of Algorithm 1}
{\textbf{Input:} parameters in Table \ref{notation_con}};  $t_{1},\cdots,t_{n}$\qquad\qquad\qquad\qquad\qquad\qquad\qquad~~~~
\begin{algorithmic}[1]
\FOR{$i=1$ \TO $n$}

\STATE 1. Predict the future demand $\widetilde{\textbf{d}}_{t_{i-1},t_{i}}$ from the $i-1$ to $i$th time slices. Then we solve the linear programming in Equation \ref{eq:online_1} under minimum headway and obtain the $\textbf{x}_{t_{i},t_{i+1}}^{*}$.
\begin{eqnarray}
     &\min \limits_{\textbf{x}_{t_{i},t_{i+1}}} ~~  \mathop{\sum}\limits_{k=t_{i}}^{t_{i+1}}TTT_{k}
     \label{eq:online_1}\\
    & \, s.t.  \quad \textbf{A}(\textbf{h}^{min},\widetilde{\textbf{d}}_{t_{i},t_{i+1}})\textbf{x}_{t_{i},t_{i+1}}  =  \textbf{B}(\textbf{h}^{min},\widetilde{\textbf{d}}_{t_{i},t_{i+1}}) \nonumber \\
    & \quad \quad ~~ \textbf{C}(\textbf{h}^{min},\widetilde{\textbf{d}}_{t_{i},t_{i+1}})\textbf{x}_{t_{i},t_{i+1}} \leq \textbf{D}(\textbf{h}^{min},\widetilde{\textbf{d}}_{t_{i},t_{i+1}}) \nonumber
\end{eqnarray}

\FOR{$k=t_{i-1}$ \TO $t_{i}$}
\STATE 2. Obtain the real-time maximin headway control $\textbf{h}^{*}(k)$ in the $k$th time interval by solving the following linear programming in Equation \ref{eq:online_2}.
\begin{eqnarray}
\max \limits_{\mathbf{h}(k)} ~~  \mathbf{1}^T~\textbf{h}(k) \qquad\qquad \quad &\,&
\label{eq:online_2}\\
s.t. \quad  \rho_{i,j}^{*}(k)\textit{v}_{i,j}^f\textit{h}_{i,j}(k) & \leq&  1-\rho_{i,j}^{*}(k)\textit{L}  \qquad i,j,k \in \left\{(i,j;k)~|~ \textit{v}_{i,j}^{f}{~} \rho_{i,j}^{*}(k) \leq \frac{1-\rho_{i,j}^{*}(k)\emph{L}}{\textit{h}_{i,j}^{*}(k)} \right\} \nonumber \\
 L_{i,j}\textit{h}_{i,j}(k) & \geq&  \Delta_{t}\textit{L}\textit{n}_{i,j}^{w,*}(k)
\qquad i,j,k \in \left\{(i,j;k)~|~ \textit{v}_{i,j}^{f}{~} \rho_{i,j}^{*}(k) \leq \frac{1-\rho_{i,j}^{*}(k)\emph{L}}{\textit{h}_{i,j}^{*}(k)} \right\} \nonumber \\
 L_{i,j}\textit{h}_{i,j}(k) & \leq&  \Delta_{t}\textit{L}\textit{n}_{i,j}^{w,*}(k)+\Delta_{t}\textit{L} \qquad i,j,k \in \left\{(i,j;k)~|~ \textit{v}_{i,j}^{f}{~} \rho_{i,j}^{*}(k) \leq \frac{1-\rho_{i,j}^{*}(k)\emph{L}}{\textit{h}_{i,j}^{*}(k)} \right\} \nonumber \\
 \textit{h}_{i,j}^{*}(k) \leq \textit{h}_{i,j}(k) & \leq&  \textit{h}_{i,j;k}^{max} 
\qquad i,j,k \in \left\{(i,j;k)~|~ \textit{v}_{i,j}^{f}{~} \rho_{i,j}^{*}(k) \leq \frac{1-\rho_{i,j}^{*}(k)\emph{L}}{\textit{h}_{i,j}^{*}(k)} \right\} \nonumber \\
 \textit{h}_{i,j}(k) & =&  \textit{h}_{i,j}^{*}(k) \qquad i,j,k \in \left\{(i,j;k)~|~ \textit{v}_{i,j}^{f}{~} \rho_{i,j}^{*}(k) \geq \frac{1-\rho_{i,j}^{*}(k)\emph{L}}{\textit{h}_{i,j}^{*}(k)} \right\} \nonumber
\end{eqnarray}

\STATE 3. Obtain the total travel time $TTT_{k}$ in the $k$ time interval under the headway control $\mathbf{h}^{*}(k)$ and real demand $\textbf{d}(k)$. 

\ENDFOR

\ENDFOR
\end{algorithmic}
{\textbf{Output:} maximin headway control $\textbf{h}^{*}=\mathop{\cup}\limits_{i=1}^{n}\mathop{\cup}\limits_{k=t{i-1}}^{t_{i}}\textbf{h}^{*}(k)$ and total travel time $TTT=\mathop{\sum}\limits_{i=1}^{n}\mathop{\sum}\limits_{k=t_{i-1}}^{t_{i}}TTT_{k}$.\qquad\qquad~~~}
\label{algo:online}
\end{breakablealgorithm}

\end{document}